\documentclass[a4paper,11pt]{article}
\usepackage[usenames]{color} 
\usepackage{mathtools}
\usepackage{gensymb}
\usepackage{enumitem}
\usepackage{algorithm}
\usepackage{amsmath,amsfonts}
\usepackage{algpseudocode}
\usepackage{mathrsfs}
\usepackage[a4paper, margin={2.5cm}]{geometry}
\usepackage{pdfpages}
\usepackage{bbold}
\usepackage{placeins}
\usepackage{graphicx}
\usepackage{lipsum}                     
\usepackage{xargs}    
\usepackage{subcaption}
\usepackage{algorithm}
\usepackage{algorithmicx}
\usepackage{braket}    
\usepackage{amsthm}
\usepackage{adjustbox}
\usepackage{shuffle}
\usepackage[bbgreekl]{mathbbol}
\usepackage{amssymb}
\usepackage[square,numbers]{natbib}

\RequirePackage[colorlinks,citecolor=blue,urlcolor=blue, linkcolor=black]{hyperref}

\theoremstyle{definition}
\newtheorem{definition}{Definition}[section]   
\newtheorem{example}[definition]{Example}

\theoremstyle{plain}
\newtheorem{theorem}[definition]{Theorem}
\newtheorem{prop}[definition]{Proposition}
\newtheorem{lemma}[definition]{Lemma}
\newtheorem{cor}[definition]{Corollary}

\theoremstyle{remark}
\newtheorem{remark}[definition]{Remark}
\newtheorem{notation}[definition]{Notation}

\usepackage[pdftex,dvipsnames]{xcolor}
\usepackage[colorinlistoftodos,prependcaption,textsize=tiny]{todonotes}

\numberwithin{equation}{section}

\newcommand{\dd}{\textrm{d}}

\newcommandx{\unsure}[2][1=]{\todo[linecolor=red,backgroundcolor=red!25,bordercolor=red,#1]{#2}}
\newcommandx{\change}[2][1=]{\todo[linecolor=blue,backgroundcolor=blue!25,bordercolor=blue,#1]{#2}}
\newcommandx{\info}[2][1=]{\todo[linecolor=OliveGreen,backgroundcolor=OliveGreen!25,bordercolor=OliveGreen,#1]{#2}}
\newcommandx{\improvement}[2][1=]{\todo[linecolor=Plum,backgroundcolor=Plum!25,bordercolor=Plum,#1]{#2}}
\newcommandx{\thiswillnotshow}[2][1=]{\todo[disable,#1]{#2}}

\newcommandx{\Tr}{\textrm{Tr}}
\newcommandx{\functionallycontrolled}{functionally controlled }
\newcommandx{\FunctionallyControlled}{Functionally Controlled }
\newcommandx{\Functionallycontrolled}{Functionally controlled }

\newcommand{\controllingfunctions}{controlling functions }
\newcommand{\controllingfunctionsnospace}{controlling functions}

\title{{Signature Methods in Stochastic Portfolio Theory}}

\author{Christa Cuchiero\thanks{
Vienna University, Department of Statistics and Operations Research, Data Science Uni Vienna, Kolingasse 14-16 1, A-1090 Wien, Austria, christa.cuchiero@univie.ac.at}
\and Janka M\"oller\thanks{Vienna University, Department of Statistics and Operations Research, Kolingasse 14-16, A-1090 Wien, Austria, janka.moeller@univie.ac.at
\newline
The authors gratefully acknowledge financial support through grant Y 1235 of the Austrian Science Fund.  }
}
\date{\today}

\begin{document}
	\maketitle
	
	\begin{abstract}

In the context of stochastic portfolio theory we introduce a novel class of portfolios which we call linear path-functional portfolios. These are portfolios which are determined by certain transformations of linear functions of a collections of feature maps that are non-anticipative path functionals of an underlying semimartingale. As main example for such feature maps we consider the signature of the (ranked) market weights. We prove that these portfolios are universal in the sense that every continuous, possibly path-dependent, portfolio function of the market weights can be uniformly approximated by signature portfolios. We also show that signature portfolios can approximate the growth-optimal portfolio in several classes of non-Markovian market models arbitrarily well and  illustrate numerically that the trained signature portfolios are remarkably close to the theoretical growth-optimal portfolios. 
Besides these universality features, the main numerical advantage lies in the fact that several optimization tasks like maximizing (expected) logarithmic wealth or mean-variance optimization within the class of linear path-functional portfolios reduce to a convex quadratic optimization problem, thus making it computationally highly tractable.
We apply our method also to real market data based on several indices. Our results point towards out-performance on the considered out-of-sample data, also in the presence of transaction costs.

\end{abstract}


	\tableofcontents
	
	\section{Introduction}

    Portfolio theory and portfolio optimization are a core topic in mathematical finance and have been the subject of vivid research for decades. The most prominent example is of course the mean-variance optimization problem as formulated by  Harry Markowitz, the founder of \emph{Modern Portfolio Theory} (see \cite{markovitz1959portfolio}). His framework takes into account risk-preferences of investors while offering a tractable form of optimization, namely being convex and quadratic. Although the choice of constant portfolio weights is arguably oversimplified, the elegance of the Markowitz model can not be denied. Ever since, researchers have been striving for more and more realistic models with fewer limiting assumptions while hoping to preserve tractability. 
    
    In the spirit of coming up with more realistic assumptions, for instance  without the need of claiming a specific form of the hardly measurable drift, Robert Fernholz developed \emph{Stochastic Portfolio Theory} (SPT) (see~\cite{Fern02} and also~\cite{Fernholz_Karatzas}). 
    A key feature of SPT, again
    in the realm of more realistic modeling assumptions,  is 
    the rather relaxed no-arbitrage condition. Indeed,
     the price process is only assumed to be a (continuous) semimartingale satisfying 
    \emph{No Unbounded Profit with Bounded Risk} and not necessarily the stronger  condition of 
     \emph{No free lunch with vanishing risk} (see \cite{delbaen1994general}).
     It is also central that the portfolios' performance is measured with respect to a benchmark portfolio, which is usually chosen to be the market portfolio, i.e.~large indices like the S\&P500. At least since the introduction of very liquid ETFs replicating market indices the out-performance of the market is indeed a major challenge for investors. In this context Fernholz introduced \emph{functionally generated portfolios} and derived the so-called \emph{master formula} describing the relative wealth process of such portfolios and allowing for the detection for relative arbitrages with respect to the market portfolio. These portfolios are constructed via the log-gradient of a function of the market weights. While appreciating the beauty of the framework and also its robustness in view of  certain optimal worst case long-run growth rates
     (see \cite{kardaras2021ergodic, itkin2022ergodic}), one may point out that the log-gradient form  as well as the neglect of  past information is somewhat restrictive.
    To reduce these limitations  various generalizations of the original framework have been developed, e.g., towards  replacing the market weights by market-to-book ratios~\cite{Kim_22}, considering trading strategies generated by Lyapunov functions~\cite{Karatzas_Ruf}, adding extra information in the form of a finite variation process~\cite{Ruf_Xie_19} and incorporating past information in a path-dependent setting~\cite{Schied_18, Karatzas_Kim_20}. In this context it is worth mentioning that the latter two papers do not work in a stochastic setting but rather in a path-wise one, using functional It\^o-calculus as in~\cite{Cont_2010, Dupire_2019} based on F\"ollmer integration. An alternative path-wise formulation allowing for more general strategies (beyond gradient type) and building on the theory of rough path has been studied in~\cite{Promel_1}. \\
    
    In this paper, we consider for simplicity the stochastic setting of continuous semimartingales (even though a path-wise formulation in the spirit of~\cite{Promel_1} would also be possible), take inspiration from the functionally generated portfolios and generalize them to what we call \emph{path-functional portfolios}. These are portfolios constructed via an auxiliary portfolio $\tau$ and non-anticipative path-functionals with a general semimartingale as  input. In doing so, we get rid of any gradient-type form of the functionals and allow for a general benchmark portfolio $\tau$ and past information of the market or other exogenous relevant factors in form of a general continuous semimartingale. The choice of continuous semimartingales is here only for the ease of exposition, but could be replaced by other (predictable stochastic) processes. {This generalization can be achieved by using rough path theory, where similar universal approximation results hold for more general processes, see e.g.~\cite{Cuchiero_Primavera_Svaluto-Ferro2022} for rough paths with jumps.}
    
    One of our main interest concerns \emph{tractable portfolio optimization} in the framework of SPT using such path-functional portfolios. To do so, we focus on \emph{linear path-functional portfolios}, which are constructed via a linear combination of possibly path-dependent feature maps and constant optimization parameters. 
    
    This allows us to incorporate modern machine learning techniques for time-series data, in particular \emph{signature methods} (see e.g.~\cite{Kalsi,Bayer20, Blanka_23} in the area of optimal control,~\cite{Buhler,NSSBWL:21} on applications to data generation/simulation,~\cite{Lyons2019,Cuchiero_Svaluto-Ferro_Gazzani, Cuchiero_Primavera_Svaluto-Ferro2022} for topics in financial modelling,~\cite{Compagnoni_2022,Cuchiero_Teichmann,cuchiero2021expressive, Gambara_2023} on randomization techniques and their application,~\cite{levin13, kiraly19,chevyrev22,andres2022signature,cuchiero2023joint, Gambara_2022} for other applications in mathematical finance and machine learning), but also other tools, like random neural networks~\cite{Herrera_Krach, Huang2006, gallicchio2021,CAO2018278}.
    In the current paper, our focus lies on signature methods, which play an important role in rough path theory (see e.g.~\cite{Lyons1998, Friz_Victoir}) and whose appeal  stems from the \emph{(global) universal approximation theorem.} It states that linear functions on the signature can approximate continuous (with respect to certain variation distances) path-functionals arbitrarily well on compact sets of paths or even globally when using the setup of weighted spaces (see \cite{Cuchiero_Schmocker_Teichmann_2023}). This result then motivates the notion of \emph{signature portfolios,} which are  linear path-functional portfolios with feature maps
    being elements of the signature.
    Indeed, with this choice we can approximate generic path-functional portfolios including functionally generated portfolios and the growth-optimal portfolio in a large class of non-Markovian markets, arbitrary well.
    
    Despite this versatility, signature portfolios and linear path-functional portfolios in general lead to remarkably tractable optimization tasks. More concretely, we show that 
    maximizing the (expected) logarithmic wealth as well as the mean-variance optimization of portfolio returns lead to \emph{convex quadratic optimization problems.} We would like to highlight that the Markowitz portfolio itself is included in the class of signature portfolios and moreover, that any chosen benchmark portfolio $\tau$ can always be \emph{attained} by signature portfolios constructed via $\tau$. 
    
    In addition to our theoretical results, we optimize signature portfolios in numerical experiments using simulated and real market data.
    By means of simulated data generated from the Black-Scholes model, volatility stabilized models \cite{Karatzas_Vol_Stab} and so-called signature market models, we
    illustrate that the trained signature portfolios (of small degree) are  remarkably
close to the theoretical growth-optimal portfolios.  

In the real market situation   we additionally incorporate transaction costs in our optimization by proposing a regularization  under which the optimization problems remain convex and quadratic. To deal with high-dimensional market indices such as the NASDAQ and S\&P500, we use dimension reduction techniques leading to what we call randomized signature portfolios and JL-signature portfolios, where the former rely on randomized signatures (see~\cite{Gambara_2022, Compagnoni_2022, Cuchiero_Teichmann, cuchiero2021expressive} and also \cite{cirone2023neural} for a connection to neural signature kernels), while the latter are based on an explicit Johnson-Lindenstrauss projection. Note that we do not reduce the dimension of the underlying process, but only of its signature. To make the computation of the JL-signature portfolios feasible, we present a memory efficient algorithm to do so. We train the portfolios in both a log-relative wealth and mean-variance optimization in the NASDAQ, the SMI and  S\&P500 markets.  We use randomized- and JL-signature portfolios in the high-dimensional markets (NASDAQ and S\&P500) and signature portfolios for the SMI.
In the NASDAQ market the portfolios are  computed from the 
 ranked market weights, while in the other two cases the name-based market weights are used.
    We incorporate proportional transaction costs of 1\% and 5\% respectively in the name-based markets. 
    In almost all configurations, our trained portfolios outperform the market portfolios during the out-of-sample period, even under 5\% of proportional transaction costs. 
    It is worth noting that in the setting with transaction costs we still re-balance our portfolios daily whereas the market portfolio \emph{does not} pay any transaction costs. \\

    To provide more context to our findings, we highlight two papers which are particularly related to our research, namely~\cite{Wong} and~\cite{Blanka_23}. 
    In~\cite{Wong}, the authors study portfolio optimization in SPT over 
    a family of rank-based functionally generated portfolios parameterized by an exponentially concave function. Their objective is to maximize the 
    relative logarithmic growth rate, which in their parameterization leads to a convex optimization problem.  In their empirical study
they use (ranked) market data of the 100 largest US stocks, which they manage to out-perform during the out-of-sample period with their trained portfolios. As they do not include transaction costs in the learning procedure (which however could be done by adopting e.g. similar transaction cost treatments as ours), the out-performance does no longer work with 0.45\% of proportional transaction costs.
This could also be related to the fact that the considered portfolios are long only functionally generated and thus a rather small class. Indeed,  the class of portfolios we consider is more general in several respects: we allow for short-selling, the inclusion of information from the past and from exogenous signals as well as for a general benchmark portfolio $\tau$ in the construction of the portfolios. At the same time  the portfolios of~\cite{Wong} can be approximated
arbitrarily well using signature portfolios, which is a result of the  universal approximation theorem. 

    Very recently, the work of~\cite{Blanka_23} appeared, where the authors study mean-variance optimization of the wealth process using signature methods. In contrast to our multiplicative setting, they consider an additive approach where the trading strategies correspond to numbers of shares, which are then directly 
    parameterized  as linear functions on the signature. This means in particular that the strategies are only self-financing if a bank-account is included, while in our case the strategies are always self-financing \emph{without} a bank-account, since we use portfolio weights. Even though the formulation of ~\cite{Blanka_23}, where the \emph{wealth process} itself is the quantity of interest,
    is the standard one in the literature of the continuous-time mean-variance portfolio selection problem (see e.g.~\cite{zhou2000continuous}), it differs from the original approach of Markowitz \cite{markovitz1959portfolio}, where mean-variance optimization of the \emph{returns} was actually the objective. The latter is  advantageous when optimizing (in a modelfree setup) along the observed trajectory since the returns are more likely to be stationary than the assets themselves. For these reasons we use returns and a multiplicative approach, which allows in particular to include the \emph{classical Markowitz portfolio} (as a very special case) in our setup (which is not possible in \cite{Blanka_23}).
    Another difference is that we consider general \emph{linear} path-functionals and not only signature portfolios, while still obtaining a convex quadratic optimization problem.
Additionally we include transaction costs and dimension reduction techniques to be able to deal with 500-dimensional price processes, whereas in the numerical implementations of~\cite{Blanka_23} only  two-dimensional price  processes are considered.  Let us also 
point out that in~\cite{Blanka_23} the signal  is augmented with the two-dimensional lead-lag process to express It\^o-integrals as linear functions on the corresponding signature. We usually only augment with time, which -- under certain conditions on the quadratic variation -- still allows to get linear expressions for the It\^o-integrals (see Remark \ref{rem:simplification}).
 As the computation of the signature becomes expensive for higher dimensions, this might be of relevance. \\

The remainder of the article is structured as follows. In Section \ref{sec:Foundations} we recall the notion of the signature of continuous semimartingales, provide two versions of the universal approximation theorem, state the Johnson-Lindenstrauss Lemma and specify the financial market setting as well as classes of important portfolios.
Section~\ref{sec:pathfunctionalportfolios} is dedicated to the introduction of (linear) path-functional portfolios with signature portfolios as special case, while Section \ref{sec:approx_sig_port}
addresses their
approximation properties, in particular in view of approximating the growth optimal portfolio. In Section \ref{sec:optimizing_signature_port} the convex quadratic optimization tasks for the mean-variance and (expected) log-(relative) wealth are discussed and Section \ref{sec:transcost} explains the treatment of proportional transaction costs. Section \ref{sec:numerical_results}
	concludes the paper with  a presentation of the numerical results. Some proofs are gathered in the Appendix.
 
	\section{Preliminaries}\label{sec:Foundations}
	
	\subsection{The Signature of a Continuous Semimartingale}\label{subsec:Intro_Signature}

    In this section, we first introduce the signature of an $\mathbb{R}^n$-valued continuous semimartingale $X$ and then state the so-called \emph{universal approximation theorem} (see Theorem~\ref{thm:UAT_sig}),  which motivates the use of signature methods. In the following, we denote by $X$ an $\mathbb{R}^n$-valued continuous semimartingale defined on some filtered probability space $(\Omega, \mathcal{F}, (\mathcal{F}_t)_{t\in[0,T]})$ with $T>0$ finite. We start by introducing several preliminary notions needed to define  the signature. Similar introductions can for instance be found in~\cite{Lyons_2007, Friz_Victoir}. Readers who are neither interested in a mathematical introduction of the signature nor the associated approximation theorems of linear functions of the signature in general or signature portfolios in particular, may choose to skip Sections~\ref{subsec:Intro_Signature},~\ref{subsec:UAT},~\ref{subsec:GUAT} and~\ref{sec:approx_sig_port}.

    \begin{definition}[(Truncated) Extended Tensor Algebra and Group]\label{def:tensor_alg}
        The signature of $X$, denoted by $\mathbb{X}$, is an element of the extended tensor algebra $T((\mathbb{R}^n))$ defined as 
	    \begin{equation*}
		    T((\mathbb{R}^n)) := \left\{ (a_0, a_1, ...) | a_k \in (\mathbb{R}^n)^{\otimes k}, k\geq 0 \right\} 
	    \end{equation*}
	    with $(\mathbb{R}^n)^{\otimes 0}:= \mathbb{R}$. For $\mathbf{a}, \mathbf{b} \in T((\mathbb{R}^n))$ we define the following operations:
	    \begin{equation*}
	    \mathbf{a}+\mathbf{b} = (a_0+b_0, a_1 +b_1 ,\dots ),
	    \end{equation*}
	    \begin{equation*}
	    \mathbf{a}\otimes\mathbf{b} = (c_0, c_1 ,\dots ) \textrm{ for } c_k= \sum_{j=0}^k a_j \otimes b_{k-j},
	    \end{equation*}
	    \begin{equation*}
	    \lambda\mathbf{a}= (\lambda a_0, \lambda a_1 ,\dots ) \textrm{ for } \lambda\in\mathbb{R}.
	    \end{equation*}
	    Note that $T((\mathbb{R}^n))$ forms a real non-commutative unitial algebra under these operations with unit element $\mathbf{1}=(1, 0, 0, \dots)$. Moreover, $T_1((\mathbb{R}^n))\subset T((\mathbb{R}^n))$, defined as 
	    $$T_1((\mathbb{R}^n)) := \left\{\mathbf{a} \mid \mathbf{a} \in T((\mathbb{R}^n)) \textrm{ with } a_0 = 1\right\}$$
	    is a group under the operation $\otimes$ with unit element $\mathbf{1} = (1, 0, 0, \dots ) $. For $\mathbf{a}\in T_1((\mathbb{R}^n))$ its inverse is defined as $$ \mathbf{a}^{-1} = \sum_{k\geq 0} \left( \mathbf{1} - \mathbf{a} \right)^k.$$ 
        We define the \emph{truncated} tensor algebra at level $N$ by 
        \begin{equation*}
		    T^N(\mathbb{R}^n) := \left\{ (a_0, a_1, ..., a_N ) | a_k \in (\mathbb{R}^n)^{\otimes k}, 0 \leq k \leq N\right \} 
	    \end{equation*}
	    and denote by $T_1^N(\mathbb{R}^n)$ the corresponding \emph{truncated} tensor group. We denote by $(e_1, ..., e_n)$ the canonical basis of $\mathbb{R}^n$. Let $I= (i_1, ..., i_m )$ be a multi-index and denote its length by $|I|=m$, then $e_I := e_{i_1} \otimes ... \otimes e_{i_m}$ a basis element of $(\mathbb{R}^n)^{\otimes m}$. We use the convention that $|I|=0$ means that $I=\emptyset$, the empty word. In this formulation, we can write any $\mathbf{a}\in T^N(\mathbb{R}^n)$ as $$\mathbf{a} = \sum_{0\leq|I|\leq N} \langle e_I, \mathbf{a} \rangle e_I,$$
	    where $\langle \cdot , \cdot \rangle$ is defined as the inner product of $(\mathbb{R}^n)^{\otimes k} $ for each $k\geq 0$.
	    Finally we introduce the two canonical projections 
	    \begin{align*}
	        \Pi_k : T((\mathbb{R}^n)) &\rightarrow (\mathbb{R}^n)^{\otimes k} \\
	                \mathbf{a} &\mapsto a_k,
	    \end{align*}
        \begin{align*}
	        \Pi_{\leq N} : T((\mathbb{R}^n)) &\rightarrow T^N(\mathbb{R}^n) \\
	                \mathbf{a} &\mapsto (a_0, ..., a_N).
	    \end{align*}
	    Of course, these projections can also be applied on the truncated tensor algebra, i.e.~$\Pi_k :T^N(\mathbb{R}^n) \rightarrow (\mathbb{R}^n)^{\otimes k}  $ for $k\leq N$ and $\Pi_{\leq N}: T^M(\mathbb{R}^n) \rightarrow T^N(\mathbb{R}^n)$ for $M\geq N$. \\ \\
        We define accordingly $T_0^N(\mathbb{R}^N):=\{\mathbf{a}\in T^N(\mathbb{R}^d) \lvert a_0=0\}$. In fact, it holds that $T_1^N(\mathbb{R}^d)$ is a Lie group under the tensor product $\otimes$ and $T_0^N(\mathbb{R}^d)$ is a Lie algebra with Lie bracket $T_0^N(\mathbb{R}^d)\times T_0^N(\mathbb{R}^d) \ni (\mathbf{a}, \mathbf{b}) \mapsto [\mathbf{a}, \mathbf{b}]:= \Pi_{\leq N}(\mathbf{a}\otimes \mathbf{b} - \mathbf{b}\otimes \mathbf{a}) \in T_0^N(\mathbb{R}^d)$. The two are related by the exponential map 
        \begin{align*}
        \exp: T_0^N(\mathbb{R}^d) &\rightarrow  T_1^N(\mathbb{R}^d)\\
        \mathbf{a} &  \mapsto \Pi_{\leq N} \left(1 + \sum_{k=1}^N \frac{\mathbf{a}^{\otimes k}}{k!}\right).
        \end{align*}
        We call $\mathfrak{g}^N(\mathbb{R}^d) = \mathbb{R}^d \oplus [\mathbb{R}^d, \mathbb{R}^d] \oplus \dots \oplus [\mathbb{R}^d, [\dots, [\mathbb{R}^d,\mathbb{R}^d]]]$
        the \emph{free step-N nilpotent Lie algebra} and $G^N(\mathbb{R}^d)= \exp(\mathfrak{g}^N(\mathbb{R}^d))$ the \emph{free nilpotent group of step N}.
        
	\end{definition}
	
	We are now ready to define the signature of a continuous semimartingale. For an introduction of the signature of a continuous semimartingale in a similar spirit, see~\cite{Cuchiero_Svaluto-Ferro_Gazzani}. To this end, consider a finite time-horizon $T>0$ and some filtered probability space $(\Omega, \mathcal{F}, (\mathcal{F}_t)_{t\in[0,T]}, \mathbb{P})$.
	
	\begin{definition}[Signature of a Continuous Semimartingale]\label{def:def_signature}
	    For an $\mathbb{R}^n$-valued continuous semimartingale $X$, we denote its signature by $\mathbb{X}$ and define $\mathbb{X}$ as the process in $T_1((\mathbb{R}^n))$ whose value at time $t$ is given via
	    \begin{align*}
	        \Pi_k(\mathbb{X}_{t}) = \int_{0 \leq u_1 \leq \dots \leq u_k \leq t} \otimes \circ dX_{u_1 } \dots  \otimes \circ dX_{u_k},
	   \end{align*}
	   where $\circ$ denotes Stratonovic integration.
	   Accordingly, we define the truncated signature at level $N$, which we denote by $\mathbb{X}^N$ as
	   \begin{equation*}
	       \mathbb{X}^N_t = \Pi_{\leq N}(\mathbb{X}_t).
	   \end{equation*}
        Note that $\mathbb{X}^N_t$ takes values in $G^N(\mathbb{R}^n)$ for all $t\in [0,T]$. 
	   We define the increments of the signature $\mathbb{X}_{s,t}$ by $\mathbb{X}_{s,t} := \mathbb{X}_s^{-1} \otimes \mathbb{X}_t$
	   and note that 
	    \begin{align*}
	        \Pi_k(\mathbb{X}_{s,t}) = \int_{s \leq u_1 \leq \dots \leq u_k \leq t} \otimes \circ dX_{u_1 } \dots  \otimes \circ dX_{u_k},
	   \end{align*} 
	   in particular, $\Pi_1(\mathbb{X}_{s,t})= X_t - X_s$.
	\end{definition}
	
    Let us now state a few important and useful properties about the signature. 
    
    \begin{lemma}[Shuffle Product Property]
        For two multiindices $I$, $J$ we define the shuffle product $\shuffle$ as 
        $$ e_I \shuffle e_{\emptyset} = e_{\emptyset} \shuffle e_I = e_I$$
        $$ (e_I\otimes e_i ) \shuffle (e_J \otimes e_j) = (e_I \shuffle (e_J \otimes e_j)) \otimes e_i + ((e_I\otimes e_i )\shuffle e_J)\otimes e_j.$$
        Then, for any two multiindices $I$, $J$ and any continuous semimartingale $X$, it holds that 
        \begin{equation*}
            \langle e_I, \mathbb{X} \rangle \langle e_J, \mathbb{X} \rangle = \langle e_I \shuffle e_J, \mathbb{X} \rangle.
        \end{equation*}
    \end{lemma}
    \begin{proof}
        This follows directly from the properties of the Stratonovic integral. 
    \end{proof}
	
	\begin{lemma}[Uniqueness of the Signature]\label{lem:uniqueness_sig}
	    Let $X$ be a continuous semimartingale with at least one strictly monotone component. Then the signature uniquely determines $X$ up to vertical translations of the trajectories.
    \end{lemma} 
    \begin{proof}
        This follows from the fact, that the signature determines the underlying semimartingale uniquely up to tree-like equivalences, see~\cite{Horatio}. A more direct proof, in the case where that component is the time, can  be found e.g.~in~\cite{Cuchiero_Primavera_Svaluto-Ferro2022}.
    \end{proof}
    
    We now define a class of functions which plays a major role in this paper, namely \emph{linear functions on the signature}:
    
    \begin{definition}[Linear Functions on the Signature]
        Let $X$ be an $\mathbb{R}^n$-valued continuous semimartingale and $\mathbb{X}$ its signature. We call a function $L: T((\mathbb{R}^n)) \rightarrow \mathbb{R}$ a \emph{linear function on the signature} if $L$ is given by 
        $$ L(\mathbb{X}_t) = \sum_{0\leq|I|\leq N} l_I \langle e_I, \mathbb{X}_t \rangle$$
        for some coefficients $l_I \in \mathbb{R}$ and any $N\in \mathbb{N}$. It is important to note that all $l_I$ are \emph{constant in time}.
        We denote the set of all linear functions on the signature by $\mathfrak{L}$. 
    \end{definition}

\subsection{Universal Approximation Theorem for Linear Functions on the Signature} \label{subsec:UAT}   
    
    Linear functions on the signature are dense in a set of certain continuous path-functionals defined on compact sets of paths. We call this the \emph{universal approximation property} and state it rigorously in Theorem~\ref{thm:UAT_sig}. To this end we need the following preparations which have been formulated 
    in a similar manner in~\cite{Bayer20, Kalsi, Oberhauser_2014}. 

Let us first define for some $p\in (2,3)$ the \emph{p-variation metric} of paths $\mathbf{x}, \mathbf{y}$ with values in  the group $T_1^N(\mathbb{R}^n)$  by
        \begin{equation}\label{eq:p-var_metric}
			d_{p-var;[0,t]}( \mathbf{x}, \mathbf{y}) = \max_{k\in {1,...,N}} \sup_{D \in\mathcal{D}} \left(\sum_{t_i, t_{i-1} \in D} |\Pi_k(\mathbf{x}_{t_{i-1},t_i} - \mathbf{y}_{t_{i-1},t_i})|^{\frac{p}{k}} \right)^{\frac{k}{p}},
		\end{equation}
        where $\mathcal{D}$ is the set of all partitions of $[0,t]$ and $|\cdot|$ is the Euclidean distance. The definition of increments of paths with values in  the group $T_1^N(\mathbb{R}^n)$ is as for the signature in Definition~\ref{def:def_signature}, namely $\mathbf{x}_{s,t} := \mathbf{x}_{s}^{-1}\otimes \mathbf{x}_t $.

    \begin{definition}[Lifted Paths]
        Let $\varphi: [0,T] \rightarrow \mathbb{R}$ be a strictly monotone increasing function. For $N\geq 2$ and $t \in [0,T]$ we define by 
        \begin{align*}
         \mathcal{G}^N_t(\varphi, x) = \bigcup_{X \textrm{ cont. semimart.}} \{  \hat{\mathbb{X}}^N_{[0,t]}(\omega) \, | \, \hat{X}_s=(\varphi(s), X_s), \omega \in \Omega_X, X_0=x\} 
        \end{align*}
        the space of \emph{lifted paths}  $\hat{\mathbb{X}}^N_{[0,t]}(\omega)$
        of all $\mathbb{R}^n$-valued continuous semimartingales with initial value $x$ and for each $X$, ${\Omega}_X \subset \Omega$  is taken to be a probability-one set on which the  metric given in \eqref{eq:p-var_metric} is well-defined for all $t\in[0,T]$. Note that we denote by $\hat{\mathbb{X}}$ the signature of $\hat{X}$ and we use the subscript $[0,t]$ to emphasize that we take the \emph{entire path} between $[0,t]$.  We equip $\mathcal{G}_t^N(\varphi, x)$ with the metric $d_{p-var; [0,t]}$ (now for paths with values in $T_1^N(\mathbb{R}^{n+1})$). 
    \end{definition}

    \begin{remark}
        We would like to point out that using a strictly monotone increasing function $\varphi$ other than the identity can be useful in practice. For example, we use the time-augmentation $\varphi_T(t):= \frac{t}{T}$ for our experiments using real market data in Section~\ref{sec:realdata} to account for the different length of the training and testing period. 
    \end{remark}

    \begin{definition}[Space of Lifted (Stopped) Paths]\label{def:stoppedlifted}
        We define the space of \emph{lifted stopped paths} via
        \begin{align*} {\Lambda}^N_T(\varphi, x) = \bigcup_{t\in[0,T]} \mathcal{G}_t^N(\varphi, x) \end{align*}
        We then equip ${\Lambda}_T^N(\varphi, x)$ with the metric $d_{\Lambda}$ with its form being based on~\cite{Cont_2010, Dupire_2019} 
        \begin{align*}
            &d_{\Lambda}\left(\hat{\mathbb{X}}^N_{[0,t]}(\omega_1), {\hat{\mathbb{Y}}}^N_{[0,s]}(\omega_2)\right) \notag \\&\quad = |t-s| + d_{p-var; [0,t\vee s]}\left(\widehat{\mathbb{X}^t}^N_{[0,t \vee s]}(\omega_1), \widehat{{\mathbb{Y}}^s}^N_{[0,t \vee s]}(\omega_2)\right),
        \end{align*}
        where we denote by $\widehat{X^t}$  the process  $(\widehat{X^t})_u = (\varphi(u),X^t_u)$, i.e. the process where we stop the semimartingale $X$ at time $t$ but not the time-augmentation $\varphi$ and write $\widehat{\,\mathbb{X}^t}$ for the signature of $\widehat{X^t}$. See \cite{Bayer20} for a similar definition in the setting of general rough paths.       
    \end{definition}
        
    \begin{remark}
        Note that if $\hat{\mathbb{X}}^N_{[0,t]}(\omega) \in \Lambda_T^N(\varphi, x)$, then $\widehat{\,\mathbb{X}^t}^N_{[0,T]}(\omega) \in \mathcal{G}_T^N(\varphi, x) $ and hence $\widehat{\,\mathbb{X}^t}^N_{[0,T]}(\omega) \in \Lambda_T^N(\varphi, x) $. This is the reason, we do not stop the time-augmentation $\varphi$.  Indeed, $(\hat{\mathbb{X}}^{N})^t_{[0,T]}(\omega) \notin \mathcal{G}_T^N(\varphi, x) $, since $\varphi$ is required to be strictly monotone.\\ 
    \end{remark}       
    Based on~\cite{Cont_2010} we formulate the notion of non-anticipative path-functionals in our setting:
    \begin{definition}[Non-Anticipative Path-Functionals]
        We call functionals $f: {\Lambda}^2_T(\varphi, x) \rightarrow \mathbb{R}$ \emph{non-anticipative path-functionals}. 
    \end{definition}

The following continuity result for the signature is crucial for the universal approximation property.
    
    \begin{cor}[Lyon's Lift]\label{cor:lyons_lift}
        For all $t\in[0,T]$ and $N\geq 2$ there exists a unique bijection 
        \begin{align*}
            S^N : \mathcal{G}_t^2(\varphi, x) \mapsto \mathcal{G}_t^N(\varphi, x)
        \end{align*}
        with inverse $\Pi_{\leq2}$ and $S^N$ is continuous on bounded sets.
        We call $S^N$ \emph{Lyons' lift}. Moreover, for each $X$ and all  $\omega \in {\Omega}_X$ it holds that 
        \begin{equation*}
            S^N\left(\hat{\mathbb{X}}^2_{[0,t]}(\omega)\right) = \hat{\mathbb{X}}^N_{[0,t]}(\omega).
        \end{equation*}
    \end{cor}
    
    \begin{proof}
        The first part of the theorem follows from \cite[Theorem 8.10, Theorem 9.5 and Corollary 9.11]{Friz_Victoir} and the second statement follows from \cite[Proposition 17.1 and Exercise 17.2]{Friz_Victoir}.
    \end{proof}

    Relying on two lemmas proved in  Appendix \ref{app:Proofs_Foundations}, we can now  show the universal approximation theorem for linear functions on the signature. We here state a version that holds uniformly in time. For this we also provide a rigorous proof in Appendix \ref{app:Proofs_Foundations} which is -- up to our knowledge --  not available in the literature. For similar assertions we however refer to~\cite{Bayer20, Kalsi}.
    In the setting of c\`adl\`ag rough paths an analogous result has been proved in \cite{Cuchiero_Primavera_Svaluto-Ferro2022}, however with different topologies.

    \begin{theorem}\label{thm:UAT_sig}
        Let $K\subset \mathcal{G}_T^2(\varphi,x)$ be a compact subset. For any continuous non-anticipative path-functional $f: \Lambda_T^2(\varphi,x) \rightarrow \mathbb{R}$ and for every $\epsilon> 0$ there exists a linear function on the signature $L$ such that 
        $$ \sup_{(t,  \hat{\mathbb{X}}^2_{[0,T]}(\omega)) \in [0,T] \times K} |f(\hat{\mathbb{X}}^2_{[0,t]}(\omega)) - L(\hat{\mathbb{X}}_t)(\omega)| < \epsilon.$$ 
 
    \end{theorem}

  Let us point out that, conceptually, the reason for Theorem~\ref{thm:UAT_sig} to hold uniformly in time is due to the time-augmentation. To show the applicability of this result let us give some examples of non-anticipative path-functionals.
    
    \begin{example} [Examples of Continuous Non-Anticipative Path-Functionals] \label{cor:example_cont_func}
       
        \begin{enumerate}
        \item[]
            \item \textbf{Solutions of SDEs driven by $X$:} 
            For $\gamma>2$, an $\mathbb{R}^n$-valued continuous semimartingale $X$ and an $\mathbb{R}^m$-valued continuous process of bounded variation $H$, consider the SDE 
            \begin{equation*}
                dY_t = \sum_{i=1}^n V_i(Y_t) \circ dX^i_t + \sum_{j=1}^m W_j(Y_t)dH^j_t, \hspace{0.5cm} Y_0=y_0\in \mathbb{R}^d,
            \end{equation*}
            where $\{V_i\}_{1\leq i \leq n}$ is a collection of vector fields in $\mathrm{Lip}^{\gamma}(\mathbb{R}^d)$ and $\{W_i\}_{1\leq i \leq m}$ is a collection of $\mathrm{Lip}^{1}(\mathbb{R}^d)$ vector fields.\footnote{Following the definition in \cite[Definition 10.2]{Friz_Victoir} a vector field $V$ is in $\mathrm{Lip}^{\gamma}$ if $V$ is $\lfloor \gamma \rfloor$ times continuously differentiable and if the supremum norm of $V$ and the supremum norm of its first $\lfloor \gamma \rfloor$ derivatives admits a common bound.} Set $Z=(X,H)$ and $z=(X_0,H_0)$ and define $D_{Z}$ to be the set 
            $$ D_{Z} = \bigcup_{t\in[0,T]} \left\{ {\hat{\mathbb{Z}}}_{[0,t]}^2(\omega) | \omega \in  \Omega_Z \right \} \subset \Lambda_T^2(\varphi, z).$$
            
            Then  the components of the solution to this SDE $(Y^i_t)_{t\in[0,T]}$, can be seen as a map 
            \begin{align*}
                Y^i_t(\omega): \, &D_{Z} \rightarrow \mathbb{R}\\
                &{\hat{\mathbb{Z}}}_{[0,t]}^2(\omega) \mapsto Y_t^i(\omega)
            \end{align*}
            and it holds almost surely (i.e.~on ${\Omega}_Z$) that $Y_t^i\in C(D_{Z}; \mathbb{R})$ for every $t \in [0,T]$. The continuity follows from the fact that the solution to the corresponding rough differential equation solves the SDE almost surely by \cite[Theorem 17.3]{Friz_Victoir}  and the RDE solution is continuous \cite[Corollary 10.28 \& Theorem 12.10]{Friz_Victoir} in the lifted input path. Hence, we can apply Theorem~\ref{thm:UAT_sig} and obtain that for any compact set $K\subset \{ \hat{\mathbb{Z}}^2_{[0,T]}(\omega) | \omega \in \Omega_Z  \}$ and all $\epsilon>0$ it holds that for all $i\in\{1,\dots,d\}$ there exists a linear function $L^i$ on the signature such that 
            $$\sup_{(t\in [0,T], \hat{\mathbb{Z}}^2_{[0,T]}(\omega)) \in K} \left| Y_t^i(\omega) - L^i( \hat{\mathbb{Z}}_t(\omega)) \right|  < \epsilon.$$         
            
            \item \textbf{Continuous Functions Without Path-Dependence:}
            Consider an arbitrary but fixed $\mathbb{R}^n$-valued continuous semimartingale $X$ with initial value $X_0=x$. For any $t\in[0,T]$ let
            \begin{align*}
                h: \mathbb{R}^{d+1} &\rightarrow \mathbb{R} \\
                \hat{X}_t(\omega) &\mapsto h\left(\hat{X}_t(\omega)\right)
            \end{align*} 
            be a continuous function.
            Such a function can for almost all  $\omega \in \Omega$ be extended to a non-anticipative path-functional $\tilde{h}$ on $\Lambda_T^2(\varphi, x)$. 
            To be precise, simply define $\tilde{h}$ as
            \begin{align*}
                \tilde{h}:\, \Lambda_T^2(\varphi, x) &\rightarrow \, \mathbb{R} \\
                \hat{\mathbb{X}}^2_{[0,t]}(\omega) &\mapsto \, \tilde{h}\left(\hat{\mathbb{X}}^2_{[0,t]}(\omega)\right):= h\left((\hat{X}_{[0,t]})_{\odot}(\omega) \right) = h\left(\hat{X}_t(\omega)\right) .
            \end{align*} 

            In particular, $\tilde{h} \in C\left( \Lambda_T^2(\varphi, x); \mathbb{R}\right)$. To see this, note that for any $\hat{\mathbf{x}}_{[0,t]}, \hat{\mathbf{y}}_{[0,s]} \in \Lambda_T^2(\varphi, x)$ it holds that
            \begin{align*}
                d_{\Lambda}(\hat{\mathbf{x}}, \hat{\mathbf{y}}) \geq |\hat{x}_{0,t} - \hat{y}_{0,s}| = |\hat{x}_t - \hat{y}_s|
            \end{align*}
            with $\hat{x}=\Pi_1(\hat{\mathbf{x}})$, $\hat{y}=\Pi_1(\hat{\mathbf{y}})$, where we replaced the $\max_{k\in \{1,\dots, N\}}$ by $k=1$ and $\sup_{D\in\mathcal{D}}$ by $D=[0,t\vee s]$ in the definition of $d_{p-var;[0,t\vee s]}$ to obtain the inequality. The last equality follows from the fact that all paths in $\Lambda_T^2(\varphi, x)$ have the same initial value. 
Due to the continuity of $h$, for every $\epsilon>0$ there exists a $\delta >0$ such that for all $t_1, t_2 \in[0,T]$ and all paths $X(\omega_1), {Y}(\omega_2)$ satisfying $\left|\hat{X}_{t_1}(\omega_1) - \hat{Y X}_{t_2}(\omega_2)\right|< \delta$ we have
$$\left|h\left(\hat{X}_{t_1}(\omega_1)\right)-h\left(\hat{X}_{t_2}(\omega_2)\right)\right|< \epsilon.$$ 
By the above estimate, $d_{\Lambda}\left(\hat{\mathbb{X}}^2_{[0,t_1]}(\omega_1),\hat{\mathbb{X}}^2_{[0,t_2]}(\omega_2)\right) < \delta$ therefore implies
            \begin{align*}
            \left|\tilde{h}\left(\hat{\mathbb{X}}^2_{[0,t_1]}(\omega_1)\right)-\tilde{h}\left(\hat{\mathbb{X}}^2_{[0,t_2]}(\omega_2)\right)\right|< \epsilon \end{align*}
       and thus proves the claim.     
        \end{enumerate}
  \end{example}
    
    \subsection{Global Universal Approximation Theorem for Linear Functions on the Signature} \label{subsec:GUAT}
    In this subsection, we state an alternative universal approximation theorem which is due to \cite{Cuchiero_Schmocker_Teichmann_2023}. We will refer to Theorem~\ref{thm:GUAT} also as \emph{global} universal approximation theorem because we are no longer limited to compact sets, however this comes at the price of approximating with respect to weighted norms instead of the uniform norm.
    The following is based on the statements and the results of \cite{Cuchiero_Schmocker_Teichmann_2023} adapted to the setting considered in this paper. However, we do formulate everything in the path-functional setting here, leading to a approximation \emph{uniformly in time} which has not been done in \cite{Cuchiero_Schmocker_Teichmann_2023}. 

    Let us start by introducing the following metrics:

    \begin{definition}[The Carnot-Carathéodory Norm and Homogeneous $\alpha$-Hölder Metric]
    \hspace*{0.1cm}
    \begin{enumerate}
        \item Let $\mathbf{a}\in G^N(\mathbb{R}^n)$, we define the \emph{Carnot-Carathéodory} norm $\|\cdot \|_{cc}$ by 
        \begin{align*}
        \|\mathbf{a}\|_{cc}:= \inf \Big\{ \int_0^T \|dX_t\|\,  \Big\lvert \, &X:[0,T]\rightarrow \mathbb{R}^d \textrm{ continuous and of finite variation,}\\ &\textrm{such that } \mathbb{X}^N_T=\mathbf{a}\Big\}.
        \end{align*}
        Moreover, for $\mathbf{a}, \mathbf{b}\in G^N(\mathbb{R}^n)$ the corresponding metric $d_{cc}(\cdot, \cdot)$ is induced by 
        $$ d_{cc}(\mathbf{a},\mathbf{b}):= \|\mathbf{a}^{-1}\otimes \mathbf{b}\|_{cc}\,.$$

        \item Recall the space of lifted path $\mathcal{G}_t^N(\varphi, x)$ for $N\geq 2$ and $t\in[0,T]$. For two lifted paths $\hat{\mathbb{X}}^N_{[0,t]}(\omega_1), \hat{\mathbb{Y}}^N_{[0,t]}(\omega_2) \in \mathcal{G}_t^N(\varphi, x)$ we define the \emph{homogeneous }$\alpha$\emph{-H\"older metric} $d_{\alpha,[0,t]}(\cdot,\cdot)$ for $0 \leq \alpha <\frac{1}{2}$ by 
        \begin{equation*}
            d_{\alpha, [0,t]}\left(\hat{\mathbb{X}}^N_{[0,t]}(\omega_1), \hat{ \mathbb{Y}}^N_{[0,t]}(\omega_2)\right) = \sup_{\substack{r,s\in[0,t]\\ r<s}} \frac{d_{cc}\left(\hat{\mathbb{X}}^N_{r,s}(\omega_1), \hat{ \mathbb{Y}}^N_{r,s}(\omega_2)\right)}{|r-s|^{\alpha}}.
        \end{equation*}
    \end{enumerate}
        
    \end{definition}

    \begin{notation}   
     In the following we shall denote by $\mathcal{G}_t^{N,\alpha}(\varphi, x):=(\mathcal{G}_t^N(\varphi, x), d_{\alpha,[0,t]})$ the space of lifted paths equipped with the homogeneous $\alpha$-H\"older metric for some $\frac{1}{3} <\alpha< \frac{1}{2}$.  Moreover, we write $\Lambda^{N,\alpha}_T(\varphi, x):=(\Lambda^{N}_T(\varphi, x), d_{\Lambda^{\alpha}})$ for the corresponding space of lifted stopped paths equipped with $$
    d_{\Lambda^{\alpha}}\left(\hat{\mathbb{X}}^N_{[0,t]}(\omega_1), \hat{\mathbb{Y}}^N_{[0,s]}(\omega_2)\right):= |t-s| + d_{\alpha, [0, t\vee s]}\left(\widehat{\,  \mathbb{X}^t}^N_{[0,t\vee s]}(\omega_1), \widehat{\mathbb{Y}^s}^N_{[0,t\vee s]}(\omega_2)\right). $$

 We shall assume without loss of generality that for each continuous semimartingale $X$, the associated probability-one set ${\Omega}_X$ in the definition of $\mathcal{G}_t^N(\varphi, x)$ is such that the homogeneous $\alpha$-H\"older norm is finite for all lifted paths and all $0\leq \alpha < \frac{1}{2}$.
    
   \end{notation}

  \begin{remark}
    Due to the fact that we are now considering different topologies on the space of lifted (stopped) paths, the continuity statements of Lemma~\ref{lem:lambda_continuity}, Corollary~\ref{cor:lyons_lift} and Lemma~\ref{lem:e_I_continuous} need to be proved again. In fact, they do hold analogously:
    \begin{itemize}
    \item The proof that the map $\lambda$ defined in Lemma~\ref{lem:lambda_continuity} is again continuous in the $\alpha$-H\"older topology works very similar to the one of Lemma~\ref{lem:lambda_continuity}, where one ultimately chooses $\delta>0$ such that $$\delta + d_{\alpha, [0,t+\delta]}\left(\widehat{ \mathbb{Y}^s}^N_{[0,t+\delta]}(\omega_1), \widehat{\,  \mathbb{Y}^t}^N_{[0,t+\delta]}(\omega_1)\right) < \epsilon.$$
    \item The continuity of Lyons' lift holds by \cite[Theorem 9.5, Corollary 9.11]{Friz_Victoir}.
    \item The proof of Lemma~\ref{lem:e_I_continuous} holds analogously in the $\alpha$-H\"older topology. 
    \end{itemize}
    \end{remark}

In order to formulate the global universal approximation theorem we need the notion of a weighted space and certain weighted function spaces being generalizations of continuous functions.

\begin{definition}[Weighted Space and Weighted Function Space]
        Let $(X, \tau_X)$ be a completely regular Hausdorff space and consider a function $\psi:X \rightarrow (0,\infty)$. We call the pair $(X, \psi)$ a \emph{weighted space} if it holds that $K_R:= \psi^{-1}((0,R])= \{x\in X \lvert \psi(x) \leq R\}$ is $\tau_X$-compact for all $R>0$. Moreover, for a weighted space $(X, \psi)$, we define $$B_{\psi}(X)=\left\{ f: X\rightarrow \mathbb{R} \, \lvert \,\sup_{x\in X} \frac{|f(x)|}{\psi(x)} < \infty \right\}$$ and equip it with $\|f\|_{\mathcal{B}_{\psi}(X)}:= \sup_{x\in X} \frac{|f(x)|}{\psi(x)}$. We denote by $\mathcal{B}_{\psi}(X)$ the closure of the space of bounded continuous functions $C_b(X)$ under $\|f\|_{\mathcal{B}_{\psi}(X)}$. We refer to $\mathcal{B}_{\psi}(X)$ as weighted function space. 
    \end{definition}

The global universal approximation theorem via linear functions on the signature hinges on an application of the weighted Stone-Weierstrass theorem \cite[Theorem 3.6]{Cuchiero_Schmocker_Teichmann_2023}. In contrast to~\cite[Theorem 5.4]{Cuchiero_Schmocker_Teichmann_2023} we want to extend the result to be uniform in time and hence to the space of lifted stopped paths. While $\Lambda_T^2(\varphi, x)$ is \emph{not} a weighted space, we can show that a larger space which contains $\Lambda_T^2(\varphi, x)$ is a weighted space. In fact, we need to consider the "stopped paths" of all \emph{weakly geometric $\alpha$-H\"older rough path} and not just those stemming from a continuous semimartingale. To this end we introduce the space of \emph{weakly geometric $\alpha$-H\"older rough paths} for $\frac{1}{3} < \alpha < \frac{1}{2}$ starting in $x\in \mathbb{R}^n$ and time-augmented by a continuous and strictly increasing function $\varphi$ and denote it by \begin{align*}
\hat{C}^{\alpha}_{t}(\varphi, x):= \Big\{ \hat{\mathbf{X}}_{[0,t]} \in C^{\alpha}([0,T]; G^2(\mathbb{R}^{n+1}))\, \big\lvert \, &\langle e_1, \hat{\mathbf{X}}_s \rangle = \varphi(s) \textrm{ for all } s\in [0,t] \\ &\textrm{ and } \langle e_i, \hat{\mathbf{X}}_0 \rangle= x^i \textrm{ for all } i \in \{2, \dots, n+1\} \Big\}.
\end{align*}
The paths which are a realization of a continuous semimartingale, i.e. $\mathcal{G}_t^{2,\alpha}(\varphi, x)$, is a subspace of $\hat{C}^{\alpha}_{t}(\varphi, x)$. In the following, we will slightly abuse notation and extend $d_{\alpha, [0,t]}$ to elements of $\hat{C}^{\alpha}_{t}(\varphi, x)$, as well as $d_{\Lambda^{\alpha}}$ to elements of $\bigcup_{t\in [0,T]} \hat{C}^{\alpha}_{t}(\varphi, x)$. To do so, we need a notion of how to stop a weakly geometric rough path such that this is compatible with how we defined the signature of the stopped path in the semimartingale setting. 

\begin{definition}
For any $\frac{1}{3} < \alpha < \frac{1}{2}$, $t\in [0,t]$ and $\hat{\mathbf{X}}_{[0,t]}\in \hat{C}^{\alpha}_t(\varphi,x)$ we define the associated \emph{ stopped rough path } $\widehat{\mathbf{X}^t}_{[0,T]}$ in the following way; for all $s\in [0,t]$, $(\widehat{\mathbf{X}^t})_s :=\hat{\mathbf{X}}_s $  and for all $r \in [t, T]$
\begin{equation*}
    \langle e_I , (\widehat{\mathbf{X}^t})_r \rangle := \begin{cases} \varphi(r) \qquad \textrm{ for } I=(1) \\ \frac{1}{2}(\varphi(r))^2 \,\textrm{ for } I=(11) \\
    \langle e_I , \hat{\mathbf{X}}_t \rangle \,\,\, \textrm{ for } I= (i) \textrm{ or } I=(ji) \textrm{ where } i\in \{2, \dots, n+1\}, j \in \{1, \dots, n+1\} \\ \varphi(r)\cdot \langle e_i, \hat{\mathbf{X}}_t \rangle - \langle e_{(1i)} , \hat{\mathbf{X}}_t \rangle \textrm{ for } I= (i1) \textrm{ where } i\in \{2, \dots, n+1\}.
    \end{cases}
\end{equation*}
\end{definition}
Let us point out that the above definition is constructed such that any stopped weakly geometric rough path remains weakly geometric and such that it coincides with how we constructed the lifts of stopped semimartingales. 

We are now ready to  formulate the following lemma which we prove in Appendix~\ref{app:Proofs_Foundations}.

\begin{lemma}\label{lem:weighted}
Let $\frac{1}{3} < \alpha < \frac{1}{2}$ and define the weight function 
\begin{equation}\label{eq:weight}
\psi(\hat{\mathbf{X}}_{[0,t]}):= \exp\left( \zeta \|\widehat{\, \mathbf{X}^t}_{[0,T]}\|^{\xi}_{\alpha, [0,T]}\right)
\end{equation}
for $\hat{\mathbf{X}}_{[0,t]} \in \bigcup_{t\in [0,T]} \hat{C}^{\alpha}_{t}(\varphi, x)$, $\zeta>0$, $\xi>2$, $\alpha < \frac{1}{2}$ and $\|\cdot\|_{\alpha, [0,T]}$ being the norm induced by $d_{\alpha, [0,T]}$. Then $\bigcup_{t\in [0,T]} \hat{C}^{\alpha}_{t}(\varphi, x)$ equipped with $d_{\Lambda^{\alpha'}}$ for $0\leq \alpha' < \alpha$ is a weighted space. 
\end{lemma}

We have thus shown that $\left(\bigcup_{t\in [0,T]} \hat{C}^{\alpha}_{t}(\varphi, x), d_{\Lambda^{\alpha'}} \right)$ is a weighted space for the weight function given by \eqref{eq:weight} and $0\leq \alpha' < \alpha$. We would like to emphasise that to achieve this result it is crucial to equip $\bigcup_{t\in [0,T]} \hat{C}^{\alpha}_{t}(\varphi, x)$ with a weaker topology than the $\alpha$-H\"older topology. We are finally ready to formulate the global universal approximation theorem for non-anticipative path-functionals. Its proof is also given in Appendix \ref{app:Proofs_Foundations}.

    \begin{theorem}[Global Universial Approximation Theorem]\label{thm:GUAT}
Let $\frac{1}{3} <\alpha < \frac{1}{2}$, $0 \leq \alpha' < \alpha$ and consider the weight function $\psi$ given in \eqref{eq:weight}.
    Then, for every non-anticipative path-functional $f\in \mathcal{B}_{\psi}\left( \big( \bigcup_{t\in [0,T]} \hat{C}_{t}^{\alpha}(\varphi, x), d_{\Lambda^{\alpha'}}\big)\right)$  (i.e.~$\bigcup_{t\in [0,T]} \hat{C}_{t}^{\alpha}(\varphi, x)$ equipped with $d_{\Lambda^{\alpha'}}$)
    and for every $\epsilon >0$ there exists a linear function on the signature $L$ such that 
    \begin{equation*}
        \sup_{\hat{\mathbb{X}}^2_{[0,t]}(\omega) \in \, \Lambda_T^{2, \alpha'}(\varphi,x)} \frac{\left| f(\hat{\mathbb{X}}^2_{[0,t]}(\omega)) - L(\hat{\mathbb{X}}_t)(\omega)  \right|}{\psi(\hat{\mathbb{X}}^2_{[0,t]}(\omega) )} < \epsilon.
    \end{equation*}

    \end{theorem}

    \begin{remark}
    Note that the approximation stated in Theorem~\ref{thm:GUAT} is indeed uniform in time, as the supremum is taken over all paths in $\Lambda^{2, \alpha'}_T(\varphi,x)$, i.e. over all lifted stopped paths. One could equivalently take the supremum over $(t, \hat{\mathbb{X}}^{2}_{[0,T]}(\omega)) \in [0,T]\times \mathcal{G}^{2, \alpha'}_T(\varphi, x)$. In fact, as shown in the proof of Theorem~\ref{thm:GUAT}, we can even take the supremum over $\bigcup_{t\in [0,T]} \hat{C}_{t}^{\alpha}(\varphi, x)$.
    \end{remark}
 
	\subsection{The Johnson-Lindenstrauss Lemma}
	The Johnson-Lindenstrauss Lemma~\cite{JL_Lemma} is an important mathematical result, which is interesting for machine learning and data science, as it can be seen as a dimension reduction technique. There are many variants of the Johnson-Lindestrauss lemma, we will here use a version stated in \cite[Theorem 3.1]{JL_variant}:
	
	\begin{theorem}\label{thm:JL}
		Let $n\in \mathbb{N}$, $\epsilon \in (0,\frac{1}{2})$, $\delta \in (0,1)$ and set $k=C\epsilon^{-2} \log(\frac{\delta}{2})$ for a suitable constant $C$. Then consider a random matrix $A$ of dimensions $k\times n$, whose entries $A_{ij}$ are i.i.d. random variables with $\mathbb{E}[A_{ij}]=0$, $\mathrm{Var}[A_{ij}]=\frac{1}{k}$ and sub-Gaussian tails. Then it holds for all $x \in \mathbb{R}^n$ that 
		$$ \mathbb{P}\left[ (1-\epsilon)\|x\| \leq \|Ax\| \leq (1+\epsilon)\|x\|\right]\geq 1-\delta.$$
		
		\begin{remark}
		    Note that since the above holds for all $x\in\mathbb{R}^n$, it immediately follows that it holds in particular for all $x-y$, where $x,y\in\mathbb{R}^n$. Hence, Theorem~\ref{thm:JL} can be understood to preserve distances up to $(1 \pm \epsilon)$.
		\end{remark}
		
		\begin{remark}
		    In this paper, we will use Theorem~\ref{thm:JL} to obtain a random projection of the truncated signature. We will choose the components $A_{ij}$ to be i.i.d. with $A_{ij}\sim \mathcal{N}(0,\frac{1}{k})$.
		\end{remark}
	    
	\end{theorem}

	\subsection{The Market and Important Portfolio Specifications}
	
	Consider a finite time-horizon $T>0$ and some filtered probability space $(\Omega, \mathcal{F}, (\mathcal{F}_t)_{t\in[0,T]}, \mathbb{P})$. We assume that the market consists of $d$ companies, with market capitalizations given by a vector $(S_t)_{t \in [0,T]}= (S^1_t,\ldots, S^d_t)_{t \in [0,T]}$ of $d$ positive continuous semimartingales with fixed initial value $S_0=s_0$. We assume this market fulfils the \emph{No Unbounded Profit with Bounded Risk (NUPBR)}~\cite{Schweizer_Hulley} condition. Furthermore, we assume that each company has issued a single infinitely divisible share, hence $S^1_t, \ldots, S^d_t$ correspond to the prices of these shares. We denote the relative capitalization of the companies, also called \emph{market weights} by $\mu_t =(\mu_t^1, ..., \mu_t^d)$ with $$\mu_t^i= \frac{S_t^i}{S_t^1 + \cdots+ S_t^d}.$$

 \begin{definition}\label{def:portfolios}
     We call a vector $\pi_t=  (\pi^1_t, ..., \pi^d_t)$ of predictable processes fulfilling $\sum_{i=1}^d \pi_t^i \equiv 1$ and being integrable with respect to $R$, with $R_t^i:=\int_0^t \frac{dS_u^i}{S_u^i}$, 
     an \emph{($R$-integrable)  portfolio}. 
 \end{definition}
 
    The component $\pi^i_t$ denotes the proportion of wealth invested in stock $i$ at time $t$, hence every portfolio is self-financing. Consider investing with a portfolio $\pi$ over a time-horizon $[t_0, t_1]$ for $0\leq t_0 < t_1 \leq T$. We denote its wealth process portfolio $\pi$ with initial wealth $w$ by $(W^{\pi, w}_{t})_{t\in[t_0,t_1]}$. Without loss of generality we will assume from now on that the initial investment is always one unit of currency, i.e. that $w=1$ and therefore denote $W^{\pi}_{t}:= W^{\pi, 1}_{t}$. The wealth process of a portfolio fulfills the stochastic differential equation 
	\begin{equation}\label{eq:val_process_S}
		\frac{d W_t^{\pi}}{W_t^{\pi}}= \sum_{i=1}^d \pi^i_t \frac{\dd S_t^i}{S_t^i}
	\end{equation}
	which is understood in It\^o-sense. 
	
	A special portfolio is the \emph{market portfolio} whose portfolio weights are given by the relative capitalizations $\mu_t$, whence we also refer to $\mu_t$ as the market weights. The wealth of the market portfolio is given by $$W^{\mu}_t= \frac{S_t^1+ \ldots+ S_t^d}{S_{t_0}^1+ \ldots+S_{t_0}^d}$$ since we require $W_{t_0}^{\mu}= 1$.
	
	We define the relative wealth process of a portfolio $\pi$ to be $$V_t^{\pi}:= \frac{W_t^{\pi}}{W_t^{\mu}}.$$ By using It\^o's-formula, it follows directly from \eqref{eq:val_process_S} that 
	\begin{equation}
		\frac{d V_t^{\pi}}{V_t^{\pi}}= \sum_{i=1}^d \pi^i_t \frac{d \mu_t^i}{\mu_t^i}.
	\end{equation}
	Furthermore, by applying It\^o's formula again, we obtain 
	\begin{equation}
		d\ln(V_t^{\pi})= \sum_{i=1}^d \frac{\pi^i_t}{\mu_t^i} d\mu_t^i- \frac{1}{2} \sum_{i,j=1}^d \frac{\pi^i_t}{\mu_t^i} \frac{\pi^j_t}{\mu_t^j} d[\mu^i, \mu^j]_t,
	\end{equation}
	where $[X,Y]_t$ denotes the co-variation process of continuous semimartingales $X,Y$.
	
    \begin{notation}
        We denote by $\Delta^n$ the unit simplex of $\mathbb{R}^n$ and by $\Delta_+^n$ its interior. We call a portfolio $\pi$ long-only, if $\pi_t \in \Delta^d$ for all $t\in[0,T]$ almost surely.
    \end{notation}
	
	We now turn to some portfolios which are of special interest, those are the \emph{numeraire},  \emph{log-optimal} and \emph{growth-optimal} portfolio.
	
	\begin{definition}[Numeraire Portfolio]
	   We call a portfolio $\rho$ the \emph{numeraire portfolio} if
    $\frac{W^{\pi}}{W^{\rho}}$ is a supermartingale
   for any other portfolio $\pi$.
	\end{definition}
	
	\begin{theorem}\label{thm:NUPBR_numeraire}
	   In a semimartingale market the NUPBR condition (see e.g.~\cite{Schweizer_Hulley}) holds if and only if the numeraire portfolio $\rho$ exists and $W^{\rho}_T < \infty$ almost surely.  
	\end{theorem}
	
	\begin{proof}
	    See~\cite[Theorem 4.12]{Karatzas07}, where in our case the predictable closed convex constraints $\mathfrak{C}$ are 
	    $$\mathfrak{C} = \left\{ \pi \in \mathbb{R}^d \, \middle\vert \, \sum_{i=1}^d \pi^i = 1 \right\}.$$
	    Note that \cite[Theorem 4.12]{Karatzas07} considers the setting with bank-account, however our constraint set $\mathfrak{C}$ prohibits us from investing into the bank account, i.e. $\pi^0\equiv 0$ (in the notation of \cite{Karatzas07}). 
	\end{proof}

 To connect the numeraire portfolio with the log-optimal one introduced in Definition~\ref{def:log} below, we recall  Proposition 4.19 from~\cite{Karatzas07}.
 
	\begin{prop}
	    Let $U:(0,\infty) \rightarrow \mathbb{R}$ be concave and strictly increasing. The utility maximization problem 
	    \begin{equation*}
	        \sup_{\pi} \mathbb{E} \left[ U(W_T^{\pi})\right] 
	    \end{equation*}
	    has no solution or has infinitely many solutions, if the NUPBR condition does not hold. 
	\end{prop}

	\begin{definition}\label{def:log}
	   Consider the log-utility optimization problem
	   \begin{equation*}
	        \sup_{\pi} \mathbb{E} \left[ \log(W_T^{\pi})\right]. 
	    \end{equation*}
	    If it is well-posed, we call the corresponding optimal portfolio \emph{log-optimal} portfolio.
	\end{definition}

 The following lemma can also be found in ~\cite{Karatzas07}. 
	\begin{lemma} \label{lem:num}
	    Under the NUPBR condition, the log-utility optimization problem admits a solution, if it is finite, in which case the numeraire portfolio is the log-optimal portfolio. 
	\end{lemma}

 Finally, let us introduce the notion of the growth-optimal portfolio.

	\begin{definition}\label{def:growth-opt}
	    Consider a market model 
	    $$ dS_t =\operatorname{diag}(S_t)( a_t dt + \Sigma_t dB_t), $$
	    where $a_t$ is a $d$-dimensional vector, $\Sigma_t$ a $d\times m$-matrix for $m\geq d$ and $B$ a $m$-dimensional Brownian motion. Assume $a_t$, $B_t$ are predictable processes and satisfy $$\sum_{i=1}^d \int_0^T |a_t^i| dt + \sum_{i=1}^d \sum_{j=1}^m \int_0^T |\Sigma_t^{ij}|^2 dt < \infty \hspace{0.3cm} \mathbb{P}\textrm{-a.s.}$$ and $\Sigma_t\Sigma_t^{\mathsf{T}}$ is almost surely invertible for all $t\in[0,T]$. The log-wealth-process of a portfolio $\pi$ in this market model is then given by 
	    $$ d\ln(W^{\pi}_t) = \left(\pi_t^{\mathsf{T}}a_t - \frac{1}{2} \pi_t^{\mathsf{T}} \Sigma_t \Sigma_t^{\mathsf{T}} \pi_t \right) dt + \pi_t^{\mathsf{T}} \Sigma_t dW_t.$$
	    We call $g^{\pi}_t =\pi_t^{\mathsf{T}}a_t - \frac{1}{2} \pi_t^{\mathsf{T}} \Sigma_t \Sigma_t^{\mathsf{T}} \pi_t$ the portfolios \emph{growth-rate} and the portfolio with maximal growth-rate, if it exists, the \emph{growth-optimal} portfolio. 
	    
	\end{definition}
	
	\begin{lemma}\label{lem:growth-opt}
	    Consider a market model given in Definition~\ref{def:growth-opt} and let us define  the (instantaneous) market price of risk $$\theta_t := \Sigma_t^{\mathsf{T}}(\Sigma_t\Sigma_t^{\mathsf{T}})^{-1}a_t.$$
	    The growth-optimal portfolio exists if and only if $\int_0^T \|\theta_t\|^2 dt < \infty$ almost surely. Moreover, the log-optimal portfolio exists if and only if $\mathbb{E}\left[\int_0^T \|\theta_t\|^2 dt \right]< \infty$. 
        Moreover, if the growth-optimal portfolio $\pi^{(g)}$ exists, it is given by 
		\begin{equation*}
			\pi^{(g)}_t = (\Sigma_t\Sigma_t^{\mathsf{T}})^{-1} \left( a_t - \kappa \mathbf{1}\right) \textrm{ for } \kappa = \frac{\sum_i ((\Sigma_t\Sigma_t^{\mathsf{T}})^{-1}a_t)_i - 1}{\sum_{i,j} (\Sigma_t\Sigma_t^{\mathsf{T}})^{-1}_{i,j}}
		\end{equation*}
		where $\mathbf{1}=(1, ..., 1 )^{\mathsf{T}}$.
	\end{lemma}
	\begin{proof}
	   Regarding the existence of the growth-optimal and log-optimal portfolios, see~\cite{Schweizer_Hulley}, also for the general case of continuous semimartingale markets. The specific form of the growth-optimal portfolio is due to~\cite{Platen}.
	\end{proof}

	\begin{lemma}\label{lem:numeraire_go}
        Consider the setting of a market as in Definition~\ref{def:growth-opt}. Then the growth-optimal portfolio exists if and only if the numeraire portfolio exists  in which case they are the same\footnote{Assuming that their generated wealth is almost surely finite  this holds if and only if NUPBR is satisfied (see Theorem 2.25).}.  
        Moreover, if the log-utility maximization problem is finite, then they also coincide with the log-optimal portfolio.
    \end{lemma}
    
    \begin{proof}
        See~\cite{Platen} and~\cite{Schweizer_Hulley}. The last assertion then follows from Lemma~\ref{lem:num}
    \end{proof}
    
    In SPT, there is another class of portfolios of special interest, those are \emph{functionally generated portfolios}, see~\cite{Fern02}.
    
    \begin{definition}[Functionally Generated Portfolios]
    Let $U$ be a neighbourhood of $\Delta^d$
    and consider a $C^2$-function $G: U \rightarrow \mathbb{R}_+$ such that $x_i D_i\log G(x)$ is bounded on $\Delta^d$. Then $G$ defines the functionally generated portfolio via
    $$\pi^i_t = \mu_t^i\left(D_i \log{G(\mu_t)} +1 - \sum_{j=1}^d \mu_t^j D_j \log{G(\mu_t)} \right).$$
    The function $G$ is called the \emph{portfolio generating function} and if it is concave, $\pi$ is a long-only portfolio.
    \end{definition}
    
    Before we conclude this section, we introduce the \emph{ranked market-weights} which are of particular interest in SPT, due to the remarkable stability of the capital distribution curves, see for example~\cite{Fern02, Karatzas_Ruf}.
	
	\begin{definition}[Ranked Capitalization and Market Weight Processes]
	    For an $\mathbb{R}^n_+$-valued continuous semimartingale $X$ we denote its \emph{ranked process} by $\mathbf{X}= (X^{(1)}, ..., X^{(n)})$ which is defined as 
	    \begin{equation*}
	        \max_{i\in \{1,\dots, n\}} X_t^i = X^{(1)}_t \geq X^{(2)}_t \geq \dots \geq X^{(n-1)}_t \geq X^{(n)}_t= \min_{i\in \{1,\dots, n\}} X_t^i,
	    \end{equation*}
	    where we break ties by allocating a lower rank to smaller labels. To be precise, if $X_t^i= X_t^j$ with $i\geq j$, we set $X_t^{(r_i)}=X_t^i$ and $X_t^{(r_j)}=X_t^j$ where $r_i\geq r_j$.
	    
	    Accordingly, we denote by $\boldsymbol{\mu}$ the \emph{ranked market weights} and by $\mathbf{S}$ the \emph{ranked capitalization process}.

	\end{definition}
	
	\begin{cor}
	    Given an $\mathbb{R}^n$-valued continuous semimartingale $X$, its ranked process is again an $\mathbb{R}^n$-valued continuous semimartingale. 
	\end{cor}
	\begin{proof}
	    See~\cite[Theorem 2.2]{Banner}.
	\end{proof}

	\section{(Linear) Path-Functional Portfolios}\label{sec:pathfunctionalportfolios}
	Before we present our application of signature methods in the context of SPT and portfolio optimization, we generalize the concept of functionally generated portfolios by introducing so-called \emph{path-functional portfolios}. 
	
	\begin{definition}[Path-Functional Portfolios]\label{def:path-func-port}
		We introduce \emph{path-functional portfolios} which are generated by a family of non-anticipating path functionals $\{f^i\}_{1\leq i \leq d}$ depending on time and the path of $X$ and which are constructed via an auxiliary portfolio $\tau$, where $\tau$ is required to be uniformly bounded. Moreover, we assume that the non-anticipating path functionals $\{f^i\}_{1\leq i \leq d}$ fulfill the necessary integrability conditions such that the following processes $\pi$ are portfolios in the sense of Definition~\ref{def:portfolios}. We consider two types of such portfolios:
		\begin{enumerate}[label={(\Roman*)}]
		    \item $\pi_t^i(\tau, {X})= \tau_t^i \left( f^i\left(t, X_{[0,t]}\right) +1 - \sum_{j=1}^d \tau_t^i f^j\left(t, X_{[0,t]}\right) \right)$
		    \item $\pi_t^i(\tau, {X})=  f^i\left(t, X_{[0,t]}\right) + \tau_t^i \left(1 - \sum_{j=1}^d f^j\left(t, X_{[0,t]}\right) \right)$
		\end{enumerate}
		where we use the notation of the subscript ${[0,t]}$ to make the dependence of on the entire path explicit. Moreover, we refer to the functionals $\{f^i\}_{1\leq i \leq d}$ as the \emph{portfolio \controllingfunctionsnospace}.
	\end{definition}
	
	\begin{remark}[Relation to (Classical) Stochastic Portfolio Theory]
        In the above definition, we consider a general continuous semimartingale and a general auxiliary portfolio $\tau$ in the construction of the path-functional portfolios. In the spirit of SPT, one would choose $\tau=\mu$ the market portfolio and $X=\mu$, the process of market weights. One can then recover the classical functionally generated portfolios from the path-functional portfolios of type $I$ by setting $X=\mu$ and choosing $f^i(t, \mu_{[0,t]}) =  f^i(\mu_t) = D_i \log(G(\mu_t))$. We would like to highlight that beyond this relation to classical SPT, constructing the path-functional portfolios in this form (i.e. via additive normalization) in crucial for our results in Sections~\ref{sec:approx_sig_port} and~\ref{sec:optimizing_signature_port}, in particular in view of getting quadratic optimization problems in the current multiplicative setting where the strategies are in terms fractions of wealth.
	\end{remark}

    \begin{remark}[Role of the process $X$]
    In relation to classical SPT, we will often choose $X=\mu$, i.e. construct our portfolios via information on the relative market weights. However, we would like to point out that other information may be of interest such as the absolute capitalization, market-to-book ratios~\cite{Kim_22} or earnings data, to just name a few. Adding information beyond the market weights to the portfolio construction, has also appeared e.g.~in \cite{Ruf_Xie_Lyapunov} in the context of SPT. However, we would like to point out that in \cite{Ruf_Xie_Lyapunov} the additional information has to be a continuous process of finite variation, while we allow it to be a general continuous semimartingale. 
    \end{remark}

	\begin{remark}[Path-Functional Portfolios of Type $II$]
	   A simple example of a path-functional portfolio of type $II$  would be any portfolio generated by a deep (recurrent) neural network, see for example~\cite{ DeepLearning_Sharpe}. Denote by $\mathcal{NN}^i(t, X_{[0,t]})$ the $i$-th output of a deep (recurrent) neural network. If 
	   \begin{itemize}
	       \item the last layer of the neural network is for example a softmax function and the portfolio weights are defined by $\pi^{(\mathcal{NN}), i}_t = \mathcal{NN}^i(t, X_{[0,t]})$, or,
	       \item the portfolio weights are constructed as $\pi^{(\mathcal{NN}), i}_t = \frac{ \mathcal{NN}^i(t, X_{[0,t]})}{\sum_{j=1}^d \mathcal{NN}^j(t, X_{[0,t]})}$,
	   \end{itemize}
	   then, the portfolios $\pi^{(\mathcal{NN})}$ are in both cases  path-functional portfolios of type $II$ but not of type $I$. This is by identifying $f^{(\mathcal{NN}), i}(t, X_{[0,t]})= \mathcal{NN}^i(t, X_{[0,t]})$ or $f^{(\mathcal{NN}), i}(t, X_{[0,t]})= \frac{ \mathcal{NN}^i(t, X_{[0,t]})}{\sum_{j=1}^d \mathcal{NN}^j(t, X_{[0,t]})}$ respectively. Note that the choice of $\tau$ is irrelevant because $\sum_{j=1}f^{(\mathcal{NN}), j}(t, X_{[0,t]})=1$ by definition. 
\end{remark}
	   
	\begin{remark}[Conversion between Path-Functional Portfolios of Type I\&II]
	   Note that for a fixed auxiliary portfolio $\tau$ with weight-processes which are continuous semimartingales and for a continuous semimartingale $(X_t)_{t\in[0,T]}$ any path-functional portfolio of type $I$ with portfolio \controllingfunctions $f^{(I),i} : X_{[0,t]} \mapsto f^{(I),i}(X_{[0,t]})$ is a path-functional portfolio of type $II$ with portfolio controlling function $f^{(II),i} :  (\tau, X)_{[0,t]} \mapsto \tau_t^if^{(I),i}(X_{[0,t]})$. This also holds vice versa, if the components of $\tau$ are all non-zero. Indeed, any path-functional portfolio of type $II$ with portfolio controlling function $f^{(II),i}: X_{[0,t]} \mapsto f^{(II),i}(X_{[0,t]})$ is a path-functional portfolio of type $I$ with portfolio controlling function $f^{(I),i} :  (\tau, X)_{[0,t]} \mapsto \frac{f^{(II),i}(X_{[0,t]})}{\tau_t^i}$. However, note that this conversion always involves adding the process of the auxiliary portfolio $\tau$ to the input. Not only does this enlarge the dimensions of the inputs, but also is $\tau$ then required to be a continuous semimartingale. Hence, this conversion is not always possible. 
	\end{remark}

	We now introduce a special class of path-functional portfolios which are particularly useful for optimizing path functional portfolios, as we will make more explicit in Section~\ref{sec:optimizing_signature_port}. This class is the one of \emph{linear path-functional portfolios}. 
	\begin{definition}[Linear Path-Functional Portfolios]\label{def:sig_port}
	    For a path-functional portfolio of type $I$ or $II$, we call it a \emph{linear} path-functional portfolio of the corresponding type, if the portfolio  \controllingfunctions  are of the form 
	    \begin{equation*}
	        f^i(t, X_{[0,t]})= \sum_{\nu\in \mathcal{V}} l_{\nu}^i \phi^{\nu}(t, X_{[0,t]})
	    \end{equation*}
	    where $\{\phi^{\nu}\}_{\nu\in \mathcal{V}}$ is a collection of feature maps which are themselves non-anticipating path functionals, $\mathcal{V}$ is a finite set of features and $l_{\nu}^i \in \mathbb{R}$ for each $i, \nu$ are constant (optimization) parameters. 
	\end{definition}
	
	\begin{remark}[Significance of the Auxiliary Portfolio $\tau$]\label{rem:auxiliary}
	    Considering linear path-functional portfolios gives us a first idea on how to do portfolio optimization in that context. Namely, optimizing the parameters $\{l_{\nu}^i\}_{
	    {1\leq i\leq d, \, \nu\in \mathcal{V}}}$ for a given collection of feature maps. Note, that the auxiliary portfolio $\tau$ can always be attained in this way by 
	    \begin{itemize}
	        \item setting $l_{\nu}^i=l_{\nu}^j$ for all $i,j \in \{1,\dots ,d\}$, $\nu\in \mathcal{V}$ for portfolios of type $I$
	        \item setting $l_{\nu}^i=0$ for all $i \in \{1,\dots ,d\}$, $\nu\in \mathcal{V}$ for portfolios of type $II$.
	    \end{itemize}
	    Hence, if one wants to learn a linear path-functional portfolio which outperforms a given benchmark portfolio $\bar{\pi}$, one should use the benchmark portfolio as the auxiliary portfolio, i.e. set $\tau=\bar{\pi}$, since then the benchmark portfolio is included in the family of portfolios one optimizes over. In the context of stochastic portfolio theory, we therefore often choose $\tau=\mu$, because we aim to outperform the market portfolio.
	
	\end{remark}

	Among all linear path-functional portfolios, we are in particularly interested in those, whose feature maps are  (randomized) elements of the signature, resulting in \emph{(JL- or randomized-) signature portfolios}. We make this more precise in the following definition:
	
	\begin{definition}[Signature Portfolios,  JL- and Randomized-Signature Portfolios]
	    For a given auxiliary portfolio $\tau$ and any $\mathbb{R}^n$-valued continuous semimartingale $X$, we call linear path-functional portfolios \emph{(JL- or randomized-)signature portfolios} if the feature maps are elements of the (JL- or randomized-)signature of the time-augmented semimartingale $\hat{X}$. More precisely, let $\hat{X}_t=(\varphi(t), X_t)$ be the time-augmented process of $X$, where $\varphi$ is a strictly increasing function. Recall that $\hat{\mathbb{X}}^{N}$ denotes the signature of $\hat{X}$ truncated at level $N$. We then consider the following portfolios:
	    
	    \begin{itemize}
	        \item \emph{Signature Portfolios:} A \emph{signature portfolio of degree $N$} is a linear path-functional portfolio with $\mathcal{V}= \{ I \, \mid \, I = (i_1, ..., i_m) \in \{1, \dots, n\}^m \textrm{ for } 0 \leq m \leq N \}$ (the set of multiindices up to length $N$) and with feature maps \begin{equation*} \phi^I(t, X_{[0,t]}) = \langle e_I, \hat{\mathbb{X}}^{N}_t \rangle= \langle e_I, \hat{\mathbb{X}}_t \rangle. \end{equation*} In other words, signature portfolios are path-functional portfolios, where the portfolio \controllingfunctions are linear functions on the signature.
	        
	        \item \emph{JL-Signature Portfolios:} A \emph{JL-signature portfolio of dimension  $(P, N)$} 
	        is a linear path-functional portfolio with $\mathcal{V}= \{ 1, \dots , P \}$, where $P$ is the dimension of the projection, and  for each $p\in \{1, \dots, P\}$ with feature maps \begin{equation*} \phi^p(t, X_{[0,t]}) = \langle A^p, \hat{\mathbb{X}}^{N}_t \rangle, \end{equation*} where $A^p$ is the $p$-th column of the Johnson-Lindenstrauss projection introduced in Theorem~\ref{thm:JL}. Hence, JL-signature portfolios are path-functional portfolios, where the portfolio \controllingfunctions are linear functions on the Johnson-Lindenstrauss projected signature. 
	        
	        \item \emph{Randomized-Signature Portfolios:} Let $\mathcal{S}$ denote the solution to
	        \begin{equation*}
	            d\mathcal{S}_t= \sum_{i=1}^{n+1} \sigma\left( b^{(i)} + A^{(i)}\mathcal{S}_t\right)\circ d\hat{X}_t^i \hspace{1cm} \mathcal{S}_0 = (1, 0, 0, \dots) \in \mathbb{R}^p,
 	        \end{equation*}
 	        where $\sigma$ is an activation function, randomly chosen  $b^{(i)}\in \mathbb{R}^p$, $A^{(i)}\in \mathbb{R}^{p\times p}$ for all $i \in \{1, ...,n+1\}$ and $\hat{X}$ is the time-augmented process of an $\mathbb{R}^n$-valued continuous semimartingale $X$. We call $\mathcal{S}(\hat{X})$ the randomized signature of dimension $p$ of $\hat{X}$, see~\cite{Gambara_2022, Compagnoni_2022, Cuchiero_Teichmann, cuchiero2021expressive}. For a connection with neural signature kernels and controlled ResNets we refer to \cite{cirone2023neural}. A \emph{randomized-signature portfolio of dimension  $P$} is a linear path-functional portfolio with $\mathcal{V}= \{ 1, \dots , P \}$, where $P$ is the dimension of the randomized signature, and with feature maps \begin{equation*} \phi^p(t, X_{[0,t]}) = \langle e_p, \mathcal{S}_t \rangle. \end{equation*}  Hence, randomized-signature portfolios are path-functional portfolios, where the portfolio \controllingfunctions are linear combinations of the elements of the randomized signature. 
	    \end{itemize}
	
	\end{definition}
	
	\begin{remark}\label{rem:examples_linear_port}
	    We have highlighted signature, JL- and randomized-signature portfolios as linear path-functional portfolios of special interest in this paper. However, we would like to mention some other examples of linear path-functional portfolios whose portfolio \controllingfunctions are
	    \begin{itemize} 
            \item linear functions on increments of the signature. That is, for a fixed time span $\Theta$, we observe a rolling window of length $\Theta<T$ and forget the information before that, i.e. $$\phi^I(t, X_{[0,t]})= \langle e_I, \hat{\mathbb{X}}_{0\vee(t-\Theta), t} \rangle.$$
            We will refer to such portfolios as signature portfolios \emph{with rolling windows}. We shall not put particular emphasis on these portfolios in this paper, since our approximation results in Section~\ref{sec:approx_sig_port} are tailored to signature portfolios \emph{without} rolling windows and since portfolios with rolling windows are numerically less tractable, as we outline in Section~\ref{sec:numerical_results}.

	        \item random neural networks, which are neural networks where the parameters of the hidden layers are not trained but randomly sampled and only the linear read-out layer is trained, see~\cite{Herrera_Krach}. Training the read-out layer, exactly amounts to training the parameters $\{l_{\nu}^i\}_{1\leq i\leq d, \nu \in \mathcal{V}}$ of path-functional portfolios. Such neural networks may have an infinite-dimensional input space, corresponding in our context to path spaces, as e.g. considered in
~\cite{Cuchiero_Schmocker_Teichmann_2023} or ~\cite{Benth_2022}.
         
	        \item given by reservoir computers with linear read-out layers, see for instance~\cite{Ortega_SAS}.  A special class of which are Echo State Networks~\cite{Jaeger, Ortega_Echo, Gonon} or Quantum Reservoir Computers~\cite{Nakajima_Exp, Nakajima_Theo}. Note, that there is an interesting connection between the Johnson-Lindenstrauss projection of signatures, randomized signature and reservoir computing, which is worked out in~\cite{Cuchiero_Teichmann}.
	        \item given by a constant, which leads to the case of constant portfolio weights when using a constant auxiliary portfolio $\tau$. In this case we recover the Markowitz portfolio optimization as a very special case in Subsection~\ref{subsec:MV_optim_theory}. Moreover, since the first element of the signature is constant, this class of portfolios is also included in the class of signature portfolios.  
	    
	    \end{itemize}
	    
	\end{remark}

\section{Approximation Properties of Signature Portfolios}\label{sec:approx_sig_port}

This section is dedicated to prove universal approximation properties of signature portfolios.
Let us recall that for a  compact subset $K \subset \mathcal{G}_T^2(\varphi, x)$ the set 
\begin{align}\label{eq:Klambda}
\mathcal{K}_{\Lambda}=\lambda([0,T]\times K) = \left\{\hat{\mathbb{X}}^2_{[0,t]}(\omega) \mid \, t \in [0,T] \textrm{ and }  \hat{\mathbb{X}}^2_{[0,T]}(\omega)\in K\right\}
\end{align} is a compact subset of $\Lambda_T^2(\varphi, x)$ by the continuity of $\lambda$, see Lemma \ref{lem:lambda_continuity}. Furthermore, recall that in Definition~\ref{def:path-func-port} we have introduced portfolio controlling functions $\{f^i\}_{1\leq i \leq d} $ as non-anticipative path-functionals taking as input the paths of some semimartingale $X$, i.e. the $\{f^i\}_{1\leq i \leq d}$ were of form  $f ^i(t, X_{[0,t]})$. In the following, we will require a more specific form, namely portfolio controlling functionals depending on the lifted path, that is $f^i(\hat{\mathbb{X}}^2_{[0,t]})$, where the dependence on the lifted path is important for the continuity statements. However, this is not a  contradiction to the general definition, in particular since $\hat{\mathbb{X}}^2$ is a continuous semimartingale itself. Note that the functions $\{f^i\}_{1\leq i \leq d}$ can of course only depend on the first level $X_{[0,t]}$ as well.
	\begin{theorem}\label{thm:approx_by_sig_port}
        Consider an arbitrary but fixed auxiliary portfolio $\tau$ and a path-functional portfolio $\pi(\tau, \cdot)$ of type $I$ (type $II$) with portfolio \controllingfunctions $f^i \in C(\mathcal{K}_{\Lambda}; \mathbb{R})$ for all $i \in \{1,\dots,d\}$. Then for every $\epsilon >0$ there exists a signature portfolio $\pi^*(\tau, \cdot)$ of type I (type II) such that it holds for all $i\in\{1,\dots,d\}$ 
        \begin{equation*}
            \sup_{(t,\hat{\mathbb{X}}^2_{[0,T]}(\omega))\in [0,T]\times K}| \pi^i(\tau, \hat{X}_{[0,t]})(\omega) - \pi^{*,i}(\tau, \hat{X}_{[0,t]})(\omega)|< \epsilon.
        \end{equation*}
        
		Moreover, take $\frac{1}{3} < \alpha < \frac{1}{2}$, $0\leq \alpha' < \alpha$, an arbitrary but fixed auxiliary portfolio $\tau$ and a path-functional portfolio $\pi(\tau, \cdot)$ of type I (type II) with portfolio controlling functions such that for all $i\in\{1,\dots,d\}$ $f^i \in \mathcal{B}_{\psi}\left( \big( \bigcup_{t\in [0,T]} \hat{C}_{t}^{\alpha}(\varphi, x), d_{\Lambda^{\alpha'}}\big)\right)$. Then for every continuous semimartingale $X$ starting in $x$ and fulfilling $\mathbb{E}\left[\exp(\zeta\|\hat{\mathbb{X}}^2_{[0,T]}\|^{\xi}_{\alpha, [0,T]})\right]< \infty$ with $\zeta>0$ and $\xi>2$, 
   it holds that for every $\epsilon, \delta>0$ there exists a signature portfolio $\pi^*(\tau, \cdot)$ of type I (type II) such that it holds for all $i\in\{1,\dots,d\}$
        $$ \mathbb{P}\left[\sup_{t\in[0,T]} \Big|\pi^i(\tau, \hat{X}_{[0,t]})-\pi^{*,i}(\tau, \hat{X}_{[0,t]})\Big| > \epsilon \right] < \delta.$$
	\end{theorem}
	
	\begin{proof} 
	    We present the proof for the case of portfolios of type $I$, however the proof is analogous for the case of type $II$. 
	    
	    Let us start by showing the first  statement of the theorem. By the universal approximation result (Theorem~\ref{thm:UAT_sig}), we know that for each $\{f^i\}_{1\leq i \leq d}$ there exits a linear function on the signature of $\hat{X}$ that approximates $f^i$ arbitrarily well, on compact sets. Those linear functions can be chosen as the portfolio \controllingfunctions $f^{i,\*}$ of a signature portfolio $\pi^*$. Hence, there exists for each $\epsilon' >0$ a signature portfolio $\pi^*$ of type $I$ such that for all $i \in \{1, ..., d\}$ it holds that 
	    \begin{equation}\label{eq:approx_f}
	         \sup_{(t, \hat{\mathbb{X}}^2_{[0,T]}(\omega)) \in[0,T]\times K} |f^i(\hat{\mathbb{X}}^2_{[0,t]}(\omega)) - f^{*,i}(\hat{\mathbb{X}}^2_{[0,t]}(\omega))|< \epsilon'
	    \end{equation}
	    for almost all $\omega \in \Omega$. 
	    
	    Now, we show that this translates to the portfolio weights. The statement is trivial for all $\omega \in \Omega$ where $\hat{\mathbb{X}}^2_{[0,T]} \notin K$. 
        
        On $\Omega^{(K)}:=\{\omega \mid \hat{\mathbb{X}}^2_{[0,T]}(\omega) \in K \}$, it is implied by~\eqref{eq:approx_f} that 
        \begin{align*} 
            &\sup_{t\in[0,T]} \left|\pi^i_t - \pi^{*,i}_t \right| = \sup_{t\in[0,T]} \Bigg|\tau^i_t\left( f^i(\hat{\mathbb{X}}^2_{[0,t]}) + 1 - \sum_{j=1}^d \tau_t^j f^j(\hat{\mathbb{X}}^2_{[0,t]}) \right) \\ &- \tau^i_t\left( f^{*,i}(\hat{\mathbb{X}}^2_{[0,t]}) + 1 - \sum_{j=1}^d \tau_t^j f^{*,j}(\hat{\mathbb{X}}^2_{[0,t]})\right) \Bigg|  \\  &= \sup_{t\in[0,T]} \left| \tau_t^i\left(f^i(\hat{\mathbb{X}}^2_{[0,t]}) - f^{*,i}(\hat{\mathbb{X}}^2_{[0,t]})\right) - \sum_{j=1}^d\tau_t^j\left(f^j(\hat{\mathbb{X}}^2_{[0,t]}) - f^{*,j}(\hat{\mathbb{X}}^2_{[0,t]})\right) \right| \\ &\leq \sup_{t\in[0,T]} \left| \tau_t^i\left(f^i(\hat{\mathbb{X}}^2_{[0,t]}) - f^{*,i}(\hat{\mathbb{X}}^2_{[0,t]})\right)\right| + \sum_{j=1}^d \sup_{t\in[0,T]} \left|\tau_t^j\left(f^j(\hat{\mathbb{X}}^2_{[0,t]}) - f^{*,j}(\hat{\mathbb{X}}^2_{[0,t]})\right) \right| \\ &\leq \, M \sup_{t \in [0,T]} \left|f^i(\hat{\mathbb{X}}^2_{[0,t]}) - f^{*,i}(\hat{\mathbb{X}}^2_{[0,t]}) \right| + M\sum_{j=1}^d \sup_{t \in [0,T]} \left|f^j(\hat{\mathbb{X}}^2_{[0,t]}) - f^{*,j}(\hat{\mathbb{X}}^2_{[0,t]}) \right| \\  &< M(d+1)\epsilon'
        \end{align*}
	    holds almost surely, where $M$ is the bound of $\tau$. The result follows by choosing $\epsilon'= \frac{\epsilon}{M(d+1)}$.

    The proof of the second approximation statement follows by similar arguments. Recall Theorem~\ref{thm:GUAT} and use the same reasoning as above to translate the approximation result from the portfolio controlling functions to the portfolio weights itself, i.e. we obtain that there exists for each $\epsilon' >0$ a signature portfolio $\pi^*$ of type $I$ such that for all $i \in \{1, ..., d\}$ it holds that 
    \begin{equation}\label{eq:guat_pi}
	    \sup_{t \in[0,T]} \frac{|\pi^i(\tau, \hat{X}_{[0,t]}) - \pi^{*,i}(\tau, \hat{X}_{[0,t]})|}{\exp(\zeta\|\widehat{\, \mathbb{X}^t}^2_{[0,T]}\|^{\xi}_{\alpha, [0,T]})}< \epsilon'
    \end{equation}
    almost surely. Therefore, 
    \begin{align*}
    \mathbb{P}\left[\sup_{t\in[0,T]} \Big|\pi^i(\tau, \hat{X}_{[0,t]})-\pi^{*,i}(\tau, \hat{X}_{[0,t]})\Big|> \epsilon \right] &\leq \epsilon^{-1}\cdot \mathbb{E}\left[\sup_{t\in[0,T]} \Big|\pi^i(\tau, \hat{X}_{[0,t]})-\pi^{*,i}(\tau, \hat{X}_{[0,t]})\Big|\right] \\
    &<\frac{\epsilon'}{\epsilon} \mathbb{E}\left[\sup_{t\in[0,T] }\exp(\zeta\|\widehat{\, \mathbb{X}^t}^2_{[0,T]}\|^{\xi}_{\alpha, [0,T]})\right],
    \end{align*}
    where the first inequality is due to Markov's inequality and the second one follows from \eqref{eq:guat_pi}. Note that by the properties of $\|\cdot \|_{cc}$ and $\| \cdot \|_{\alpha,  [0,T]}$ it holds that $\sup_{t\in [0,T]} \| \widehat{\, \mathbb{X}^t}^2_{[0,T]}\|_{\alpha,[0,T]} =   \| \hat{\mathbb{X}}^2_{[0,T]}\|_{\alpha,[0,T]}$.
    Hence the statement is obtained by setting $\delta = \frac{\epsilon'}{\epsilon} \mathbb{E}\left[\exp(\zeta\|\hat{\mathbb{X}}^2_{[0,T]}\|^{\xi}_{\alpha, [0,T]})\right]$ and choosing $\epsilon'$ accordingly.
    
    \end{proof}

    \begin{remark}
        Note, that by Corollary~\ref{cor:example_cont_func}, Theorem~\ref{thm:approx_by_sig_port} also ensures that signature portfolios approximate classical functionally generated portfolios arbitrarily well, on compact sets. In particular, this includes the functionally generated portfolios as introduced by Fernholz~\cite{Fern02} and those considered for functional portfolio optimization in~\cite{Wong}.
    \end{remark}

    \begin{remark}
        The requirements of Theorem~\ref{thm:approx_by_sig_port} on $f^i$ is a simple continuity condition and examples for such continuous path-functionals are provided in Example~\ref{cor:example_cont_func}.
    \end{remark}
	
	Having obtained a universal approximation theorem of path-functional portfolios, we will now apply it to obtain universal approximation results of the growth-optimal portfolio by signature portfolios in several markets. 

    \subsection{Universal Approximation of the Growth-Optimal Portfolio}
	In this subsection, we present a large class of markets where the growth-optimal portfolio can be regarded as a path-functional portfolio and show that it can therefore be  approximated by a signature portfolio. \\

    The following statements are to be understood in the market setting outlined in Definition~\ref{def:growth-opt} and we therefore assume that the corresponding necessary conditions are satisfied, in particular those ensuring the existence of the growth-optimal portfolio (Lemma~\ref{lem:growth-opt}).
    Moreover, recall also the definition of the compact set $\mathcal{K}_{\Lambda}$ given in \eqref{eq:Klambda}.

    \begin{theorem}\label{thm:approx_go}
        Consider the following class of markets with $d$ stocks:
        \begin{equation}\label{eq:SDEmarket}
            dS_t= \operatorname{diag}(S_t)(a(\hat{\mathbb{X}}^2_{[0,t]}) dt + \Sigma(\hat{\mathbb{X}}^2_{[0,t]}) dB_t),
        \end{equation}
        where the components of $a$, $\Sigma$ are non-anticipative path-functionals in $C(\mathcal{K}_{\Lambda}; \mathbb{R})$ and $X$ some $\mathbb{R}^n$-valued semimartingale with $n \in \mathbb{N}$. Then for any auxiliary portfolio $\tau$ there exists a signature portfolio of type $II$ which approximates the weights of the growth-optimal portfolio arbitrarily well on $\mathcal{K}_{\Lambda}$. Moreover, if the components of the auxiliary portfolio $\tau$ are deterministic, then there also exists a signature portfolio of type $I$ which approximates the weights of the growth-optimal portfolio almost surely arbitrarily well on $\mathcal{K}_{\Lambda}$.\\ 
        
        Alternatively, let $\frac{1}{3}< \alpha < \frac{1}{2}$, $0\leq \alpha' < \alpha$ and let $a, \Sigma$ be such that the components of $\pi^{(g)}$ are non-anticitpative path-functionals in $\mathcal{B}_{\psi}\left( \big( \bigcup_{t\in [0,T]} \hat{C}_{t}^{\alpha}(\varphi, x), d_{\Lambda^{\alpha'}}\big)\right)$.\footnote{Note that a sufficient condition for this to hold is if the components of $a$ are non-anticipative path-functionals in $\mathcal{B}_{\psi}\left( \big( \bigcup_{t\in [0,T]} \hat{C}_{t}^{\alpha}(\varphi, x), d_{\Lambda^{\alpha'}}\big)\right)$ and $\Sigma$ is constant.} Then for every continuous semimartingale $X$ staring in $x$ and fulfilling $\mathbb{E}\left[\exp(\zeta\|\hat{ \mathbb{X}}^2_{[0,T]}\|^{\xi}_{\alpha, [0,T]})\right]< \infty$ with $\zeta>0$ and $\xi>2$, it holds that for every $\epsilon, \delta>0$ there exists a signature portfolio $\pi^*(\tau, \cdot)$ for type $II$ for any auxiliary portfolio $\tau$ (or of type $I$ for a deterministic portfolio $\tau$) such that it holds for all $i\in\{1,\dots,d\}$
        $$ \mathbb{P}\left[\sup_{t\in[0,T]} \Big|\pi^{(g),i}(\tau, \hat{X}_{[0,t]})-\pi^{*,i}(\tau, \hat{X}_{[0,t]})\Big| > \epsilon \right] < \delta,$$
        where $\pi^{(g)}$ is the growth-optimal portfolio of the respective market.
        \end{theorem}

\begin{remark} \label{rem:pathdependentSDE}
Note that the semimartingale $X$ in \eqref{eq:SDEmarket} could be (a functions of) $S$, so that we are actually dealing with path-dependent SDEs. In the particular case where $a$ and $\Sigma$ are entire functions of the signature as introduced in \cite{Cuchiero_Svaluto-Ferro_Teichmann2023} we are then in the tractable setup of \emph{signature SDEs.}
\end{remark}

    \begin{proof}[Proof of Theorem~\ref{thm:approx_go}]
        Recalling Lemma~\ref{lem:growth-opt}, it is straightforward that for any auxiliary portfolio $\tau$, the growth-optimal portfolio $\pi^{(g)}$ is a path-functional portfolio of type $II$ with portfolio \controllingfunctions $$ f^{(g, II), i}= \pi^{(g), i}.$$ Likewise, $\pi^{(g)}$ is a path-functional portfolio of type $I$ for any deterministic auxiliary portfolio and portfolio \controllingfunctions $$f^{(g, I) , i}= \tau^i_{\cdot} \cdot \pi^{(g), i}.$$  By Lemma~\ref{lem:matrix_inv} and the form of $\pi^{(g),i}$  Theorem~\ref{thm:approx_by_sig_port} is applicable for $\pi^{(g),i}$ (viewed as a path-functional portfolio of type $I$ resp. type $II$). Hence, the statements follow. 
    \end{proof}
    
    \begin{remark}
        If $X=\mu$ or $X=S$, the growth optimal is also a path-functional portfolio of type $I$ with auxiliary portfolio $\tau=\mu$.
    \end{remark}

    \begin{remark}[Examples of Markovian Markets]
		We here give two prominent examples of markets, to which  Theorem~\ref{thm:approx_go} can be applied. 
		\begin{itemize}
			\item Black-Scholes Market: 
			$$\dd S_t =\operatorname{diag}(S_t)( a dt + \Sigma dB_t),$$
			where $a$ and $\Sigma$ are constant. 
			\item Volatility Stabilized Markets:
			$$\frac{\dd S^i_t}{S^i_t} = \frac{1+\alpha}{2} \frac{1}{\mu_t^i} dt + \sqrt{ \frac{1}{\mu_t^i}} dB_t^i \textrm{ for all } 1\leq i \leq d$$
			and $\alpha\geq 0$. Volatility stabilized markets are of great interest in SPT because they reflect the observation in real markets that smaller stocks tend to have greater volatility than larger stocks, see for example~\cite{Fernholz_Karatzas}. See also~\cite{Cuchiero_Poly_SPT, Karatzas_Vol_Stab} for further properties of volatility stabilized markets.
		\end{itemize}
	\end{remark}

    \subsection{A Class of Non-Markovian Markets where the Growth-Optimal Portfolio is a Signature Portfolio}
	Let us now turn to a class of (possibly) non-Markovian markets, where the growth-optimal portfolio can not only be approximated by a signature portfolio but \emph{is} a signature portfolio. 

	\begin{theorem}
		Let $X$ be a continuous semimartingale. For a market of $d$ stocks, consider the class of \emph{Sig-market} models
		$$ dS_t = \operatorname{diag}(S_t)(a(\hat{X}_{[0,t]}) dt + \Sigma dB_t)$$ where 
		$ a(\hat{X}_{[0,t]})_i = \sum_{ 0 \leq|I|\leq N} \alpha^{(i)}_I \langle e_I, \hat{\mathbb{X}}_t \rangle $ for $N\geq 0$, $\alpha^{(i)}_I \in \mathbb{R}$, $\Sigma$ is a constant $d\times m$ matrix and $B_t$ is an $m$-dimensional Brownian motion for $m\geq d$. 
  Then, the growth-optimal portfolio is for any auxiliary portfolio a signature portfolio of type $II$ and a signature portfolio of type $I$ for any constant and deterministic auxiliary portfolio $\tau$.
	\end{theorem}
	
	\begin{proof}
        The integrability and measurability assumptions required in Definition~\ref{def:growth-opt} and  Lemma \ref{lem:growth-opt} follow from the integrability and measurablilty of elements of the signature. Applying Lemma~\ref{lem:growth-opt}, it follows that the growth-optimal weights are linear functions on the signature of $\hat{X}$. This is for any auxiliary portfolio a signature portfolio of type $II$, or a signature portfolio of type $I$ for any constant and deterministic auxiliary portfolio $\tau$.
	\end{proof}
	
	\begin{remark}[Existence of Solutions to Sig-Markets]
 As addressed in Remark \ref{rem:pathdependentSDE},
	the existence of solutions is non-trivial for Sig-markets if $X=h(S)$, i.e. if the semimartingale of which we construct the signature is a function of the price process itself. In particular, if  $h$ is a real analytic function, then we deal with signature-SDEs in the spirit of \cite{Cuchiero_Svaluto-Ferro_Teichmann2023}. The study of existence of solutions to such equations goes beyond the subject of this paper, however, we give a simple example below, where the existence is guaranteed. To this end, let us denote by $$\mathcal{I}_1:=\left\{ I | I=(i_1, ..., i_m) \textrm{ for } 0\leq n\leq N \textrm{ and } i_2=i_3=\dots=i_n=1  \right\}$$
	Let $X=\log{S_t}$. If $\alpha_I^{(i)}=0$ for all $I\notin \mathcal{I}_1$, then 
	   \begin{equation}\label{eq:strong_sol} {dS_t}= \operatorname{diag}(S_t)\cdot \left(a((\widehat{\log{S}})_{[0,t]})  dt + \Sigma dB_t \right) \end{equation} admits a unique strong solution.
	      To see this, note that \eqref{eq:strong_sol} can be seen a part of a (larger) system of linear equations. Setting $Y=\log{S}$, $\hat{Y}= (t, Y)$ and denoting by $\hat{\mathbb{Y}}^N\lvert_{\mathcal{I}_1}= (\langle e_I, \hat{\mathbb{Y}}\rangle)_{I\in \mathcal{I}_1}$ the restriction to elements of the signature corresponding to multiindices in $\mathcal{I}_1$. Let us denote the vectorization of the truncated signature denoted by $\mathbf{vec}(\cdot)$ and fix a labelling function $\mathscr{L}: \mathcal{I}_1 \rightarrow \{1, \dots, | \mathcal{I}_1|\}$ such that $\mathbf{vec}(\hat{\mathbb{Y}}^N\lvert_{\mathcal{I}_1})_{\mathscr{L}(I)}= \langle e_I , \hat{\mathbb{Y}}^N\lvert_{\mathcal{I}_1} \rangle$. Moreover, for a multiindex $I=(i_1, \dots, i_{|I|})$, we denote by $I':=(i_1, \dots, i_{|I|-1})$ corresponding multiindex shortend by the last letter. Then said linear system of equations is
	   
	   \begin{equation*} d\left(\mathbf{vec}(\hat{\mathbb{Y}}^N\lvert_{\mathcal{I}_1})_t\right)= \left(b + A\cdot\mathbf{vec}(\hat{\mathbb{Y}}^N\lvert_{\mathcal{I}_1})_t \right)dt + \tilde{\Sigma} dB_t. \end{equation*}
	where 
    $\tilde{\Sigma}_{j,r}= \Sigma_{s, r}\delta_{j,\mathscr{L}(s)}$ for $1\leq j \leq |\mathcal{I}_1|$, $1 \leq r \leq m$ and for $k,l\in\{1, \dots, |\mathcal{I}_1|\}$
    \begin{align*} b_k= &\begin{cases} -\frac{1}{2} \Sigma^{\mathsf{T}}\Sigma+ a_{\emptyset}^{(i)} &\textrm{ if } k=\mathscr{L}(i), \, i\neq 1 \\
    1 &\textrm{ if }  k=\mathscr{L}(1)\\
    0 &\textrm{ else}\end{cases}\\  A_{k,l}= &\begin{cases} a^{(i)}_{\mathscr{L}(J)} &\textrm{ if } k=\mathscr{L}(i),\, l= \mathscr{L}(J), \, i\neq 1 \\
    1 &\textrm{ if } |I|\geq2, \, k=\mathscr{L}(I), \, l=\mathscr{L}(I') \\ 0 & \textrm{ else.}\end{cases}
    \end{align*}
    \end{remark}

    \begin{remark}[NUPBR in Sig-Markets]
        The NUPBR condition in the above Sig-markets holds. This follows from the existence of the growth-optimal portfolio by Lemma~\ref{lem:growth-opt},  Lemma~\ref{lem:numeraire_go} and Theorem~\ref{thm:NUPBR_numeraire}. 
    \end{remark}

	\section{Optimization Tasks for Linear Path-Functional Portfolios}\label{sec:optimizing_signature_port}
	
	We now want to study some tasks for portfolio optimization and their form for linear path-functional portfolios.
    To formulate them in the most general way, let us introduce what we call a \emph{universe} of stocks. 
	
		\begin{definition}[Universe]
	    Consider a market of $d$ stocks with capitatlization process $S= (S^1, ..., S^d)$. We call $\mathcal{U}\subseteq \{1,...,d\}$ a \emph{universe} of stocks with capitalization process $S^{\mathcal{U}} = (S^u)_{u\in \mathcal{U}}$. Moreover, we define the universe weights (resp.~universe portfolio) $\mu^{\mathcal{U} }$ to be 
	    \begin{equation*}
	        \mu^{\mathcal{U}, i}_t = \begin{cases} \frac{S^{\mathcal{U},i}_t}{\sum_{j\in\mathcal{U}}S^{\mathcal{U},j}_t} & \textrm{if }i \in \mathcal{U}  \\ 0 & \textrm{otherwise} \end{cases}.
	    \end{equation*}
	    Moreover, denote by $(W_t^{\mathcal{U}})_{t \in [0,T]}$ the wealth process of universe $\mathcal{U}$, which is given by $$W_t^{\mathcal{U}}= \frac{\sum_{i\in \mathcal{U}} S_t^i}{\sum_{i\in \mathcal{U}} S_{t_0}^i}.$$
	\end{definition}
	
	This notion of a universe of stocks is useful if one does not necessarily want to invest in all the stocks in the market but just in a subset of stocks. Of course, this can always be achieved by fixing certain weights of a portfolio to be zero, however using our notion of a universe of stocks is particularly useful, if one wants to compare the wealth process of a portfolio to the wealth process of the universe. Let us make this precise:
	
	\begin{cor}
		Take a universe $\mathcal{U} \subseteq \{1,...,d\}$ and consider a portfolio $\pi$ where $\pi^i \equiv 0$ for $i\not\in \mathcal{U}$. Then, the relative wealth process of $\pi$ with respect to $\mu^{\mathcal{U}}$ is given by 
		\begin{equation}
		    \dd \left(\frac{W^{\pi}_t}{W^{{\mathcal{U}}}_t}\right) = \left(\frac{W^{\pi}_t}{W^{{\mathcal{U}}}_t}\right)\sum_{i\in \mathcal{U}} \pi_t^i \frac{\dd \mu_t^{\mathcal{U},i}}{\mu_t^{\mathcal{U},i}}.
		\end{equation}
		Or, equivalently 
		\begin{equation}\label{eq:rel_val_process_universe}
		    \dd \log \left( \frac{W^{\pi}_t}{W^{{\mathcal{U}}}_t}\right) = \sum_{i\in \mathcal{U}} \pi_t^i \frac{\dd \mu_t^{\mathcal{U},i}}{\mu_t^{\mathcal{U},i}}- \frac{1}{2} \sum_{i,j\in \mathcal{U}} \frac{\pi^i_t}{\mu^{\mathcal{U}, i}_t}  \frac{\pi^j_t}{\mu^{\mathcal{U}, j}_t} \dd[\mu^{\mathcal{U}, i}, \mu^{\mathcal{U}, j}]_t.
		\end{equation}
	\end{cor}

	\begin{proof}
	    Using the definition of $\pi$ and $W^{\mathcal{U}}$, the statement follows directly from \eqref{eq:val_process_S} and applying It\^o's formula. 
	\end{proof}
	
	\begin{remark}
	    When we consider path-functional portfolios investing in a universe $\mathcal{U}$, it is convenient to construct them via the universe portfolio $\mu^{\mathcal{U}}$. Note that for path-functional portfolios $\pi^{(I)}$ of type $I$  it directly follows that $\pi^{(I),i}(\mu^{\mathcal{U}}, \cdot) =0$ for $i\not\in \mathcal{U}$. However, for  path-functional portfolios $\pi^{(II)}$ of type $II$, we have to require additionally  that $\pi^{(II),i}(\mu^{\mathcal{U}}, \cdot) =0$ for $i\not\in \mathcal{U}$ by setting $f^i\equiv0$ for $i\not\in \mathcal{U}$.
	\end{remark}
	
	\begin{remark}[Investments in Ranked Markets] Rank-based portfolios are of particular interest in SPT, see~\cite{Fern02}. We aim to therefore incorporate those into our considerations. Recall that our very basic assumptions on the financial market were that the stocks' capitalizations are positive continuous semimartingales. Considering a market $M$ which fulfills these assumptions, we know that there exists another market  which also fulfills these assumptions and for which  the stocks' capitalizations  are the same as the ranked capitalizations of the market $M$. This follows simply from the fact that the ranked capitalizations are again positive continuous semimartingales. However, 
    the assumption of NUPBR does not translate directly (see~\cite[Section 3.2]{Karatzas_Ruf} for a comment on the potential lack of a local martingale deflator in ranked markets and~\cite{Kardaras12} for a connection between the existence of a local martingale deflator and the NUPBR condition). For the optimization problems to be meaningful in the ranked market, we have to assume NUPBR again additionally. Hence, under the assumption of NUPBR for the ranked market weights, the optimization problems we present in the following hold also for investments in the ranked market.
	\end{remark}

	It  is sometimes useful to write optimization tasks in a vectorized form. To this end, let us introduce what we call a \emph{labelling function}.
	\begin{definition}
	    For a given linear path-functional portfolio $\pi$ investing in a universe of stocks $\mathcal{U}$, let $\mathcal{V}$ be the set of features, $|\mathcal{V}|$ the number of features and $|\mathcal{U}|$ the number of stocks in a universe.
		A bijective function $\mathscr{L}$ of the form
		\begin{align*}
			\mathscr{L}: \mathcal{U} \times \mathcal{V} &\longrightarrow \{1,...,|\mathcal{U}|\cdot|\mathcal{V}|\}\\
			(i, \nu) &\longmapsto \mathscr{L}(i,\nu)
		\end{align*}
		is called \emph{labelling function.}
	\end{definition}
	
    Before studying two concrete  optimization tasks for linear path-functional portfolios, let us make the following statement about a more general class of optimization tasks, which turn out to be \emph{quadratic optimization problems} for linear path-functional portfolios:

    \begin{prop}
        Let us denote by $\mathbb{S}^d$ the space of symmetric matrices of dimension $d\times d$.
        For a universe $\mathcal{U}$, consider a class of linear path-functional portfolios $\Pi$ (of either type) given by a feature space $\mathcal{V}$, a collection of feature maps $\{\phi^{\nu}\}_{\nu \in \mathcal{V}}$ consisting of non-anticipative path-functionals, a continuous semimartingale $X$ being the underlying process and an auxiliary portfolio $\tau$. Recall that the portfolio controlling functionals $\{f^i\}_{ i \in \mathcal{U}}$ of such portfolios are given by $f^i(t, X_{[0,t]}) = \sum_{\nu \in \mathcal{V}} l^i_{\nu} \phi^{\nu}(t, X_{[0,t]})$ for $i \in \mathcal{U}$. Then, any optimization problem of the form
        \begin{equation*}
            \inf_{\pi\in \Pi} \mathbb{E}\left[ \int_{t_0}^t \pi^{\mathsf{T}}_s C_s \pi_s \lambda_1(ds) - \int_{t_0}^t b_s^{\mathsf{T}} \pi_s \lambda_2(ds)\right]
        \end{equation*}
        where $b$, $C$ are stochastic processes with values in $\mathbb{R}^d$, $\mathbb{S}^d$ respectively and $\lambda^i$ are signed measures on $[t_0, t]$, is a quadratic optimization problem in the linear coefficients $\{l_{\nu}^i\}_{\nu \in \mathcal{V}, 1\leq i \leq d}$.
    \end{prop}

    \begin{proof}
        Recall the form of linear path-functional portfolios of type $I$ and $II$ respectively. Therefore, it is clear that the parameters $\{l_{\nu}^i\}$ appear at most in quadratic terms. Moreover, since they are constant in time and deterministic, they can be pulled outside of the integrals.  
        
    \end{proof}

	Let us now, present two prominent optimization tasks which fall into the above class, namely \emph{mean-variance and log-wealth optimization}. We assume that the following optimization problems are well-posed and that the expectations are finite.
	
	\subsection{Mean-Variance Optimization}\label{subsec:MV_optim_theory}
    
    In this optimization problem, we consider a strategy, where we invest at time $t$ with portfolio weights $\pi_t$ of a linear path-functional portfolio into a universe $\mathcal{U}$ and follow a buy-and-hold strategy over the time-span $\Delta$. The (relative) return of this strategy over the time-horizon $[t, t+\Delta]$ is given by 
	\begin{equation}
	    \mathcal{R}^{W, \pi}_{t, t+\Delta} = \frac{W^{\pi}_{t+\Delta}}{W^{\pi}_{t}} -1 = \sum_{j\in \mathcal{U}} \pi_{t}^j \frac{S_{t+\Delta}^j}{S_{t}^j} -1
	\end{equation}
	and respectively the relative return (relative to the universe portfolio $\mu^{\mathcal{U}}$)
	\begin{equation}
	    \mathcal{R}^{V, \pi}_{t, t+\Delta} = \frac{V^{\pi}_{t+\Delta}}{V^{\pi}_{t}} -1  = \sum_{j\in \mathcal{U}} \pi_{t}^j \frac{\mu_{t+\Delta}^{\mathcal{U},j}}{\mu_{t}^{\mathcal{U},j}} -1.
	\end{equation}
	
	In this setting, we can state the following corollary about the mean-variance optimization problem:
	\begin{cor}\label{cor:mv}
	    Let $\lambda \in \mathbb{R}_+$ be a risk-tolerance parameter, specifying the investor's risk-preferences. Then, for a linear path-functional portfolio $\pi(\tau, X)$ the following two optimization problems
        are convex quadratic optimization problems in the parameters $\{l_{\nu}^i\}_{i\in\mathcal{U}, \nu \in \mathcal{V}}$. More precisely, for the
	    \begin{enumerate} 
	    \item mean-variance optimization of returns we have
	    \begin{align*}
	        &\min_{\pi} \hspace{0.1cm} \mathrm{Var}\left(\mathcal{R}^{W, \pi}_{t, t+\Delta}\right) - \lambda \mathbb{E}\left(\mathcal{R}^{W, \pi}_{t, t+\Delta} \right)  \\
	        \Leftrightarrow &\min_{\{l_{\nu}^i\}} \hspace{0.1cm} \mathbf{l}^\mathsf{T} \mathrm{Var}\left( \mathbf{Y}^W \right) \mathbf{l}-  \mathbf{l}^\mathsf{T} \left(\lambda\mathbb{E}\left( \mathbf{Y}^W \right) - \boldsymbol{\sigma}^W \right)
	    \end{align*}
	    for $\mathbf{l}_{\mathscr{L}(i, \nu)} = l_{\nu}^i$ , $\boldsymbol{\sigma}^W_{\mathscr{L}(i,\nu)}= 2\mathrm{Cov}(\mathbf{Y}^W_{\mathscr{L}(i, \nu)},\mathcal{R}^{W,\tau}_{t, t+\Delta}) $,
	    \begin{equation*}
	     \mathbf{Y}^W_{\mathscr{L}(i, \nu)} =  \begin{cases} \tau_t^i \phi^{\nu}(t, X_{[0,t]})\left(\frac{S^i_{t+\Delta}}{S^i_t} - \mathcal{R}^{W,\tau}_{t, t+\Delta} -1\right) &\textrm{ if $\pi$ is of type $I$}  \\ \phi^{\nu}(t, X_{[0,t]})\left(\frac{S^i_{t+\Delta}}{S^i_t} - \mathcal{R}^{W,\tau}_{t, t+\Delta} -1\right) &\textrm{ if $\pi$ is of type $II$} \end{cases}
	    \end{equation*}
        and where $\mathrm{Var}(\mathbf{Y}^W)$ refers to the covariance matrix of $\mathbf{Y}^W$.
	    
	    \item mean-variance optimization of relative return we have
	    \begin{align*}
	       & \min_{\pi} \hspace{0.1cm} \mathrm{Var}\left(\mathcal{R}^{V, \pi}_{t, t+\Delta}\right) - \lambda \mathbb{E}\left(\mathcal{R}^{V, \pi}_{t, t+\Delta} \right) \\
	    \Leftrightarrow &\min_{\{l_{\nu}^i\}} \hspace{0.1cm} \mathbf{l}^\mathsf{T} \mathrm{Var}\left( \mathbf{Y}^V \right)\mathbf{l} -  \mathbf{l}^\mathsf{T} \left(\lambda\mathbb{E}\left( \mathbf{Y}^V \right)-\boldsymbol{\sigma}^V \right)
	    \end{align*}
	    for $\mathbf{l}_{\mathscr{L}(i, \nu)} = l_{\nu}^i$, $\boldsymbol{\sigma}^V_{\mathscr{L}(i,\nu)}= 2\mathrm{Cov}(\mathbf{Y}^V_{\mathscr{L}(i, \nu)},\mathcal{R}^{V,\tau}_{t, t+\Delta}) $ and
	    \begin{equation*}
	     \mathbf{Y}^V_{\mathscr{L}(i, \nu)} =  \begin{cases} \tau_t^i \phi^{\nu}(t, X_{[0,t]})\left(\frac{\mu^{\mathcal{U},i}_{t+\Delta}}{\mu^{\mathcal{U},i}_t} - \mathcal{R}^{V,\tau}_{t, t+\Delta}-1\right) &\textrm{ if $\pi$ is of type $I$}  \\ \phi^{\nu}(t, X_{[0,t]})\left(\frac{\mu^{\mathcal{U},i}_{t+\Delta}}{\mu^{\mathcal{U},i}_t} - \mathcal{R}^{V,\tau}_{t, t+\Delta}-1\right) &\textrm{ if $\pi$ is of type $II$} \end{cases}
	    \end{equation*}
         and where $\mathrm{Var}(\mathbf{Y}^V)$ refers to the covariance matrix of $\mathbf{Y}^V$.
	    \end{enumerate}
	 
	\end{cor}
	
	\begin{proof}
	    We proof the above statement for the case of $\mathcal{R}^{W, \pi}$ and for $\pi$ of type $I$. The proof of the other cases is completely analogous. 
	    \begin{align*} 
	        \mathcal{R}^{W, \pi}_{t, t+\Delta} &= \sum_{i\in \mathcal{U}} \tau_t^i \left( f^i(t, X_{[0,t]}) + 1 - \sum_{j\in \mathcal{U}} \tau_t^j f^j(t, X_{[0,t]}) \right) \left( \frac{S_{t+\Delta}^i}{S_{t}^i} \right) -1\\
	        & = \sum_{i\in \mathcal{U}} \tau^i_t f^i(t, X_{[0,t]}) \left( \frac{S^i_{t+\Delta}}{S_t^i} - \mathcal{R}^{W,\tau}_{t, t+\Delta} -1\right) + \mathcal{R}^{W,\tau}_{t, t+\Delta} +1 -1 \\ &= \mathbf{l}^{\mathsf{T}} \mathbf{Y}^W + \mathcal{R}^{W,\tau}_{t, t+\Delta} 
	    \end{align*}
	    Hence, 
	    \begin{align} 
	        \mathbb{E} \left(\mathcal{R}^{W, \pi}_{t, t+\Delta} \right)  = \mathbf{l}^{\mathsf{T}} \mathbb{E}\left(  \mathbf{Y}^W\right) +   \mathbb{E}\left( \mathcal{R}^{W,\tau}_{t, t+\Delta}\right). \label{eq:E_R}
	    \end{align}
	    Note that the  rightmost term of \eqref{eq:E_R} do not contribute in the optimization problem, as they are independent of $\mathbf{l}$. And, 
	    \begin{align} 
	        \mathrm{Var}\left(\mathcal{R}^{W, \pi}_{t, t+\Delta} \right)&= \mathrm{Var}\left(\mathbf{l}^{\mathsf{T}}\mathbf{Y}^W \right) + 2 \mathrm{Cov}\left(\mathbf{l}^{\mathsf{T}}\mathbf{Y}^W, \mathcal{R}^{W,\tau}_{t, t+\Delta} \right)+ \mathrm{Var}\left( \mathcal{R}^{W,\tau}_{t, t+\Delta} \right) \notag\\ &= \mathbf{l}^{\mathsf{T}}\mathrm{Var}\left(\mathbf{Y}^W \right)\mathbf{l} + \mathbf{l}^{\mathsf{T}} \boldsymbol{\sigma}^W + \mathrm{Var}\left( \mathcal{R}^{W,\tau}_{t, t+\Delta} \right),\label{eq:V_R}
	    \end{align}
	    where the rightmost term of \eqref{eq:V_R} does not contribute in the optimization problem.
	    The fact that the optimization problem is convex follows directly from the fact that the covariance matrix $\mathrm{Var}\left(\mathbf{Y}^W \right)$ is positive semi-definite.
	\end{proof}
    
	\begin{remark} 
	    Note that the expression of the optimization problem for the relative return simplifies considerably if we use a linear path-functional portfolio of type $I$ and choose the universe portfolio as the auxiliary portfolio, i.e. $\tau=\mu^{\mathcal{U}}$. To be precise, we obtain 
	    \begin{align*}
	       & \min_{\pi} \hspace{0.1cm} \mathrm{Var}\left(\mathcal{R}^{V, \pi}_{t, t+\Delta}\right) - \lambda \mathbb{E}\left(\mathcal{R}^{V, \pi}_{t, t+\Delta} \right) \\
	    \Leftrightarrow &\min_{\{l_{\nu}^i\}} \hspace{0.1cm} \mathbf{l}^\mathsf{T} \mathrm{Var}\left( \mathbf{Y}^V \right)\mathbf{l} - \lambda \mathbf{l}^\mathsf{T} \mathbb{E}\left( \mathbf{Y}^V \right)
	    \end{align*}
	    for $\mathbf{l}_{\mathscr{L}(i, \nu)} = l_{\nu}^i$ and $\mathbf{Y}^V_{\mathscr{L}(i, \nu)}= \phi^{\nu}(t, X_{[0,t]})(\mu^{\mathcal{U},i}_{t+\Delta} - \mu^{\mathcal{U},i}_t)$, since $\mathcal{R}^{V, \mu^{\mathcal{U}}}_{t, t+\Delta} =0$.
	\end{remark}

    \begin{remark}\label{rem:first_order_sol}
        If $\textrm{Var}(\mathbf{Y}^V)$ or $\textrm{Var}(\mathbf{Y}^W)$ respectively are almost surely invertible, that is the optimization problem is \emph{strictly} convex and hence admits a unique set of optimal parameters, then the optimal vector $\mathbf{l}^*$  can be found via first-order conditions, i.e. 
        $$ \mathbf{l}^* = \textrm{Var}(\mathbf{Y}^V)^{-1}\left( \lambda \mathbb{E}(\mathbf{Y}^V)- \mathbf{\sigma}^V\right) $$
        and likewise for the case of optimizing returns. 
        Note that the problem is never strictly convex for \emph{any} linear path-functional portfolio of type $I$, because due to the form of such portfolios, the weights are invariant under scalar shifts of the parameters $\{l_{\nu}^i\}_{\nu \in \mathcal{V}, 1\leq i\leq d}$. More precisely, for linear path-functional portfolios of type $I$, let $\pi^*(\tau, X)$ be optimal portfolio weights corresponding to an optimal set $\{l^{*, i}_{\nu}\}_{\nu \in \mathcal{V}, 1\leq i\leq d}$, then the set $\{l^{*,i}_{\nu}+\alpha\}_{\nu \in \mathcal{V}, 1\leq i\leq d}$ for $\alpha \in \mathbb{R}$ results in \emph{the same} portfolio weights $\pi^*(\tau, X)$. 

    \end{remark}

   \begin{remark}\label{rem:ergodic}
Let us here outline how the mean-variance optimization problem (of the relative returns) of Corollary \ref{cor:mv} can be related to a (model-free) real market situation where only the observation of one trajectory is available (see also Section \ref{sec:realdata}).
Tailored to Corollary \ref{cor:mv} we assume a discrete time situation with time steps of length $\Delta$ and suppose that the  conditional law of $\mu_{i\Delta}$ is time-homogeneous and depends on the last $n \wedge i$ values of $\mu$ for some $n \in \mathbb{N}$. More precisely, given $(\mu_{(i-n)\Delta \wedge 0}, \ldots, \mu_{(i-1)\Delta})=(x_0,\ldots,x_{(n-1) \wedge (i-1)}) \in (\Delta^d)^{n \wedge i}$, we know the conditional law $\rho(x_0,\ldots,x_{(n-1) \wedge (i-1)}, \cdot)$ of $\mu_{i\Delta}$.
Consider now a portfolio $\pi$ that is a non-anticipative path functional depending also only on the past $n$ values of $\mu$, i.e. a portfolio with rolling window in the spirit of Remark~\ref{rem:examples_linear_port}.
Then for $i \geq n$
\begin{align*}
\mathbb{E}\left[\frac{V^{\pi}_{i\Delta}}{V^{\pi}_{(i-1)\Delta}} \bigg| \mathcal{F}_{(i-1)\Delta} \right]&= \mathbb{E}[(R_{(i-1)\Delta, i\Delta}^{V,\pi} +1)| \mathcal{F}_{(i-1)\Delta} ]\\
&=\int_{\Delta^d }\sum_{j=1}^d\pi^j(x_0,\ldots,x_{n-1}) \frac{y^j}{x^j_{n-1}}\rho(x_0,\ldots,x_{n-1}, dy)
\end{align*}
and similarly for $i < n$, taking just the history starting from $0$ into account.
Assume now further that we have an invariant measure for the distribution of the rolling path segments of length $n$, which we  denote by $\rho$. Then under this invariant measure we have for each $i \geq n$
\begin{align*}
\mathbb{E}\left[\frac{V^{\pi}_{i\Delta}}{V^{\pi}_{(i-1)\Delta}}  \right] &=\int_{(\Delta^d)^n}\int_{\Delta^d }\sum_{j=1}^d\pi^j(x_0,\ldots,x_{n-1}) \frac{y^j}{x^j_{n-1}}\rho(x_0,\ldots,x_{n-1}, dy)\rho(dx_0, \ldots dx_{n-1})\\ &=:M^{\pi}
\end{align*}
and accordingly for $i < n$ assuming compatible invariant measures for the paths of shorter lengths.
Under these assumptions we then obtain the following ergodicity result
\[
\lim_{N \to \infty} \frac{1}{N} \sum_{i=1}^N \frac{V^{\pi}_{i\Delta}}{V^{\pi}_{(i-1)\Delta}}=
1+\lim_{N \to \infty} \frac{1}{N} \sum_{i=1}^N R_{(i-1)\Delta, i\Delta}^{V,\pi} =M^{\pi}.
\]
This can be deduced from 
Birkhoff’s ergodic theorem for discrete time Markov processes (see e.g.~[Theorem
2.2, Section 2.1.4]\cite{E:16} and also \cite[Section 3] {christa}
for the standard Markov situation with $n=1$).
Indeed, this becomes applicable by 
 identifying the states of the Markov process as the past path segements of length $n$, i.e. a state $z_{i\Delta}$ is given by $z_{i\Delta}=(\mu_{(i-n)\Delta}, \ldots, \mu_{i\Delta})$. The transition probabilities from $z_{i\Delta}$ to $z_{(i+1)\Delta}$ must then be chosen such that they assign $0$  probability to elements where there exists some $j \in \{1,\ldots, n-1\}$ with $z^j_{(i+1)\Delta} \neq z^{j+1}_{i\Delta}$. 

Clearly the above reasoning also holds true for 
$\textrm{Var}\left[V^{\pi}_{i\Delta}/ V^{\pi}_{(i-1)\Delta} \right]=\textrm{Var}[R_{(i-1)\Delta, i\Delta}^{V,\pi}].
$  Therefore, under such ergodicity assumptions all the expected values in Corollary \ref{cor:mv} can be approximated by time averages, which will be important in our numerical implementations as outlined in Section \ref{sec:realdata}.

    \end{remark}

	\subsection{Optimizing the (Expected) Log-(Relative)-Wealth}\label{sec:logutility}

	Similarly as the mean-variance portfolio optimization problem also the log-optimal portfolio over linear path functional portfolios can be found by solving a convex quadratic optimization task. We state it here in form of the relative log-utility optimization problem, i.e. the goal is to solve
 \[
\sup_{\pi}\mathbb{E} \left[ \log \left( \frac{W^{\pi}_{t}}{W^{{\mathcal{U}}}_{t}}\right) \right]=\sup_{\pi}\mathbb{E}[\log(V_t^{\pi})],
 \]
 where $V_t^{\pi}$ denotes again the relative wealth with respect to the universe $\mathcal{U}$.
 Note however that under similar ergodicity conditions as in Remark \ref{rem:ergodic} the use of time averages
 \[
\frac{1}{N} \log (V^{\pi}_{N\Delta})= \frac{1}{N} \sum_{i=1}^N \log\left( \frac{V^{\pi}_{i \Delta}}{V^{\pi}_{(i-1) \Delta}}\right)
 \]
is justified to approximate $\frac{1}{N}\mathbb{E}[\log(V_{N\Delta}^{\pi})]$.  
This can also be translated to continuous time, where
under appropriate assumptions (see e.g.~\cite[Section 4.2.3 and Theorem 4.9]{christa}) we also have that
 \[
 \lim_{t \to \infty}\frac{1}{t}\log(V_t^{\pi})=  \lim_{t \to \infty}\frac{1}{t}\mathbb{E}[\log(V_t^{\pi})],
 \]
whence 
\[
\frac{1}{t} \log(V_t^{\pi})\approx \frac{1}{t}\mathbb{E}[\log(V_t^{\pi})].
\]
In this case the expected values in the subsequent theorem, i.e.~$\mathbb{E}\left[\mathbf{Q}(t) \right]$
and  $\mathbb{E}\left[\mathbf{c}(t)\right]$ can therefore be just replaced by $\mathbf{Q}(t) $ and $\mathbf{c}(t)$ computed along the observed trajectory.
This will play again an important role in our empirical implementations.

	\begin{theorem}[Relative Log-Utility Optimization] \label{th:log_opt_task}
		Consider a universe $\mathcal{U}\subseteq \{1,...,d\}$ and a linear path functional portfolio $\pi(\mu^{\mathcal{U}}, X)$. Let $t_0\geq 0$ be the time at which we start investing. Denote by $\mathscr{L}$ an arbitrary but fixed labelling function. Then, for the (relative) log-utility optimization problem we have

		\begin{equation}\label{eq:opt_task_min}
			\sup_{\{l_{\nu}^{(i)}\}_{{i \in \mathcal{U}}, {\nu} \in \mathcal{V}}} \mathbb{E} \left[ \log \left( \frac{W^{\pi}_{t}}{W^{{\mathcal{U}}}_{t}}\right) \right] \Leftrightarrow \inf_{\mathbf{l}\in \mathbb{R}^{|\mathcal{U}|\cdot|\mathcal{V}|}} \frac{1}{2} \mathbf{l}^{\mathsf{T}} \mathbb{E}\left[\mathbf{Q}(t) \right]\mathbf{l} - \mathbb{E}\left[\mathbf{c}(t)^{\mathsf{T}}\right] \mathbf{l}.
		\end{equation}
	
		where $\mathbf{l}_{\mathscr{L}(i, \nu)} = l_{\nu}^{i}$ for $i \in \mathcal{U}$, $\nu \in \mathcal{V}$ and $$(\mathbf{c}(t))_{\mathscr{L}(i, \nu)}= \int_{t_0}^t 
		    \gamma_s^i \phi^{\nu}( s, X_{[0,s]}) \dd \mu^{\mathcal{U},i}_s$$ 
		    
		    $$ (\mathbf{Q}(t))_{\mathscr{L}(i, \nu), \mathscr{L}(j, \rho)}= \int_{t_0}^t  \gamma_s^i\gamma_s^j\phi^{\nu}( s, X_{[0,s]}) \phi^{\rho}( s, X_{[0,s]}) \dd[\mu^{\mathcal{U},i},\mu^{\mathcal{U},j}]_s$$
        for $$\gamma_s^i = \begin{cases} 1 &\textrm{ if } \pi \textrm{ is of type $I$} \\  \left(\mu^{\mathcal{U}, i}_s\right)^{-1} &\textrm{ if } \pi \textrm{ is of type $II$}\end{cases}. $$
		Moreover, it is a convex quadratic optimization problem. 
	\end{theorem}

    \begin{remark} \label{rem:simplification}
    
    Let us highlight how the expression of $\mathbf{c}$, $\mathbf{Q}$ simplify in the case of a signature portfolio of type $I$, $t_0=0$, $\mathcal{U}=\{1,\dots,d\}$, $X=\mu$ and $\mu$ is such that it holds for some $N_0>0$
    $$ \dd[\mu^i,\mu^j]_t= \sum_{0\leq|I|\leq N_0} \alpha^{i,j}_I \langle e_I, \hat{\bbmu}_t\rangle,$$
    where $\alpha^{i,j}_I \in \mathbb{R}$ for all $i,j\in \{1,\dots,d\}$, $0\leq|I|\leq N_0$. Then we obtain
    $$(\mathbf{c}(t))_{\mathscr{L}(i, I)} = \langle e_I\otimes e_i, \hat{\bbmu}_t\rangle$$
    $$(\mathbf{Q}(t))_{\mathscr{L}(i,I), \mathscr{L}(j,J)}= \sum_{0\leq|K|\leq N_0} \alpha^{i,j}_K \langle (e_I\shuffle e_J)\otimes e_K, \hat{\bbmu}_t\rangle \, .$$

    More generally, consider a signature portfolio $\pi(\mu, \breve{\mu})$ of type $I$ where $\breve{\mu}_t:=(t, \mu_t, [\mu, \mu]_t)$ for all $t\in[0,T]$ in a vectorized form, i.e. fix a labelling function $\kappa(\cdot, \cdot)$ such that $\breve{\mu}^{\kappa(i,j)}= [\mu^i, \mu^j]$ for $1\leq i, j\leq d$ and set $t_0=0$, $\mathcal{U}=\{1,\dots,d\}$. Then the elements of $\mathbf{c}$ and $\mathbf{Q}$ are linear functions on the signature of $\breve{\mu}$, namely 
    $$(\mathbf{c}(t))_{\mathscr{L}(i, I)} = \langle e_I\otimes e_i, \breve{\bbmu}_t\rangle$$
    $$(\mathbf{Q}(t))_{\mathscr{L}(i,I), \mathscr{L}(j,J)}= \langle (e_I\shuffle e_J)\otimes e_{\kappa(i,j)}, \breve{\bbmu}_t\rangle \, .$$
    \end{remark}
    
	\begin{remark}
	    Note that although we start trading at time $t_0\geq 0$ we feed to the feature maps $\phi^{\nu}$ the semimartingale $X$ starting at time $0$. The interpretation of this is that we start ``observing'' the market at time $0$ and but only start trading at time $t_0\geq 0$. This time-difference between observing and starting to invest can be important in  non-Markovian settings.
	\end{remark}

	Before we prove Theorem~\ref{th:log_opt_task}, let us state the following lemma.
	
	\begin{lemma}\label{lem:rel_val_process}
		For a path-functional portfolio $\pi(\mu^{\mathcal{U}}, {X})$ it holds that 
		\begin{align}\label{eq:rel_val_process}
		\begin{split}
			\log\left( \frac{W^{\pi}_{t}}{W^{{\mathcal{U}}}_{t}}\right) =& \sum_{i \in \mathcal{U}} \int_{t_0}^t \gamma_s^i f^i(s, {X}_{[0,s]}) d \mu_s^{\mathcal{U},i} \\ &- \frac{1}{2} \sum_{i,j \in \mathcal{U}} \int_{t_0}^t \gamma_s^i \gamma_s^j f^i(s, {X}_{[0,s]}) f^j(s, {X}_{[0,s]}) d[\mu^{\mathcal{U},i},\mu^{\mathcal{U},j}]_s,
		\end{split}
		\end{align}
		where $$\gamma_s^i = \begin{cases} 1 &\textrm{ if } \pi \textrm{ is of type $I$} \\  \left(\mu^{\mathcal{U}, i}_s\right)^{-1} &\textrm{ if } \pi \textrm{ is of type $II$}\end{cases}. $$

	\end{lemma}
	
    We refer the reader to Appendix~\ref{app:optimizing_signature_port} for a proof of Lemma~\ref{lem:rel_val_process}. 
	
	\begin{proof}[Proof of Theorem~\ref{th:log_opt_task}]
        To show the equivalence in \eqref{eq:opt_task_min}, we use Lemma~\ref{lem:rel_val_process} and make use of the fact that we are using linear path-functional portfolios.
		\begin{align}
			\log\left( \frac{W^{\pi}_{t}}{W^{{\mathcal{U}}}_{t}}\right)  =& \sum_{i \in \mathcal{U}} \sum_{\nu \in \mathcal{V}} l_{\nu}^{i} \int_{t_0}^t \gamma_s^i \phi^{\nu}(s, X_{[0,s]}) \dd \mu_s^{\mathcal{U},i} \\&- \frac{1}{2} \sum_{i,j \in \mathcal{U}} \sum_{\nu, \rho \in \mathcal{V}} l_{\rho}^{j} l_{\nu}^{i} \int_{t_0}^t \gamma_s^i\gamma_s^j \phi^{\nu}(s, X_{[0,s]}) \phi^{\rho}(s, X_{[0,s]})   \dd[\mu^{\mathcal{U},i},\mu^{\mathcal{U},j}]_s
		\end{align}
		Hence, we have shown that 
		\begin{equation}
			\log\left( \frac{W^{\pi}_{t}}{W^{{\mathcal{U}}}_{t}}\right)=  \mathbf{c}(t)^{\mathsf{T}} \mathbf{l}  - \frac{1}{2}\mathbf{l}^{\mathsf{T}} \mathbf{Q}(t) \mathbf{l}
		\end{equation}
		from which the equivalence in  \eqref{eq:opt_task_min} follows directly. 
		
	    To show that the optimization problem given in \eqref{eq:opt_task_min} is convex, we show that the matrix $\mathbf{Q}(t)$ is positive semidefinite. Consider the stochastic process $Y_t= \sum_{j=1}^d \int_{t_0}^t \gamma_s^jf^j(s, X_{[0,s]}) d\mu^{\mathcal{U},j}_s$ whose quadratic variation  $[Y,Y]_t $ is given by it holds that 
		\begin{align*}
			[Y,Y]_t = \sum_{i=1}^d \sum_{j=1}^d \int_{t_0}^t \gamma_s^i \gamma_s^j f^i(s, X_{[0,s]}) f^j(s, X_{[0,s]}) d[\mu^{\mathcal{U},i}, \mu^{\mathcal{U},j}]_s .
		\end{align*}
		On the other hand $$\mathbf{l}^T\mathbf{Q}(t) \mathbf{l} = \sum_{i=1}^d \sum_{j=1}^d \int_{t_0}^t \gamma_s^i \gamma_s^j f^i(s, X_{[0,s]}) f^j(s, X_{[0,s]}) d[\mu^{\mathcal{U},i}, \mu^{\mathcal{U},j}]_s.$$
		Hence, 
		$\mathbf{l}^T\mathbf{Q}(t) \mathbf{l} = [Y,Y]_t \geq 0 \hspace{0.2cm} \textrm{for all } t\in [t_0, T] \textrm{ and for all } \mathbf{l}\in \mathbb{R}^{(|\mathcal{U}|\cdot  |\mathcal{V}|)}$, proving that $\mathbf{Q}(t)$ is positive semidefinite for every $t \in [t_0, T]$. 
	\end{proof}
    
    \begin{remark}
    If the log-utility maximization problem is finite, it holds that
        \begin{equation}\label{eq:opt_task_max}
			\sup_{\{l_{\nu}^{(i)}\}_{{i \in \mathcal{U}}, {\nu} \in \mathcal{V}}} \mathbb{E} \left[ \log \left( {W^{\pi}_t}\right) \right] 
			\Leftrightarrow \sup_{\{l_{\nu}^{(i)}\}_{{i \in \mathcal{U}}, {\nu} \in \mathcal{V}}} \mathbb{E} \left[ \log \left( \frac{W^{\pi}_t}{W^{{\mathcal{U}}}_t}\right) \right].
		\end{equation} 
 
        This simply follows from the following equivalences
	    \begin{align*}
	        \sup_{\{l_{\nu}^{(i)}\}_{{i \in \mathcal{U}}, {\nu} \in \mathcal{V}}} \mathbb{E} \left[ \log \left( {W^{\pi}_t}\right) \right] 
			& \Leftrightarrow  \sup_{\{l_{\nu}^{(i)}\}_{{i \in \mathcal{U}}, {\nu} \in \mathcal{V}}} \mathbb{E} \left[ \log \left( {W^{\pi}_t}\right) - \log \left({W^{{\mathcal{U}}}_t}\right) \right] \\ & \Leftrightarrow \sup_{\{l_{\nu}^{(i)}\}_{{i \in \mathcal{U}}, {\nu} \in \mathcal{V}}} \mathbb{E} \left[ \log \left( \frac{W^{\pi}_t}{W^{{\mathcal{U}}}_t}\right) \right].
	    \end{align*}            
    \end{remark}

    \begin{remark}
        Similar to Remark~\ref{rem:first_order_sol}, the optimal vector $\mathbf{l}^*$ is given by $$\mathbf{l}^*= 2\mathbb{E}[\mathbf{Q}(t)]^{-1}\mathbb{E}[\mathbf{c}(t)]\, ,$$ if $\mathbb{E}[\mathbf{Q}(t)]$ is almost surely invertible. As mentioned this can never hold for linear path-functional portfolios of type $I$. 
    \end{remark}

	\subsection{Constraints and Regularization}
	Given an optimization task involving a linear path-functional portfolio which is convex and quadratic, we can make the following statements about adding constraints or  regularization:
	\begin{prop}\label{prop:mod_optim}
    	Given an optimization problem 
			$$\min_{\mathbf{l}} \hspace{0.2cm} -\mathbf{l}^T\mathbf{c}(t) + \frac{1}{2} \mathbf{l}^T \mathbf{Q}(t)\mathbf{l}$$
		which is convex quadratic, it remains convex quadratic under the following modifications.
		\begin{enumerate}
			\item {\bf Bounds Constraints}:
			 $|l_{\nu}^{i}|^2 \leq b_{\mathscr{L}(i, \nu)}$ for $b_{\mathscr{L}(i, \nu)}\geq 0$ for all $i\in \mathcal{U}$, $\nu \in \mathcal{V}$. 
			\item {\bf L2-Regularization}: 
		we can add an $L^2$-regularization term of the form $\gamma \mathbf{l}^T \mathbf{l}$ such that the optimization problem becomes
			\begin{align} 
				\min_{\mathbf{l}} \hspace{0.2cm} -\mathbf{l}^T\mathbf{c}(t) + \frac{1}{2} \mathbf{l}^T \mathbf{Q}(t)\mathbf{l} + \gamma \mathbf{l}^T \mathbf{l}.
			\end{align}
			Note that, this corresponds to a new matrix $\tilde{\mathbf{Q}}(t)= \mathbf{Q}(t)+ \gamma \mathbb{I}$, where $\mathbb{I}$ is the identity matrix. 
		\end{enumerate}
		Moreover, the above modifications can be combined and the optimization task remains convex.
	\end{prop}
	
	\begin{proof}
		
		\begin{enumerate}
			\item $|l_{\nu}^{i}| \leq b_{\mathscr{L}(i,\nu)}$ for $b_{\mathscr{L}(i,\nu)}\geq 0$ corresponds to the constraint $$\mathbf{l}^{\mathsf{T}} \delta_{\mathscr{L}(i,\nu), \mathscr{L}(i,\nu)} \mathbf{l} - b^2_{\mathscr{L}(i,\nu)} \leq 0,$$ where $\delta_{\mathscr{L}(i,\nu), \mathscr{L}(i,\nu)}$ is the matrix $$\delta_{jk}= \begin{cases} 1 & \textrm{if }  j=k=\mathscr{L}(i, \nu) \\ 0 & \textrm{otherwise} \end{cases}, $$ which is obviously positive semi-definite.
 			
			\item Since $\mathbf{Q}(t)$ and $\mathbb{I}$ are positive semi-definite and the sum of positive semi-definite matrices is positive semi-definite, $\tilde{\mathbf{Q}}(t)$ is again positive semi-definite.
		\end{enumerate}
	\end{proof}
   
    \begin{remark}[Long-Only and Leverage Constraints]
    We can add several forms of short-selling constraints under which the convex quadratic optimization problems remain convex and quadratic for linear path-functional portfolios. However, we would like to emphasise that these constraints are only effective during training and there is no theoretical guarantee that the portfolio-weights respect such constraints during the out-of-sample period. The reason for this is that different dynamics of the underlying process $X$ may lead to violation of the constraints during the testing period. 
    \begin{enumerate}
        \item {\bf Long-Only}: $\pi^i(t) \geq 0$ for all $i\in \mathcal{U}$;
        \item {\bf Bounds for Short-Selling}:$-a \leq \pi^i(t) \leq b $ for all $i\in \mathcal{U}$ and some $a, b \in \mathbb{R}$;
        \item {\bf Leverage constraints} (see~\cite{Leverage}): $\pi^i(t) \geq 0$ for all $i\in \mathcal{U}_L \subseteq \mathcal{U}$ i.e.,~impose long-only constraints on a subset $\mathcal{U}_L$ of stocks and impose a leverage constraint $\sum_{i \in \mathcal{U}_L} \pi^i_t \leq b$ for some $b>0$.
                
    \end{enumerate}

    \end{remark}

\begin{prop}[Long-Only Path-Functional Portfolios] 
When choosing $\tau$ to be long-only, we can formulate the following sufficient conditions for a path-functional portfolio to be long-only which increase in strength in the sense that $(iv) \Rightarrow(iii) \Rightarrow (ii) \Rightarrow (i)$. The following are required to hold $\mathbb{P}$-a.s. for all $t\in [0,T]$. 
\begin{enumerate}[label=(\roman*)]
	\item $|f^i(t, X_{[0,t]})- f^{max}(t, X_{[0,t]})| \leq 1 $ where $f^{max}(t, X_{[0,t]}):= \max\{f^1(t, X_{[0,t]}), \dots, f^d(t, X_{[0,t]})\}$ for all $i\in\{1,...,d\}$ 
	\item $ |f^i(t, X_{[0,t]}) -f^j(t, X_{[0,t]})| \leq 1$ for all $i,j\in\{1,...,d\}$
	\item $0\leq f^i(t, X_{[0,t]}) \leq 1$ for all $i\in\{1,...,d\}$
    \item $ f^i(t, X_{[0,t]}) \in [\alpha, \alpha+1]$, $\alpha\in \mathbb{R}$ for all $i\in\{1,...,d\}$ 
\end{enumerate}

\end{prop}

\begin{proof}
By definition, we know that for each $i \in \{1,...,d\}$ for path-functional portfolios of type $I$ it holds:
\begin{align*}
	\pi^{i}(t) \geq 0 \Leftrightarrow \Big(f^i(t, X_{[0,t]})+ 1 - \sum_{j=1}^d f^j(t, X_{[0,t]}) \mu_t^j\Big)\geq 0 \Leftrightarrow 
	(f^i(t, X_{[0,t]})- \sum_{j=1}^d f^j(t, X_{[0,t]}) \mu_t^j)\geq -1.
\end{align*}
One sufficient condition for this to hold is
\begin{equation*}
	|f^i(t, X_{[0,t]})- f^{max}(t, X_{[0,t]})| \leq 1
\end{equation*}
where $f^{max}(t, X_{[0,t]}):= \max\{f^1(t, X_{[0,t]}), \dots, f^d(t, X_{[0,t]})\}$. That is because 
\begin{align*}
	\Big(f^i(t, X_{[0,t]}) - \sum_{j=1}^d f^j(t, X_{[0,t]}) \mu_t^j\Big)&\geq f^i(t, X_{[0,t]})- f^{max}(t, X_{[0,t]})\sum_{j=1}^d \mu_t^j\\
 &= f^i(t, X_{[0,t]})- f^{max}(t, X_{[0,t]}).
\end{align*}
Similarly, for path-functional portfolio of type $II$ we find:
\begin{align*}
\pi^i(t)\geq 0 \Leftrightarrow f^i(t,X_{[0,t]}) - \tau^i_t \sum_{j=1}^d f^j(t,X_{[0,T]}) \geq -\tau_t^i \geq -1
\end{align*}
and also
\begin{align*}
f^i(t,X_{[0,t]}) - \tau^i_t \sum_{j=1}^d f^j(t,X_{[0,T]}) &\geq f^i(t,X_{[0,t]}) - d\cdot\tau^i_t  f^{max}(t,X_{[0,T]})\\&\geq f^i(t, X_{[0,t]})- f^{max}(t, X_{[0,t]}).
\end{align*}

\end{proof}

The previous sufficient conditions are formulated on the level of portfolio controlling functions. How they translate to the optimization parameters $\{l_{\nu}^i\}_{\nu \in \mathcal{V}, 1\leq i\leq d}$ for linear path-functional portfolios depends on the form of the feature maps and on the properties of $X$. The following example holds for a particular case of signature-portfolios with $X=\mu$. 

\begin{example}
    Let us choose $X=\mu$, $\tau=\mu$ and a linear path-functional portfolio of either type $I$ or $II$ with feature maps $$\phi^{i_1\shuffle \dots \shuffle i_m}(t, \mu_{[0,t]})= \langle e_{i_1\shuffle \dots \shuffle i_m}, \hat{\bbmu}_t\rangle = \hat{\mu}^{i_1}_t  \cdots  \hat{\mu}^{i_m}_t$$
    for $0\leq m\leq N$ and $i_1, \dots, i_m \in \{1, \dots, d+1\}$. Hence, by the form of the above the feature maps are simply polynomials and it holds that 
    $$0\leq \phi^{i_1\shuffle \dots \shuffle i_m}(t, \mu_{[0,t]})\leq 1 \textrm{ for all } t\in [0,T] \textrm{ and all } i_1, \dots, i_m \in \{1, \dots, d+1\}, \, m\leq N$$
    assuming w.l.o.g that the time-augmentation $\varphi(t)\in [0,1]$ 
     for all $t\in [0,T]$ (recalling that $\hat{\mu}_t^1= \varphi(t)$). In this case it follows that
\begin{align}
	|f^j(t, X_{[0,t]})- f^k(t, X_{[0,t]})| = | \sum_{m=0}^N \sum_{1\leq i_1 \leq ... \leq i_m\leq d+1 }( l^{(j)}_{i_{1} \shuffle ... \shuffle \ {i_{m}}}-l^{(k)}_{i_{1} \shuffle ... \shuffle \ {i_{m}}}) \langle e_{i_1 \shuffle ... \shuffle i_m}, \bbmu \rangle_t | \\ \leq 
	\sum_{m=0}^N \sum_{1\leq i_1 \leq ... \leq i_m\leq d+1 } | l^{(j)}_{i_{1} \shuffle ... \shuffle \ {i_{m}}}-l^{(k)}_{i_{1} \shuffle ... \shuffle \ {i_{m}}}| \cdot|\langle e_{i_1 \shuffle ... \shuffle i_m}, \bbmu \rangle_t |  \\ \leq \sum_{m=0}^N \sum_{1\leq i_1 \leq ... \leq i_m\leq d+1 } | l^{(j)}_{i_{1} \shuffle ... \shuffle \ {i_{m}}}-l^{(k)}_{i_{1} \shuffle ... \shuffle \ {i_{m}}}|
\end{align}
We denote by $\tilde{N}$ the number of words of length $0\leq m \leq N$ with letters in $\{1,...,d+1\}$, where the order does not matter, more precisely, 
$$\tilde{N}= 1+ \sum_{m=1}^N \binom{d+m}{m}.$$
Then, $| l^{(j)}_{i_{1} \shuffle ... \shuffle \ {i_{N}}}-l^{(k)}_{i_{1} \shuffle ... \shuffle \ {i_{N}}}| \leq 1/\tilde{N}$
implies that
\begin{equation}
 \sum_{N=0}^n \sum_{1\leq i_1 \leq ... \leq i_N\leq d } | l^{(j)}_{i_{1} \shuffle ... \shuffle \ {i_{N}}}-l^{(k)}_{i_{1} \shuffle ... \shuffle \ {i_{N}}}| \leq 1.
\end{equation}
And hence the quadratic constraint, preserving the convex quadratic structure of the optimization problem,
$$(l^{(j)}_{i_{1} \shuffle ... \shuffle \ {i_{N}}}-l^{(k)}_{i_{1} \shuffle ... \shuffle \ {i_{N}}})^2 \leq \frac{1}{\tilde{N}^2}$$
is sufficient to ensure that the signature portfolios are long-only. Of course this is a
particular case of a signature portfolio based on multivariate polynomials whose coefficients should not be too far away from each other.
\end{example}

 \section{Transaction Costs for Portfolios with Short-Selling and Without Bank Account}\label{sec:transcost}
	
	When trading in the presence of (proportional) transaction costs, we need to consider re-balancing at discrete times. Let $t$ be an arbitrary but fixed time at which we re-balance. At time $t^-$ (just before re-balancing) our portfolio has weights
	$\pi_{t^-}$ and a wealth of $W^{\pi}_{t^-}$. We want to re-balance to some fixed target weights $\pi_t$. This situation has been studied in~\cite{Ruf_Xie_19} for the case of long-only weights. However, we allow for short-selling in our portfolios which makes the problem significantly more involved, as we illustrate below. \\
	
    As the transaction costs are not paid on the portfolio weights, but on the dollar amounts invested in a stock, we define $\psi^i_{t^-}$, $\psi^i_t$ to be the dollar amounts invested in stock $i$ directly before and after re-balancing. Note, that of course it holds 
	\begin{equation*}
	    \psi^i_t= \pi^i_t  \cdot W^{\pi}_t
	\end{equation*}
	and the transaction costs to be paid are
	\begin{equation*}
	    TC_t= c\sum_{i=1}^d | \psi^i_t- \psi^i_{t^-} | 
	\end{equation*}
	where $c$ is the percentage of transaction costs to be paid. The difficulty lies in solving for $\psi^i_t$, since the amount which we can invest in stock $i$ depends on the amount of transaction costs to be paid \emph{and vice-versa}. This difficulty arises from the fact, that we do not have a bank-account. Since we know the portfolio weights at each time, we need to solve for $W^{\pi}_t$.  We can use the self-financing identity which requires 
	\begin{equation}\label{eq:TC_1}
	    W^{\pi}_t = W^{\pi}_{t^-}- TC_t = W^{\pi}_{t^-} - c\sum_{i=1}^d | \psi^i_t- \psi^i_{t^-} |. 
	\end{equation}
	Moreover, it is convenient to define $\alpha(t):= \frac{W^{\pi}_t}{W^{\pi}_{t^-}}$ and rewrite \eqref{eq:TC_1} as 
	\begin{equation}\label{eq:TC_2}
	    1 - \alpha_t= c\sum_{i=1}^d | \alpha_t\pi^i_t- \pi^i_{t^-} |. 
	\end{equation}
	Hence, we reduced the problem to solving for $\alpha_t$. Since we consider an arbitrary but fixed re-balancing time $t$, let us set $\alpha:=\alpha_t$ in the following. We can make the following statement about solutions for $\alpha$:
	
	\begin{prop}\label{prop:transcosts}
	    Let $\alpha^*$ be a solution to the equation
	        \begin{equation*}
	            1 - \alpha= c\sum_{i=1}^d | \alpha\pi^i_t- \pi^i_{t^-} | \label{eq:alpha}
	        \end{equation*}
	   and denote by $L(\alpha)= 1- \alpha$ and $R(\alpha)= c\sum_{i=1}^d | \alpha\pi^i_t- \pi^i_{t^-} |$. We can make the following statements about $\alpha^*$:
	   \begin{enumerate} 
	    \item If $\sum_{i=1}^d |\pi^i_{t^-}| < \frac{1}{c}$ a unique solution $\alpha^* \in [0,1]$ exists.\\
	    \item If $\sum_{i=1}^d |\pi^i_{t^-}| \geq \frac{1}{c}$ there \emph{may not exist a solution}, or at least none for $\alpha \in [0,1]$. If a solution exists, $\alpha^*\leq 1$ and it \emph{may not be unique}.
	   \end{enumerate}
	\end{prop}
	
	\begin{remark}
	   If we have long-only portfolios (and if $c< 1$ i.e. less than 100\% transaction costs) only the first case of Proposition~\ref{prop:transcosts} is relevant. Hence, in that case one always has a unique solution $\alpha^*\in[0,1]$. However, when we allow for short-selling, the second case of Proposition~\ref{prop:transcosts} becomes important. 
	\end{remark}
	
	We give a proof of Proposition~\ref{prop:transcosts} in Appendix~\ref{app:transcosts} and move directly to the interpretation of situation in the second case of Proposition~\ref{prop:transcosts}. 
	
	\begin{remark}[Interpretation of No or Non-Unique Solutions for $\alpha$]
	    We want to give some interpretation of the three possible outcomes of the second case of Proposition~\ref{prop:transcosts}. 
	    \begin{enumerate}[label=(\alph*)]

	            \item \emph{no solution}: The trade is infeasible. In order to pay the transaction costs, we have to reduce the dollar amounts we invest. However, reducing the dollar amounts $\psi_t$ leads to more transaction costs (because the difference to the previous investment amount grows, i.e. $|\psi_t- \psi_{t^-}|$ increases). If the transaction costs grow faster than the money we gain from reducing the dollar amount $\psi_t$, the trade is infeasible.

	            \item \emph{no solution in $[0,1]$}: Again, if we need to pay transaction costs without a bank-account we have to free-up the money by selling stocks. Unlike case a) this might be feasible, but the amount of transaction costs to be paid is more than the total wealth of the portfolio. This leads to a solution with negative $\alpha$.

	            \item \emph{no unique solution in $[0,1]$}: Here, we are in the case where the transaction costs do not exceed our total wealth. Since reducing the dollar amount leads to an increase in transaction costs, there might be several strategies with respect to dollar amounts invested leading to the desired portfolio weights after transaction costs. For example, one could do a more expensive trade (smaller $\psi_t$) and pay more transaction costs or a less expensive trade (larger $\psi_t$) where one pays less transaction costs. 
	        \end{enumerate}
	        
	        In practice, we treat the above cases as follows. We consider cases a) and b) to lead to ruin, because both are caused by strategies which are infeasible given our wealth. If one of those cases occurs, we set our wealth to zero and terminate the investment. In case c) we choose the solution with highest $\alpha$. This is what any reasonable investor would do: choose the trade with the least transaction costs.
	\end{remark}

	\section{Numerical Results}\label{sec:numerical_results}

    In this section we present our numerical results using simulated and real market data.\footnote{ The code corresponding to this section is available at \url{https://github.com/janka-moeller/Sig-SPT}}
	
	\subsection{Simulated Data}
	In this subsection, we aim to learn the  signature portfolios which maximize $\mathbb{E}[\log \frac{W_T^{\pi}}{W^{\mu}_T}]$ in three financial market models respectively, by using simulations of those market models. This corresponds to the optimization task described in Theorem~\ref{th:log_opt_task}. To apply this in a numerical setting, we use the following Monte-Carlo optimization to approximate the expected values:
	
	\begin{cor}
            Let $\tilde{\Omega}$ be a probability-one set on which $\mathbf{Q}(t)$ is positive-semidefinite. Consider a Monte-Carlo type optimization of the expected log-relative wealth. Recall that $V_t^{\pi}= \frac{W_t^{\pi}}{W^{\mu}_t}$, hence we optimize for $$\max_{\pi} \hspace{0.2cm} \frac{1}{M}\sum_{m=1}^M \log\left(V_t^{\pi}\right)( \omega_m)$$
			for $\omega_1,\ldots, \omega_m \in \tilde{\Omega}$. 
			Again, this can be formulated as a convex quadratic optimization problem, namely 
			\begin{align} 
				\min_{\mathbf{l}} \hspace{0.2cm}  \frac{1}{2} \mathbf{l}^T \hat{\mathbf{Q}}(t)\mathbf{l}-\mathbf{l}^T\hat{\mathbf{c}}(t) ,
			\end{align}
			where $\hat{\mathbf{Q}}(t)= \frac{1}{M}\sum_{m=1}^M \mathbf{Q}(t)(\omega_m)$ and $\hat{\mathbf{c}}(t)= \frac{1}{M}\sum_{m=1}^M \mathbf{c}(t)(\omega_m)$.
	\end{cor}
	
	\begin{proof}
	    We have shown in the proof of Theorem~\ref{th:log_opt_task} that $\mathbf{Q}(t)$ is almost surely positive semi-definite for all $t\in[t_0,T]$, hence such a set $\tilde{\Omega}$ exists. Since the sum of positive semi-definite matrices is positive semi-definite, $\hat{\mathbf{Q}}(t)$ is positive semi-definite and the statement follows. 
	\end{proof}
	
	Moreover, we will compare the performance of the learned portfolio with the growth-optimal portfolio of the respective markets. This is justified by Lemma~\ref
 {lem:numeraire_go}, which states that the growth-optimal portfolio (which is also the numeraire portfolio) solves the relative log-utility optimization problem. 
	
	For each market, we simulate $M_{train}$ trajectories over the time-horizon $[0,1]$ of a given market model and consider investments over the whole time-horizon, i.e. $t_0=0$, $t=1$. We train a signature portfolio using the Monte-Carlo type optimization and the formulas for $\mathbf{Q}(t)(\omega)$ and the vector $\mathbf{c}(t)(\omega)$ provided in Theorem~\ref{th:log_opt_task}, where we use signature portfolios $\pi(\mu, \hat{\mu})$, i.e. choose $X=\mu$ to construct the portfolios.  Note, that calculating the matrix $\mathbf{Q}(t)(\omega)$ and vector $c(t)(\omega)$ for each training-sample can be  paralellized and computed offline.
 Having obtained the respective $\hat{\mathbf{Q}}(t)$ and $\hat{c}(t)$, we use the \texttt{gurobipy.model.optimize()}\footnote{for more informartion, see \url{https://www.gurobi.com/documentation/9.1/refman/py\_python\_api\_overview.html\#sec:Python}.} method to solve the convex quadratic optimization problem. We want to emphasize again, that we only need to solve the convex quadratic optimization task once and we never have to update $\hat{\mathbf{Q}}(t)$ and $\hat{c}(t)$, which is very beneficial for the computational tractability of the optimization task. 
	
	After the optimization, we compare the weights and performance of the trained signature portfolio and the theoretical growth-optimal portfolio on $M_{test}$ out-of-sample trajectories.
	
	More concretely, we consider the following three market models and signature portfolios: 
	\begin{enumerate}
		\item \textbf{Correlated Black-Scholes Model}:  $$\dd S_t = \operatorname{diag}(S_t)(a dt + \Sigma dB_t),$$
		where $a \in\mathbb{R}^d$ and $\Sigma \in \mathbb{R}^{d \times d}$, $B_t$ is a $d$-dimensional Brownian motion. 
   We here take $d=3$ and  train a signature portfolio $\pi(\mu, \hat{\mu})$ of type $I$ of degree three. We simulate  1'000 time-steps. 
		\item \textbf{Volatility Stabilized Market}:$$\frac{\dd S^i_t}{S^i_t} = \frac{1+\alpha}{2} \frac{1}{\mu_t^i} dt + \sqrt{ \frac{1}{\mu_t^i}} dB_t^i, \quad  1\leq i \leq d.$$
		We choose again $d=3$, so that $B$ is a three-dimensional Brownian motion. We take $\alpha=10$ and train a signature portfolio $\pi(\mu, \hat{\mu})$ of type $I$ of degree three. We simulate 10'000 time-steps.
		\item \textbf{Signature Market}: 
		\begin{equation*}
			dS_t= \operatorname{diag}(S_t)(\mathbf{a}_t dt + \Sigma \, dB_t).
		\end{equation*}
We choose again $d=3$, so that $B$ is a three-dimensional Brownian motion.
  Moreover, $(\mathbf{a}_t)_i = \sum_{0\leq |I| \leq 2} \alpha_I^{(i)} \langle e_I, \hat{\bbmu} \rangle_t$, $\alpha_I^{(i)} \in \mathbb{R}$ with the restriction that $\alpha_I^{(i)}=0$ if $I=(i_1, i_2)$ and $i_2\neq 1$ and $\Sigma$ is a constant $3\times3$ matrix. Here we train a signature portfolio $\pi(\mu, \hat{\mu})$ of type $II$ of degree two and simulate 1'000 time-steps. 
	\end{enumerate}
	
	We use $M_{train}=100'000$ in-sample trajectories for training and $M_{test}=100'000$ out-of-sample trajectories for evaluation, i.e. computing the average logarithmic relative wealth of the respective portfolios on the test-samples. Furthermore, we used the bound constraints $|l_I^{i}| \leq 10'000$ for all $I$ and for all $1 \leq i \leq d$ for each signature portfolio.

    \begin{remark}\label{rem:num_optim_lin_func_sig}
      In a volatility stabilized market, signature portfolios $\pi(\mu, \hat{\mu})$ of type $I$, with $\hat{\mu}=(t, \mu)$ fall into the first case of Remark~\ref{rem:simplification}, where 
      \begin{equation*}
            d\left[ \mu^i, \mu^j\right]_t = \begin{cases}
            \mu_t^i(1-\mu_t^i) dt \textrm{ for } i=j \\
            - \mu_t^i\mu_t^j dt \textrm{ for } i\neq j. \end{cases}
        \end{equation*}
        
    \end{remark}

	\begin{remark}
       In Lemma \ref{lem:numeraire_go} we have already established that the growth-optimal portfolio solves the relative log-utility optimization problem. We now want to outline that the growth-optimal portfolio exists \emph{and} that the relative log-utility optimization
       for signature portfolios is well-posed, i.e. that the expectation of this optimization task is finite for all signature portfolios. Let us denote the growth-optimal portfolio by $\pi^{(g)}$, recall its form given by Lemma~\ref{lem:growth-opt} and note that a sufficient condition for both of the above to be satisfied is if $\rvert \mathbb{E}[\log W_T^{\pi^{(g)}} ] \rvert < \infty$. We will now present the arguments for the existence of $\pi^{(g)}$ and the well-posedness of the optimization problem for each of the markets. 
	    \begin{enumerate}
	        \item Correlated Black-Scholes Model: It can easily be checked that for the growth-optimal portfolio $\pi^{(g)}$ it holds that $\rvert \mathbb{E}[\log W_T^{\pi^{(g)}} ] \rvert < \infty$.
	        
	        \item Volatility Stabilized Market:
            Our optimization problem is well-posed for signature portfolios, because volatility stabilized market models are  polynomial processes, see~\cite{Cuchiero_Poly_SPT} and  all coefficients of $\mathbf{c}$ and $\mathbf{Q}$ are linear functions on the signature of $\hat{\mu}$, as established in Remarks~\ref{rem:simplification} and~\ref{rem:num_optim_lin_func_sig}. Moreover, we know that the truncated signature of a polynomial process is again a polynomial process (see \cite{cuchiero2023joint}) and hence their expectation is finite. Therefore, the expectation of any coefficient of $\mathbf{c}$ and $\mathbf{Q}$ is finite and $|\mathbb{E}\left[ \frac{W^{\pi}}{W^{\mu}}\right]| <  \infty$ holds for any signature portfolio $\pi$ of type $I$. 
	        
	        \item Signature Market: Assuming that a solution exists, the coefficients $\mathbf{a}_i$ are all uniformly bounded on $[0,1]$. This follows form the fact that $\mu^i$ are positive and uniformly bounded for all $i=1,\dots, d$, hence so is $\int_0^T \mu^i_t dt$. It thereby follows $\rvert \mathbb{E}[\log W_T^{\pi^{(g)}} ] \rvert < \infty$ for the growth-optimal portfolio $\rho$.
	    \end{enumerate}
	    
	\end{remark}

	\begin{remark}
	    Note that the growth-optimal portfolios of the Black-Scholes and volatility stabilized markets are also functionally generated portfolios in the classical sense. Indeed, in the Black-Scholes case the portfolio generating function is 
	    \begin{equation*}
	        G^{(BS)}(\mu_t) = \prod_{i=1}^d (\mu_t^i)^{c_i},
          \end{equation*} with portfolio weights $\pi^{(BS), i}_t= c_i$, where
	    \begin{equation*}
	        c_i= ( (\Sigma^{\mathsf{T}}\Sigma)^{-1} a)_i - \sum_{j=1}^d ((\Sigma^{\mathsf{T}}\Sigma)^{-1})_{ij}\left( \frac{\sum_{k=1}^d (\Sigma^{\mathsf{T}}\Sigma)^{-1} a)_k -1}{\sum_{j,k=1}^d ((\Sigma^{\mathsf{T}}\Sigma)^{-1})_{jk}} \right).
	    \end{equation*}
	    In the volatility stabilized market case  the generating function of the growth optimal portfolio is given by
	    \begin{equation*}
	        G^{(Vol)}(\mu_t) = \prod_{i=1}^d (\mu_t^i)^{\frac{1+\alpha}{2}} \exp{\left(\mu^i_t \frac{d}{2}(\alpha-1)\right)},
         \end{equation*}with portfolio weights $$\pi^{(Vol), i}_t= \frac{1+\alpha}{2}+ \mu^i_t \frac{d}{2}(\alpha-1).
	 $$
	    Therefore, these two examples are not only examples of signature portfolios approximating the respective growth-optimal portfolios but also examples for signature portfolios approximating functionally generated portfolios. 
	\end{remark}
	
	Figure~\ref{fig:simulated_markets} shows the trained signature portfolios (right) and the growth-optimal portfolios (left) evaluated at one out-of-sample trajectory in each market respectively (rows). As  expected from the theoretical results, the signature weights are very similar to the theoretical growth-optimal weights in all markets. We want to emphasize, that the growth-optimal weights were never shown to the signature portfolio during training, but  the signature portfolio was trained to maximize the expected logarithmic relative wealth.
	
	\begin{figure}[ht] 
		\centering
		
		\begin{subfigure}{0.49\columnwidth}
			\includegraphics[width=\linewidth]{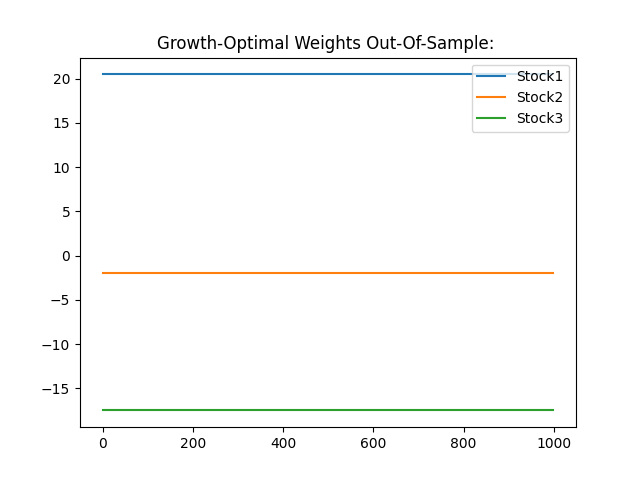}
			\caption{Black-Scholes market: growth-optimal weights}
		\end{subfigure}
		\begin{subfigure}{0.49\columnwidth}
			
			\includegraphics[width=\linewidth]{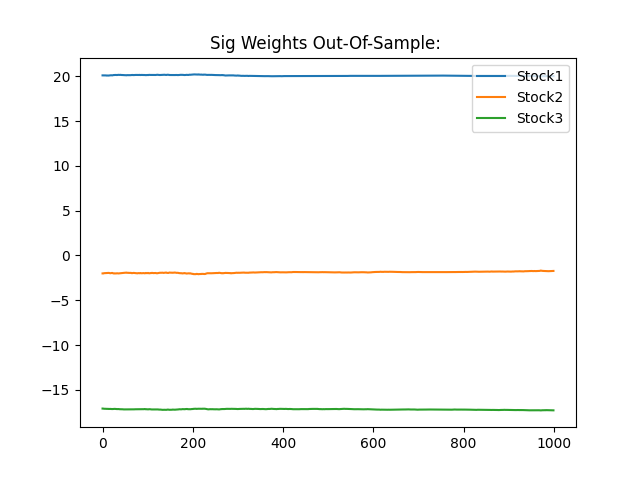}
			\caption{Black-Scholes market: signature weights}
		\end{subfigure}
		\vspace*{0.7cm}\\
		\begin{subfigure}{0.49\columnwidth}
			\includegraphics[width=\linewidth]{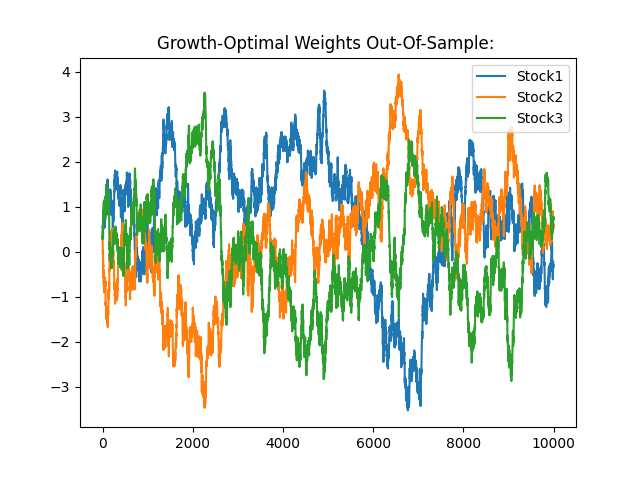}
			\caption{Vol.~stabilized market: growth-optimal weights}
		\end{subfigure}
		\begin{subfigure}{0.49\columnwidth}
			
			\includegraphics[width=\linewidth]{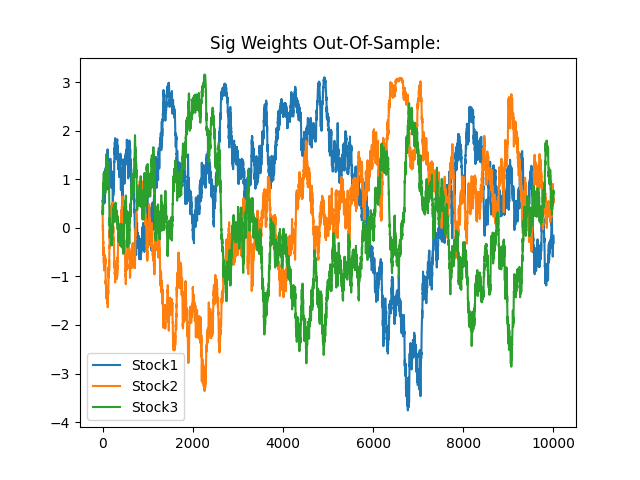}
			\caption{Vol.~stabilized Market: signature weights}
		\end{subfigure}
		\vspace*{0.7cm}\\
		\begin{subfigure}{0.49\columnwidth}
			\includegraphics[width=\linewidth]{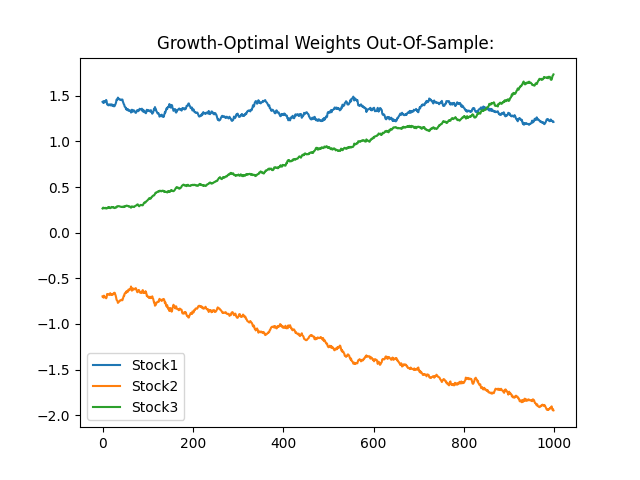}
			\caption{Signature market: growth-optimal weights}
		\end{subfigure}
		\begin{subfigure}{0.49\columnwidth}
			
			\includegraphics[width=\linewidth]{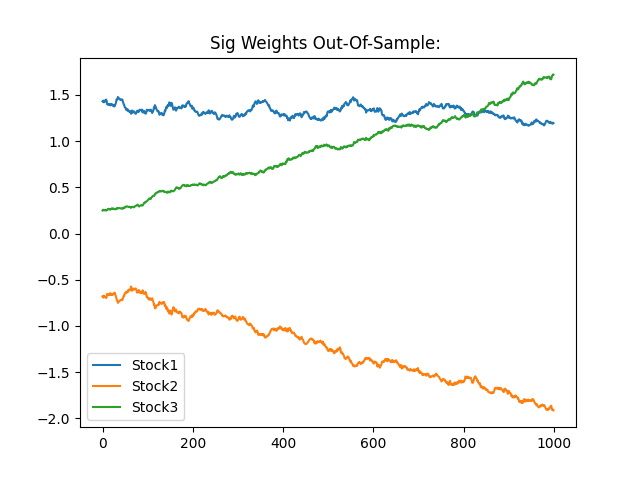}
			\caption{Signature market: signature weights}
		\end{subfigure}
		\caption{ \label{fig:simulated_markets} The theoretical growth-optimal weights (left) and the signature portfolio's weights (right) for the Black-Scholes, volatility stabilized and signature market respectively, evaluated at one test sample.}
	\end{figure}
	
	Furthermore, we want to highlight, that although the growth-optimal portfolio in the Black-Scholes market has constant weights, this approximation task is far from trivial for a signature portfolio of type $I$. Because signature portfolios of type $I$ approximate the \controllingfunctions of the growth-optimal portfolio, the approximation task was actually $$ f^{(BS), i}(t,\mu_{[0,t]}) \approx \frac{c_i}{\mu_t^i}. $$
	And likewise, in the volatility stabilized market, the approximation task was 
	$$ f^{(Vol), i}(t,\mu_{[0,t]}) \approx \frac{\alpha+1}{2\mu_t^i} + \frac{d}{2}(\alpha-1). $$
	We quantify the performance by comparing the average logarithmic relative wealth over the $M_{test}$ test samples of the growth-optimal and signature portfolios, where in both cases we compute the average logarithmic relative wealth numerically on the test samples. This explains why the signature portfolio sometimes even leads to higher values. We present these results in Table~\ref{tb:sim_mkt_results}. 
    
    \begin{table}
	\centering
	
	\begin{tabular}{ c | c | c } 
		Mean log-relative wealth of & growth-optimal portfolio & signature portfolio\\ \hline
		Black-Scholes market: & 9.0115  & 9.0122  \\ 
		Volatility stabilized market: & 8.7619  & 8.7417 \\ 
		Signature market: & 0.4399  & 0.4398   \\ 
	\end{tabular}
	
	\caption{ \label{tb:sim_mkt_results} Mean logarithmic relative wealth of the theoretical growth-optimal and signature portfolios, evaluated on 100'000 test-samples, in each market respectively.}
\end{table}

\subsection{Real Market Data}\label{sec:realdata}
    \subsubsection{Details on the Optimization and Investment Procedure}
    
        Before we present the performance of signature portfolios, as well as, JL- and randomized-signature portfolios in real markets, we want to give some details on the optimization and investment procedure.

        \paragraph{Approximating Expectations by Time-Averages}
        When working with real market data, we only have one realization available. Therefore, as already discussed in Remark~\ref{rem:ergodic} and at the beginning of Section \ref{sec:logutility} we replace the expectation by a time-average. More precisely, consider $[0, \Delta, 2\Delta, \dots, N \Delta=t]$ to be a equally spaced grid of $[0, t]$ with distance $\Delta$. Then,  for $s\in [0,t-\Delta]$ and analogously to Remark \ref{rem:ergodic} we consider the following approximations:
$$
\mathbb{E}\left[\log \left(\frac{V^{\pi}_{s+\Delta}}{V^{\pi}_s}\right)\right] \approx \frac{1}{N} \sum_{i=1}^N \log \left(\frac{V^{\pi}_{i\Delta}}{V^{\pi}_{(i-1)\Delta}}\right)=\frac{1}{N} \log(V_t^{\pi})
$$

        $$\mathbb{E}[\mathcal{R}^{V,\pi}_{s,s+\Delta}]\approx \frac{1}{N}\sum_{i=1}^{N} \mathcal{R}^{V,\pi}_{(i-1)\Delta,i\Delta}=:\overline{\mathcal{R}}^{V,\pi},$$

        $$\textrm{Var}[\mathcal{R}^{V,\pi}_{s,s+\Delta}]\approx \frac{1}{N-1}\sum_{i=1}^{N} (\mathcal{R}^{V,\pi}_{(i-1)\Delta,i\Delta}- \overline{\mathcal{R}}^{V,\pi})^2.$$
Recall that sufficient ergodicity/stationarity conditions guaranteeing that these approximations are justified have been discussed in Remark \ref{rem:ergodic}.

Note however that we cannot expect that  stationarity of the relative returns and log-relative returns is satisfied for   \emph{every} portfolio in our optimization class.     
       Consider for example signature portfolios of type $I$ with $\tau= (\frac{1}{d}, \dots \frac{1}{d})^{\mathsf{T}}$, $\hat{X}_t= (t, X_t)$ and $\mathcal{U}=\{1,\dots, d\}$. It is easy to see that the stationarity assumption can not hold for all portfolios in this class, just by considering the following portfolio $\pi^1$ where the only parameters that are non-zero are chosen to be $l^d_{(1)}$. Then 
        \begin{align*}
            \mathbb{E}\left[\mathcal{R}^{V, \pi^1}_{t, t+\Delta}\right]= \frac{1}{d} \mathbb{E}\left[t\cdot \frac{\mu^d_{t+\Delta}}{\mu^d_t} \right] + (\frac{1}{d}- \frac{t}{d^2})\mathbb{E}\left[\sum_{j=1}^d \frac{\mu^j_{t+\Delta}}{\mu^j_t}\right]- \frac{t+1}{d}.
        \end{align*}    
        
         However, there are of course classes of linear path-functional portfolios where the stationarity assumptions holds for all portfolios in the class. Indeed, the simplest example is  the class of Markowitz-type portfolios, i.e.~with $\tau$ being constant and constant portfolio maps. A far-reaching generalization thereof are
         signature portfolios with  \emph{with rolling windows}, as defined in Remark~\ref{rem:examples_linear_port} and for which stationarity holds under the conditions of Remark \ref{rem:ergodic}. Computing these portfolios is slightly more expensive, since one needs to compute the increments of the signature. Therefore we did not consider them in our implementations.

       Finally, let us emphasize that the optimization using time-averages does make sense in practice, even if the stationarity of returns may not hold for every portfolio. For example  consider $\Delta=1 \textrm{ day}$, then the mean-variance optimization using time-averages can be seen as looking for a portfolio with high average daily returns (in time) but without the daily returns varying too much over time and likewise for the log-relative wealth optimization. Moreover, we argue that the relative returns and log-relative returns of our optimized portfolios do exhibit some "stability" in time, as almost all the learnt portfolios perform very well in the out-of-sample period, as we will present in the following.
        Note, that all of the above optimization problems remain convex quadratic optimization problems if we replace expectations by time-averages.

        \paragraph{Optimization under Transaction Costs}
        In practice, we re-balance our portfolio-weights once a day since our market data is also of daily frequency. In this setting, we would like to incorperate transaction costs during optimization. \\
        
        As explained in detail in Section~\ref{sec:transcost}, including transaction costs is not trivial in our setting. In particular, since we have shown that \eqref{eq:alpha} does not necessarily have a (unique) solution, we cannot aim for any closed-form. Therefore, incorporating the exact amount of transaction-costs to be paid during the optimization would render the optimization task infeasible. However, we propose the following penalization which preserves the form of a convex quadratic optimization problem. Moreover, our empirical results in Subsection~\ref{subsec:SMI_SPX} verify that this penalization is effective in all cases and indeed useful for learning portfolios which perform well even under 5\% of transaction costs. 
        
        \begin{cor}
            Given a convex quadratic optimization problem of the form 
            $$ \min_{\mathbf{l}} \, \frac{1}{2}\mathbf{l}^{\mathsf{T}}\mathbf{Q}(T) \mathbf{l} + \mathbf{c}(T)^{\mathsf{T}} \mathbf{l}$$
            adding the penalization for linear path-functional portfolios $\pi$
            \begin{equation*}
                +\frac{\beta}{T} \sum_{t=0}^{T-1} \sum_{i \in \mathcal{U}} \left(\frac{\pi^i_{t+1}}{\mu^{\mathcal{U}, i}_{t+1}} -\frac{\pi^i_{t}}{\mu^{\mathcal{U}, i}_{t}} \right)^2 
            \end{equation*}
            preserves the form of a convex quadratic optimization problem. Here, $\beta$ is a hyperparameter that has to be chosen appropriately.
        \end{cor}
        
        \begin{proof}
            Clearly the penalization is quadratic in $\pi$ and hence quadratic in the parameters $\{l_{\nu}^i\}_{i\in \mathcal{U}, \nu \in \mathcal{V}}$. Moreover, the penalization is positive for all $\{l_{\nu}^i\}_{i\in \mathcal{U}, \nu \in \mathcal{V}}$ and hence the convexity is preserved as well. 
        \end{proof}
        
        The motivation for this penalization is, first of all, that the universe portfolio $\mu^{\mathcal{U}}$ is not punished, which should be the case because this portfolio has no transaction costs at all. Moreover, the penalization punishes changes in the weights which exceed changes in the market weights. Note that those are exactly the changes that lead to transaction costs.\\
        
        \FloatBarrier
        \begin{algorithm}
            {\footnotesize
            \caption{Optimization Under Transaction Costs}\label{algo:find_beta}
            \begin{algorithmic}[1]
                
                \Function{Optimize\_Under\_TC}{beta}:
                \State $l \gets$  \textsc{Solve\_Convex\_Quadratic}(beta)
                \State weights $\gets$  \textsc{Get\_Insample\_Weights}($l$)
                \State wealth $\gets$  \textsc{Wealth\_Under\_TC}(weights)
                \State \Return $-1\cdot$  wealth
                \EndFunction
                \While{($-1 \cdot$\textsc{Optimize\_Under\_TC}(beta\_0)) $<10^{-4}$} 
                \State beta\_0+=0.5 \EndWhile
                \State optimal\_beta $\gets$\textsc{Minimize}(\textsc{Optimize\_Under\_TC}, beta\_0)
                \State optimal\_$l$ $\gets$ \textsc{Solve\_Convex\_Quadratic}(optimal\_beta)
            \end{algorithmic}
            }
        \end{algorithm}
        
        The downside of this penalization is, that for a given level of transaction costs, we do not know how to choose $\beta$. To find an appropriate $\beta$, we propose the procedure described in Algorithm~\ref{algo:find_beta}. That is, we solve the quadratic optimization problem for a given $\beta$ and calculate the in-sample performance \emph{under the true transaction costs} for the found optimal portfolio. We then minimize over $\beta$ in order to find the one which best penalizes for transaction costs over the in-sample period. Note that choosing the initial value $\beta_0$ can be delicate. Namely, if in a neighbourhood around $\beta_0$ all optimal strategies lead to ruin under transaction costs, one may not move away from the initial value. Therefore, we test if $\beta_0$ leads to ruin under transaction cost and, if so, increase it, in Algorithm~\ref{algo:find_beta}. Note that the condition in Line 8 must be false for some $\beta_0$ because the universe portfolio itself is included in the set of portfolios we optimize over.

        \FloatBarrier
        
        \paragraph{The Optimization Problems}
        In the following, we consider optimization problems of various signature-type linear path-functional portfolios with $L^2$-regularization, with and without a regularization for transaction costs. The two types of optimization problems are maximizing the expected log-relative wealth
        \begin{equation} \max_{\{l_{\nu}^{(i)}\}_{{i \in \mathcal{U}}, {\nu} \in \mathcal{V}}} \mathbb{E} \left[ \log \left( \frac{W^{\pi}_{t}}{W^{{\mathcal{U}}}_{t}}\right) \right] - \gamma_{\mathrm{LO}} \sum_{\substack{{i \in \mathcal{U}}\\ {\nu} \in \mathcal{V}}} (l_{\nu}^{(i)})^2 - \frac{\beta}{T} \sum_{t=0}^{T-1} \sum_{i \in \mathcal{U}} \left(\frac{\pi^i_{t+1}}{\mu^{\mathcal{U}, i}_{t+1}} -\frac{\pi^i_{t}}{\mu^{\mathcal{U}, i}_{t}} \right)^2 \tag{Log-Opt.}
        \end{equation}
        and maximizing mean-variance
        \begin{equation} \max_{\{l_{\nu}^{(i)}\}_{{{i \in \mathcal{U}} {\nu} \in \mathcal{V}}}}  \hspace{0.1cm} \lambda \mathbb{E}\left(\mathcal{R}^{V, \pi}_{t, t+\Delta} \right) - \mathrm{Var}\left(\mathcal{R}^{V, \pi}_{t, t+\Delta}\right) - \gamma_{\mathrm{MV}} \sum_{\substack{{i \in \mathcal{U}}\\ {\nu} \in \mathcal{V}}} (l_{\nu}^{(i)})^2 - \frac{\beta}{T} \sum_{t=0}^{T-1} \sum_{i \in \mathcal{U}} \left(\frac{\pi^i_{t+1}}{\mu^{\mathcal{U}, i}_{t+1}} -\frac{\pi^i_{t}}{\mu^{\mathcal{U}, i}_{t}} \right)^2 \tag{MV-$\lambda$}
        \end{equation}
        where in both cases we added bound constraints $\lvert l^{(i)}_{\nu}\rvert \leq 10'000$ for all ${i \in \mathcal{U}}, {\nu} \in \mathcal{V}$.

                \paragraph{Hyperparameters and Cross-Validation}
        The main hyperparameters of our optimization procedure are:
        \begin{itemize}
            \item $t_0$: this is the time at which we start to invest, with respect to the time when we start calculating the signature. The interpretation of this is that we can observe the market for some time before we start to invest, which is particularly relevant in the non-Markovian setting. In the following we choose $t_0=100$, hence, we always start calculating the signature 100 days before the respective investment period starts.

            \item $\gamma$: this is the parameter of the $L2$-regularization described in Proposition \ref{prop:mod_optim}. For the case of maximizing the (expected) log-relative wealth, we choose our $\gamma$ during cross-validation. We do not train this regularization parameter in the mean-variance optimization but fix it apriori to a small value, as the minimization of the variance itself can already be regarded as regularizing. 
        \end{itemize}
        
        For performing the cross-validation, we first split our data into a in-sample-period of $T_{ins}$ days, a cross-validation of $T_{cv}$ (consecutive to the in-sample period) and a testing period of $T_{test}$ days, starting at the end of the cross-validation period. Once we found the optimal hyper-parameters during cross-validation, we learn the parameters of our portfolios again on the $T_{ins}$ days prior to the testing period (using the obtained hyper-parameters). We then evaluate the performance of the learned portfolios during the testing-period.
        
        \paragraph{Calculating the JL-Signature}
        In order to calculate the JL-signature, the true truncated signature needs to be calculated before applying the JL-projection. However, calculating the true signature can be computationally too heavy in large markets. We therefore propose the following memory-efficient procedure (Algorithm~\ref{alg:randomsig}), which is based on the observation that for each component of the signature at level $l$, at most $l$ components of the underlying path appear in the integrals. Hence, one can instead compute the signature of combinations of $l$ components of the path and apply the projection "batch-wise". It is important to note that by doing so, some words are computed multiple times, for example $$\int_0^t \int_0^s \circ dX_u^1 \circ dX_s^1 $$
        can arise from the combination $(X^1, X^2 ,X^3)$ and from $(X^1, X^4, X^5)$ and from many more. Therefore, one needs to be careful which words to keep and our proposed algorithm takes care of that. Moreover, it is important to store the random projection matrix, once a realization is computed. 
        
        \begin{algorithm}
            {\footnotesize
            \caption{Calculate JL-Signature}\label{alg:randomsig}
            \begin{algorithmic}[1]
                
                \Function{JL\_Signature}{path, level}:
                    \State Initialize A\_list=\texttt{list}()
                     \State Initialize JL\_sig=0
                    \For{$l \in \{1,\dots, \textrm{level}\}$} 
                        \State combos $\gets$ \textsc{Combnations}(\texttt{dim}(path), $l$)
                        \For{$c \in  \textrm{combos}$} 
                            \State A\_slice $\sim$  N(0, 1/k)
                            \State sig, words $\gets$ \textsc{Signature}(path[c], level)
                            \State words\_to\_keep= \texttt{list}()
                            \For{word $\in$ words}
                                \If{\texttt{length}(\texttt{set}(word))== $l$} 
                                    \State words\_to\_keep.\texttt{append}(word)
                                \EndIf
                            \EndFor
                            \State sig= sig[words\_to\_keep]
                            \State A\_slice= A\_slice[\, :\, , words\_to\_keep]
                            \State JL\_sig+=\texttt{matrix\_product}(A\_slice, sig)
                        \EndFor
                    \EndFor
                    \State \Return JL\_sig  
                \EndFunction
            \end{algorithmic}}
        \end{algorithm}

        \paragraph{Parametrization of Time-Components}  We parameterize the time-component of the time augmentation in the following way: Let $\hat{X}_t = (\varphi_T(t), X_t)$. For a given trading horizon $T_{hor}$, we consider the parametrization $\varphi_{T_{hor}}(t)= \frac{t}{T_{hor}}$. Concretely, for the in-sample period of 2000 days, we set $T_{hor}=2000$ and for the out-of-sample period we set $T_{hor}=750$. This time-augmentation therefore contains information about the amount of time that is left (or has passed) in the current trading period. This is to compensate for the different training and testing periods. Another way to deal with this would be to use signature portfolios with rolling windows as defined in Remark~\ref{rem:examples_linear_port}.

        \paragraph{Performance Metrics}
        Recall that we re-balance our portfolios once a day and follow a buy-and-hold strategy in between. We evaluate the out-of-sample performance of portfolios by their log-relative wealth
        $$\log \left(V^{b\&h}_T\right) = \log \left( \prod_{t=t_0}^{T-1} \sum_{i \in \mathcal{U}} \pi^i_t \frac{\mu^{\mathcal{U}, i}_{t+1}}{\mu^{\mathcal{U}, i}_{t}}\right)$$ and their log-wealth
        $$\log \left(W^{b\&h}_T\right) = \log \left(\prod_{t=t_0}^{T-1} \sum_{i \in \mathcal{U}} \pi^i_t \frac{S^{\mathcal{U}, i}_{t+1}}{S^{\mathcal{U}, i}_{t}}\right).$$

    \paragraph{Survivorship Bias}
    In our name-based approach we only consider a universe of stocks to invest in which were present in the market for the entire in- and out-of-sample period. This opens up the issue of survivorship bias. Note that the same holds true for the ranked NASDAQ market, since there was always a 100-largest stock present, even though the company holding that rank may have changed over time. As a first measure to reduce this bias, we compare the performance of our trained portfolios only to the universe portfolio of the surviving stocks and not to the entire market portfolio. Nevertheless, our trained portfolios may benefit from the fact that stocks do not go bankrupt in other ways which the universe portfolio can not exploit. Let us address them in more detail:
    \begin{itemize}
    \item {\bf Leverage:} If bankruptcy of stocks is not possible, leverage (i.e. short-selling of some stocks in order to invest more heavily in others) is less risky. We do not impose any leverage constraints in our optimization. Nevertheless we do not observe extreme short-selling during the out-of-sample period. We report the minimum observed out-of-sample weights observed in Tables~\ref{tb:minpi_SMI} and~\ref{tb:minpi_SPX}, were we highlight those in bold which are long-only. At least for those which do not short-sell we can exclude leveraging as an unfair advantage. In the cases where only small negative weights were observed, we dare to say that an unfair advantage through leveraging is unlikely to have been the driving factor for potential out-performance of the universe portfolio. Aside from leveraging, we would like to point out that regularization for transaction costs reduces short-selling.

    \item {\bf Over-weighting small stocks:} If  stocks' capitalization cannot go to zero, a potential strategy may be to put a lot of weight on stocks with capitalizations close to zero, because the upside is much higher than the downside. To investigate this, we point to Figures~\ref{fig:JL_semilog}~and~\ref{fig:RSIG_semilog}, which show the average out-of-sample weights of the ranked NASDAQ universe. We observe that the trained portfolios \emph{underweight} small stocks on average compared to the universe portfolio.  
    \end{itemize}

    \begin{table}
    
	   \centering
	        \begin{subtable}[h!]{\textwidth}
                \centering
            {\footnotesize
            \begin{tabular}{ c | c | c | c | c | c| c | c| c | c} 
	    	    Portfolio & \multicolumn{3}{|c|}{Signature Portfolio }& \multicolumn{3}{|c|}{JL-Signature Portfolio }&\multicolumn{3}{|c}{Rand.-Signature Portfolio}\\ 
	    	    Reg. TC & 0\% & 1\% & 5\% & 0\% & 1\% & 5\%& 0\% & 1\% & 5\% \\\hline
		         Log-Opt. & -0.5 & -0.45   & -0.13   & -0.32  & -0.32 &-0.12  & -0.01 & -0.01 &  {$\mathbf{3\cdot 10^{-3}}$}\\ 
		        MV $\lambda=1$ & -3.31 & -1.45  & -0.16 & -3.70 & -1.67 & -0.18 & -3.37& -2.26 &-0.20\\
		        MV $\lambda=0.75$ & -2.43 & -1.72 & -0.16 & -2.72 &-2.10  & -0.18 & -2.48 & -2.22 & -0.21 \\
		        MV $\lambda=0.5$ & -1.56 & -1.54 & -0.17 &-1.75& -1.75 & -0.19 & -1.59  &-1.54 & -0.22\\
		        MV $\lambda=0.05$ &-0.05& -0.05 & -0.04 & -0.05& -0.05& -0.04 &-0.04 & -0.04& -0.03\\
    	    \end{tabular}}
            \subcaption{\label{tb:minpi_SMI} Minimum portfolio weight observed during the out-of-sample period for the SMI universe for the respective portfolios with and without regularization for transaction costs. }
            \end{subtable}
	        
	        \vspace{0.36cm}
	        
	        \begin{subtable}[h!]{\textwidth}
                \centering
            {\footnotesize
            \begin{tabular}{ c | c | c | c | c | c| c } 
	    	    Portfolio & \multicolumn{3}{|c|}{JL-Signature Portfolio }&\multicolumn{3}{|c}{Rand.-Signature Portfolio}\\ 
	    	    Reg. TC & 0\% & 1\% & 5\%& 0\% & 1\% & 5\% \\\hline
		         Log-Opt. & -1.67 &  -0.66   & -0.11  & -0.80  & -0.63 &-0.15\\ 
		        MV $\lambda=1$ & -0.14 & -0.14  & -0.11 & -0.04 & -0.04 & -0.04\\
		        MV $\lambda=0.75$ & -0.10 & -0.10  & -0.10 & -0.03 &-0.03  & -0.03 \\
		        MV $\lambda=0.5$ & -0.06& -0.06 &-0.06 & -0.01 & -0.01 & -0.01 \\
		        MV $\lambda=0.05$ &{$\mathbf{8\cdot 10^{-5}}$}& \textbf{$\mathbf{8\cdot 10^{-5}}$}  & \textbf{$\mathbf{8\cdot 10^{-5}}$} & \textbf{$\mathbf{8\cdot 10^{-5}}$}& \textbf{$\mathbf{8\cdot 10^{-5}}$} & \textbf{$\mathbf{8\cdot 10^{-5}}$}\\
    	    \end{tabular}}
            \subcaption{\label{tb:minpi_SPX} Minimum portfolio weight observed during the out-of-sample period for the S\&P500 universe for the respective portfolios with and without regularization for transaction costs. }
            \end{subtable}
            
            \caption{The tables show the minimum portfolio weights observed during the out-of-sample period of the respective universe. The ones which are positive (i.e. no short-selling) are highlighted in bold.}
        \end{table}

    \paragraph{Choice of a Benchmark}
    In the following numerical experiments we will use the universe weights as a auxiliary portfolio $\tau$. Recall from Remark~\ref{rem:auxiliary} that the auxiliary portfolio $\tau$ is the natural benchmark for linear path-functional portfolios because it can always be attained by them. As an additional benchmark we included the equally-weighted portfolio. We want to emphasise that of course other auxiliary portfolios and hence benchmark portfolios could be used.

    \subsubsection{Rank-based approach: NASDAQ}\label{subsec:NASDAQ}
         In this part of our empirical analysis we study a rank-based approach to portfolio optimization. We consider the \emph{ranked} NASDAQ market and choose as a universe the stocks with ranks $1, \dots, 100$. We obtained the data from the CRSP database.\footnote{The raw/processed data required to reproduce our findings cannot be shared at this time due to legal reasons.} We train JL-signature portfolios of dimension $(50,3)$ of type $I$ and randomized signature portfolios of dimension 50 of type $I$ investing in the universe of the 100 largest stocks and  use as input the universe weights i.e. $\hat{X}= \hat{\boldsymbol{\mu}}^{\mathcal{U}}$ and the universe portfolio as the auxiliary portfolio (i.e. $\tau= \boldsymbol{\mu}^{\mathcal{U}}$). 
        The portfolio is daily re-balanced. We train such portfolios in two ways, once in the log-wealth optimization and once in the mean-variance optimization, where in both optimization problems, we measure the performance of the portfolios by the log-relative wealth and the log-wealth achieved in the out-of-sample period. 
        
        We choose $t_0=100$ and take as an in-sample period $2000$ trading days and as an out-of-sample period the following $750$ trading days. For the mean-variance optimization task we set the $L^2$-regularization parameter $\gamma_{MV}= 10^{-6}$ and for the log-wealth optimization task we found $\gamma_{LO}= 4.849\cdot 10^{-3}$ for JL-signature portfolios and $\gamma_{LO}= 1\cdot 10^{-2}$ for randomize-signature portfolios during the cross-validation. The grid-search for $\gamma_{LO}$ was performed over 100 equally-distant points in $[10^{-6}, 10^{-2}]$. 
        
        We report the out-of-sample performance of the trained portfolios in Table~\ref{tb:Results_NASDAQ_100} in terms of log-(relative)-wealth with respect to the universe portfolio. All but one of the trained portfolios outperform the universe portfolio. 
        
        \begin{table}
	        \centering
	        {\small
	        \begin{tabular}{ c | c | c | c | c } 
	        Optimization & \multicolumn{2}{|c|}{JL-Signature }& \multicolumn{2}{|c}{Randomized Signature}\\
	    	    Task & $\log V_T^{b\&h}$ & $\log W_T^{b\&h}$ &$\log V_T^{b\&h}$ & $\log W_T^{b\&h}$\\ \hline
		        Log-Opt. & 0.1045 & 0.6622& 0.0002 & 0.5579 \\ 
		        MV $\lambda=4$ &0.3395 & 0.8972& -0.0045 & 0.5531 \\
		        MV $\lambda=3$ &0.3173 & 0.8750& 0.0081 & 0.5658 \\
		        MV $\lambda=2$ & 0.2531 & 0.8108 & 0.0131 & 0.5707\\
		        MV $\lambda=1$ & 0.1472 & 0.7049 & 0.0104 & 0.5680\\
		        MV $\lambda=0.75$ & 0.1143 & 0.6720 & 0.0085 & 0.5662\\
		        MV $\lambda=0.5$ & 0.0788 & 0.6364& 0.0061 & 0.5638\\
		        MV $\lambda=0.05$ & 0.0083 & 0.5660& 0.0007 & 0.5584\\
    	    \end{tabular}}
	
	        \caption{ \label{tb:Results_NASDAQ_100} Performance of the trained portfolios over the out-of-sample period in terms of log-(relative)-wealth. We ran the mean-variance optimization task for seven different risk-tolerances $\lambda$. All but one of the optimized portfolios out-performed the universe portfolio over the out-of-sample period. }
            
        \end{table}
        
        In Figure~\ref{fig:Value_NASDAQ_100} and Figure~\ref{fig:RSIG_Value_NASDAQ_100} we display the wealth processes of the trained portfolio, the one of the universe portfolio and as a benchmark also of the equally-weighted portfolio. The wealth processes of the mean-variance portfolio are the more volatile, the higher the risk-tolerance is, as one would expect. Moreover, the wealth-processes of the randomized signature portfolios are more tamed than their JL-counterparts.
        
        We present the average values (in time) of the portfolio weights for each rank in Figure~\ref{fig:Avg_Weight_NASDAQ_100}. The trained portfolios mainly take over- and under-weighted positions in the largest stocks. While the positions are extremer, the higher the risk-tolerance. Although we do not enforce any long-only constraints or regularization, we do not observe any extreme short-selling positions.
        
        \begin{figure}[ht] 
		\centering
		    \begin{subfigure}{0.48\columnwidth}
		    \includegraphics[width=\linewidth]{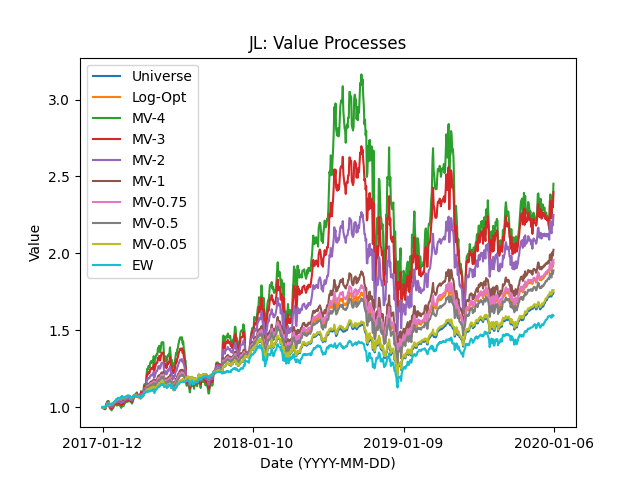}
			\caption{\label{fig:Value_NASDAQ_100} Wealth processes of the learned JL-signature portfolios. }
		    \end{subfigure}
		    \begin{subfigure}{0.48\columnwidth}
		    \includegraphics[width=\linewidth]{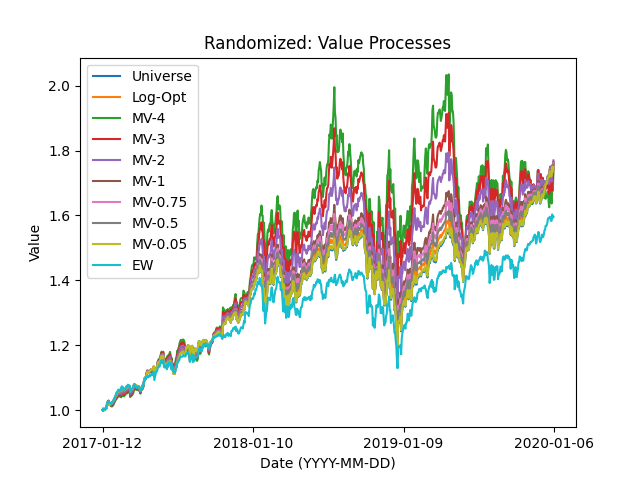}
			\caption{\label{fig:RSIG_Value_NASDAQ_100} Wealth processes of the randomized signature portfolios. }
		    \end{subfigure}
            \caption{ Wealth processes of NASDAQ the learned signature portfolios, the universe portfolio and the equally-weighted portfolio, over the out-of-sample period.}

	    \end{figure}
	   
	    \begin{figure}[ht] 
		\centering
		\begin{subfigure}{0.45\columnwidth}
		    \includegraphics[width=\linewidth]{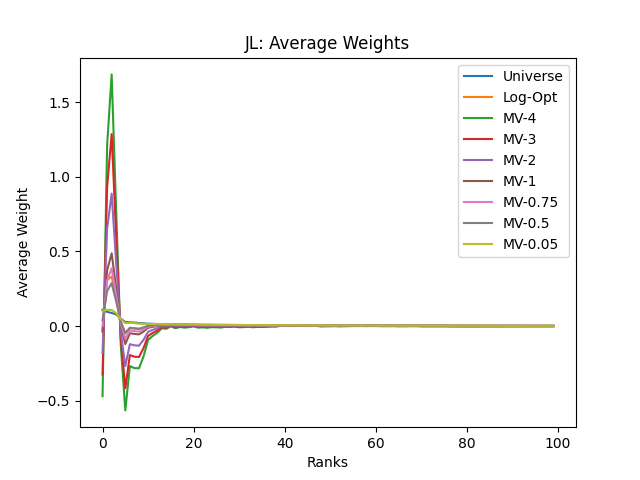}
			\caption{ Average weights of the JL-signature portfolios.}
		\end{subfigure}
        \begin{subfigure}{0.45\columnwidth}
		    \includegraphics[width=\linewidth]{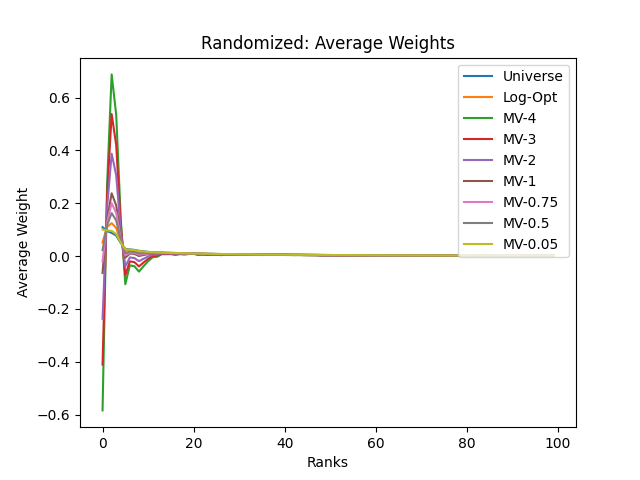}
			\caption{Average weights of the randomized signature portfolios.}
		\end{subfigure}
        \begin{subfigure}{0.45\columnwidth}
		    \includegraphics[width=\linewidth]{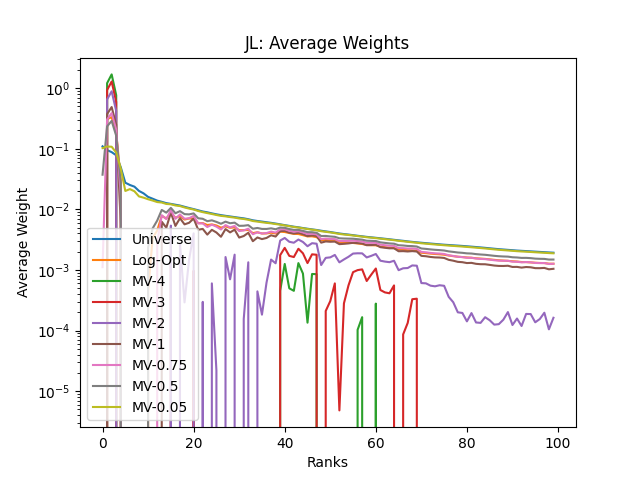}
			\caption{ \label{fig:JL_semilog}Plot of the logarithm of the average weights of the JL-signature portfolios.}
		\end{subfigure}
        \begin{subfigure}{0.45\columnwidth}
		    \includegraphics[width=\linewidth]{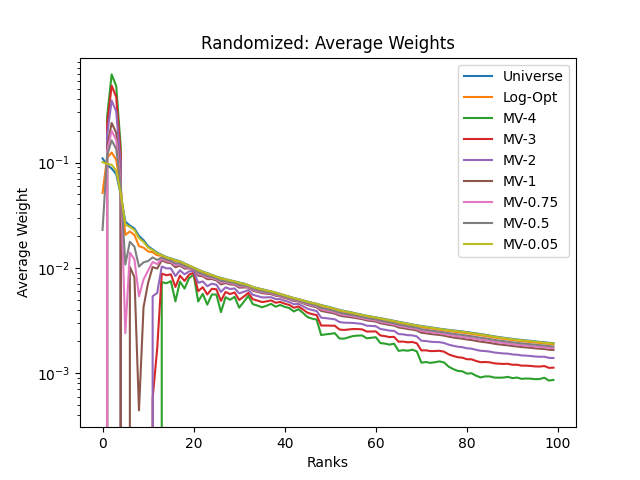}
			\caption{ \label{fig:RSIG_semilog}Plot of the logarithm of the average weights of the randomized-signature portfolios.}
		\end{subfigure}
  
		\caption{\label{fig:Avg_Weight_NASDAQ_100} Average weights in the NASDAQ of the signature portfolios and the universe portfolio, averaged over the out-of-sample period.}
		    
	    \end{figure}
    
    \subsubsection{Name-based approach: SMI and S\&P500}\label{subsec:SMI_SPX}
        In the name-based setting we tackle the mean-variance and log-relative wealth optimization problems under transaction costs. We do this in two markets; the Swiss Market Index (SMI) and S\&P500 Index. For the SMI, we consider the universe of stocks which survived between 2000-2022 and for the S\&P500 those that survived between 2001-2022. This amounts to 17 stock in the SMI universe and to 378 stocks in the S\&P500 universe. For both markets, we obtained the data from Reuters Datastream.\footnote{The raw/processed data required to reproduce our findings cannot be shared at this time due to legal reasons.} As inputs and as auxiliary portfolio we use the respective universe weights, i.e. $X= \tau= \mu^{\mathcal{U}}$. We compare the performance of the trained portfolios with the universe portfolio. Again, we choose $t_0=100$, an in-sample period of 2000 days and an out-of-sample period of 750 days. When we include the regularization cost, we choose $\beta_0=0.5$ and enforce a lower bound $\beta \geq 10^{-8}$. We want to emphasize, that even with transaction costs our portfolios are re-balanced daily.
        
        \paragraph{SMI Market}
        We train signature portfolios $\pi(\mu^{\mathcal{U}}, \hat{\mu}^{\mathcal{U}})$ of type $I$ with three configurations of feature maps 
        \begin{itemize}
            \item the true signature up to degree two;
            \item the JL-signature of dimension (30,2);
            \item the randomized signature of dimension 30.
        \end{itemize}  We set $\gamma_{MV}=10^{-8}$ and found $\gamma_{LO}= 2.02\cdot 10^{-3}$ for the true signature, $\gamma_{LO}= 2.42\cdot 10^{-3}$  for the JL-signature and $\gamma_{LO}= 1\cdot 10^{-2}$ for the randomized signature respectively during cross-validation, where the grid-search was performed over 100 equally-distant points in $[10^{-8}, 10^{-2}]$. We present the out-of-sample performance in terms of log-relative wealth in Tables~\ref{tb:Results_SMI},~\ref{tb:Results_SMI_JL},~\ref{tb:Results_SMI_RSIG} as well as the annualized Sharpe-ratios in Table~\ref{tb:SR_SMI}. The first three columns of Tables~\ref{tb:Results_SMI},~\ref{tb:Results_SMI_JL},~\ref{tb:Results_SMI_RSIG} show the performance of portfolios with and without transaction costs, where we added no regularization for transaction costs. The next four columns show the performance with and without transactions costs, but with a regularization for transaction costs at the respective level. The corresponding regularization parameters $\beta$ which we found during the in-sample training are shown in Table~\ref{tb:beta}. We note that the lower the risk-tolerance $\lambda$ the less regularization for transaction costs is needed. However, adding such a regularization, the signature portfolios trained under the mean-variance optimization often outperform the universe portfolio under transaction costs during the out-of-sample period. This is remarkable, since the universe portfolio does not pay any transaction costs. The portfolio trained under log-relative wealth optimization does not out-perform the universe portfolio under transaction costs for the true and JL-signature portfolios and neither do some of the portfolios trained using the JL- and randomized signature under mean-variance optimization with higher risk-tolerances. Nevertheless, the regularization for transaction cost proved effective, since all portfolios performed better with the regularization than without, under the respective level of transaction costs. 
        
        We show the wealth processes of the signature portfolios for the SMI universe in Figure~\ref{fig:SMI_value_processes}. It is obvious that the higher $\lambda$ is, the more volatile the portfolios are, which lead to some signature portfolios under-performing the universe portfolio under transaction costs during some parts of the out-of-sample period. The signature portfolio with $\lambda=0.05$, however, did rather well over the entire period even under transaction costs, which we highlight in Figures~\ref{fig:SMI_value_processes_TC1_onlyMV005} \&~\ref{fig:SMI_value_processes_TC5_onlyMV005}.

        Turning to the Sharpe-ratios in Table~\ref{tb:SR_SMI} we see that while the Sharpe-ratios of the optimized portfolios are often higher than the one of the universe portfolio without transaction costs, the universe portfolio is less often out-performed under transaction costs in terms of Sharpe-ratios. Only the mean-variance optimized portfolio with $\lambda=0.05$ achieves a consistently higher Sharpe ratio than both universe and equally-weighted portfolio. It is also worth highlighting, that often the Sharpe-ratio increases for higher transaction costs (and thereby higher regularization for transaction costs).

        \paragraph{S\&P500 Market}
        We trained JL-signature portfolios $\pi(\mu^{\mathcal{U}}, \hat{\mu}^{\mathcal{U}})$  of dimension (30,2) and type $I$ as well as randomized-signature portfolios $\pi(\mu^{\mathcal{U}}, \hat{\mu}^{\mathcal{U}})$  of dimension 30 and type $I$ in the S\&P500 market. For the mean-variance optimization, we set $\gamma_{MV}=10^{-6}$ and for obtained $\gamma_{LO}=1.02\cdot 10^{-4} $ for both the JL-signature portfolio and  the randomized signature portfolio during cross-validation over an equally-spaced grid of 100 points in $[10^{-6}, 10^{-2}]$. We report the out-of-sample results of the optimized portfolios in Table~\ref{tb:Results_SPX} \&~\ref{tb:Results_SPX_randomized} for the JL- and randomized-signature portfolios respectively. It is remarkable that many portfolios also out-performed the universe portfolio without any regularization for transaction costs. Indeed, in many cases the optimization yielded $\beta=10^{-8}$, as we report in Table~\ref{tb:beta}. In the cases where such a regularization was needed, the performance with transaction costs was always improved with the regularization. Apart from the log-wealth optimized portfolio under 5\% of transaction costs, the learned portfolio all out-performed the universe portfolio under transaction costs. We show the corresponding wealth processes in Figure~\ref{fig:SPX_value_processes}.

        The out-of-sample annualized Sharpe-ratios are reported in Table~\ref{tb:SR_SPX}. The mean-variance optimized portfolios manage to out-perform the universe and equally-weighted portfolios also in terms of Sharpe-ratios, even under transaction costs. Only the log-wealth optimized portfolio fails to achieve a higher Sharpe-ratio than the universe portfolio under transaction costs.

        \begin{figure}[ht] 
		\centering
		\begin{subfigure}{0.55\columnwidth}
		    \centering
		    \includegraphics[width=\linewidth]{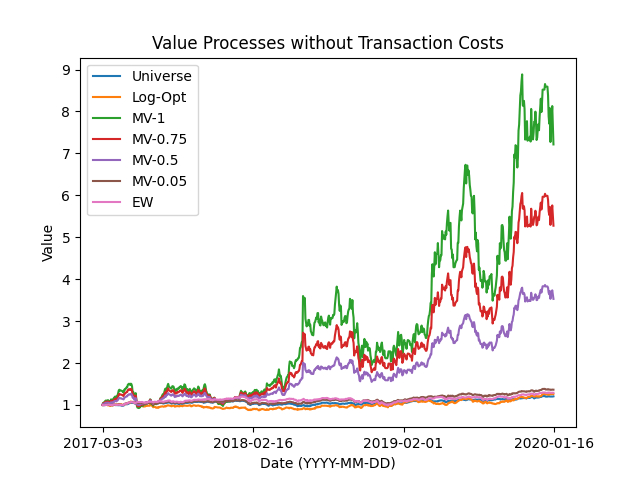}
		    \caption{\label{fig:SMI_value_processes_noTC}Wealth process without transaction costs}
	    \end{subfigure}\\
		\begin{subfigure}{0.48\columnwidth}
			\includegraphics[width=\linewidth]{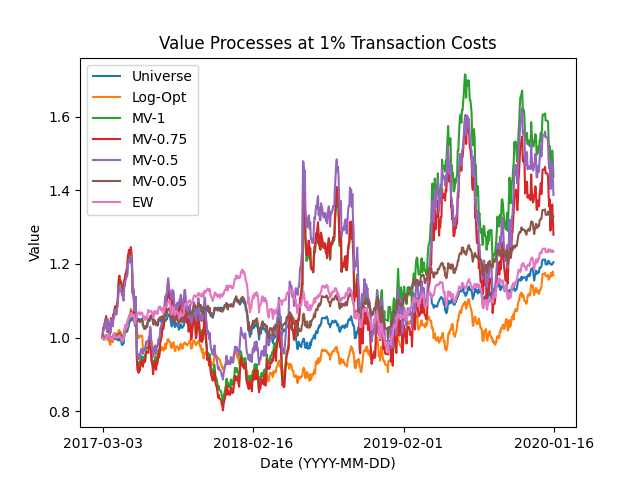}
			\caption{\label{fig:SMI_value_processes_TC1}Wealth process with 1\% proportional transaction costs}
		\end{subfigure}
		\begin{subfigure}{0.48\columnwidth}
			\includegraphics[width=\linewidth]{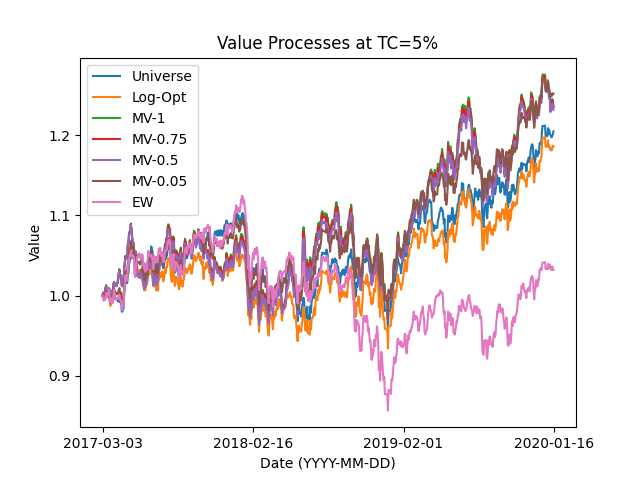}
			\caption{\label{fig:SMI_value_processes_TC5}Wealth process with 5\% proportional transaction costs}
		\end{subfigure}
		
		\begin{subfigure}{0.48\columnwidth}
			\includegraphics[width=\linewidth]{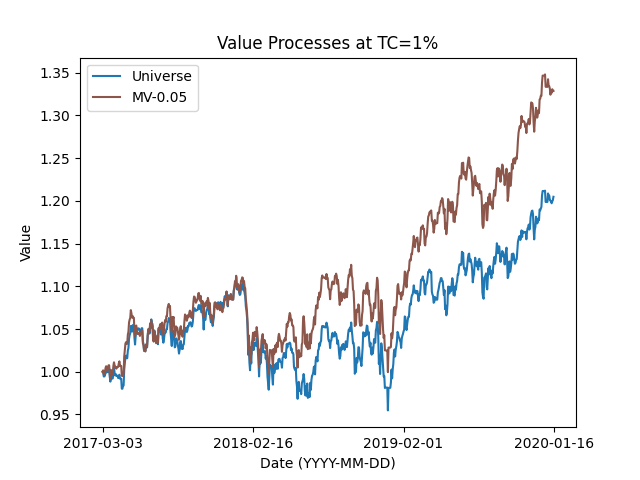}
			\caption{\label{fig:SMI_value_processes_TC1_onlyMV005}Wealth process with 1\% proportional transaction costs}
		\end{subfigure}
		\begin{subfigure}{0.48\columnwidth}
			\includegraphics[width=\linewidth]{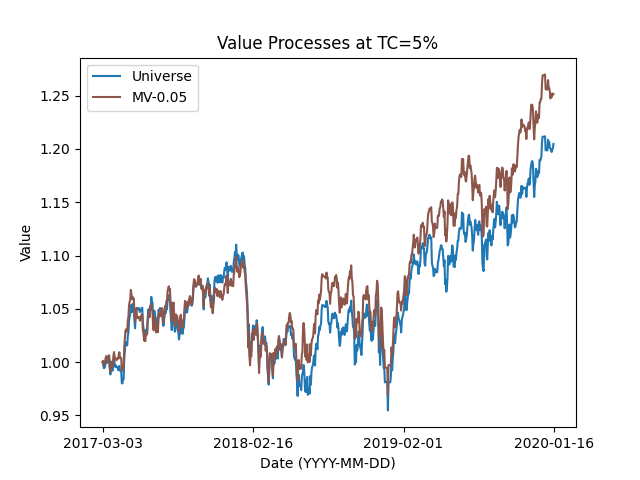}
			\caption{\label{fig:SMI_value_processes_TC5_onlyMV005}Wealth Process with 5\% Proportional Transaction Costs}
		\end{subfigure}
		\caption{ \label{fig:SMI_value_processes} SMI wealth processes of the signature portfolios, the universe portfolio and the equally-weighted portfolio with and without transaction costs. Note that signature portfolios in the settings with transaction cost are those, where we included the regularization for transaction costs. }
    \end{figure}

    \paragraph{Conclusion}
    We would like to offer some conclusion on the performance of our optimized portfolios in particular comparing the signature, JL-signature and randomized signature portfolios. We highlight that computing the true signature portfolio of the S\&P500 universe would not have been feasible, since the number of stocks would have been too high. However, the JL-signature portfolios did consistently achieve a better out-of-sample performance both in terms of log-relative wealth and Sharpe-ratios than the randomized signature portfolios. For the SMI universe we are impressed with the performance of the JL- and randomized-signature portfolios compared to those of the true signature portfolio, especially due to the vast difference in number of optimization parameters (510 for JL- \& rand.-signature each vs. 10'290 for the true signature portfolios!). In terms of Sharpe-ratio the JL- and randomized-signature portfolios even slightly out-performed the true signature portfolios, but no significant difference between the two randomization approached was observable. In terms of log-relative wealth, the true signature portfolio performed better than the JL-signature portfolios, which in turn achieved better results than the randomized-signature ones, however we would like to point out that the differences were small. Overall, in terms of feasibility and performance, we consider the JL-signature portfolios to be the most favorable.

    \begin{figure}[ht] 
		\centering
		
		\begin{subfigure}{1\columnwidth}
		\begin{subfigure}{0.48\columnwidth}
			\includegraphics[width=\linewidth]{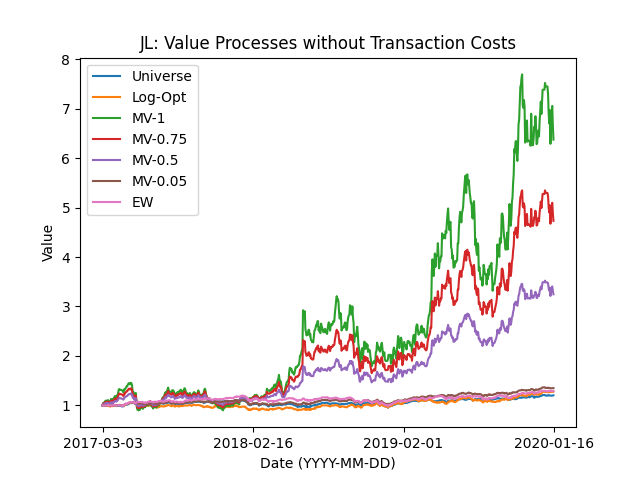}
		\end{subfigure}
		\begin{subfigure}{0.48\columnwidth}
			\includegraphics[width=\linewidth]{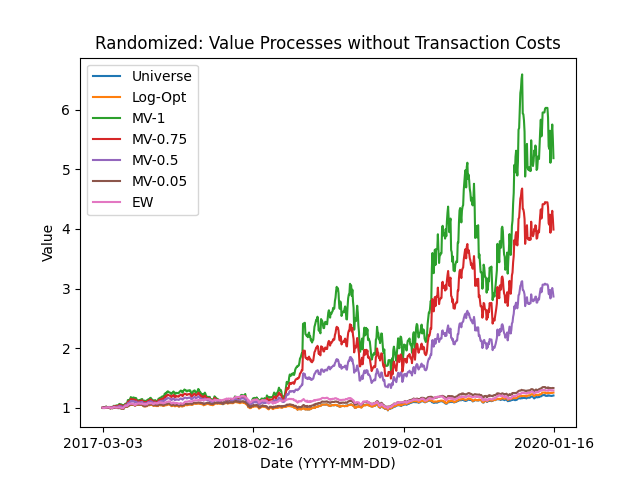}
		\end{subfigure}
		\caption{\label{fig:SMI_value_processes_noTC_rand}Wealth process without transaction costs for JL-signature portfolios (left) and randomized-signature portfolios (right)}
		\end{subfigure}
		
		\begin{subfigure}{1\columnwidth}
		\begin{subfigure}{0.48\columnwidth}
			
			\includegraphics[width=\linewidth]{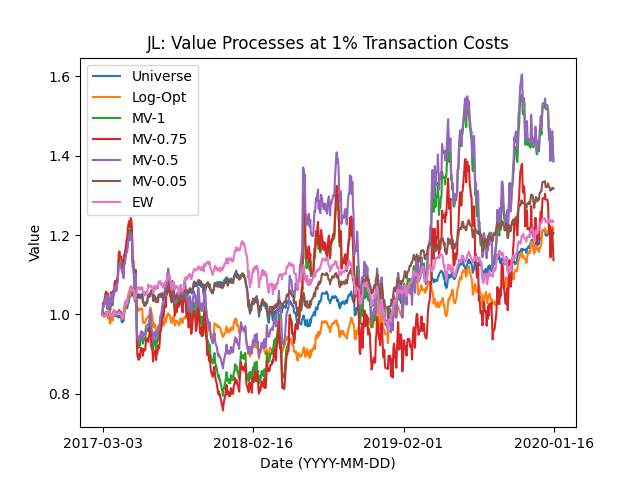}
		\end{subfigure}
		\begin{subfigure}{0.48\columnwidth}
			\includegraphics[width=\linewidth]{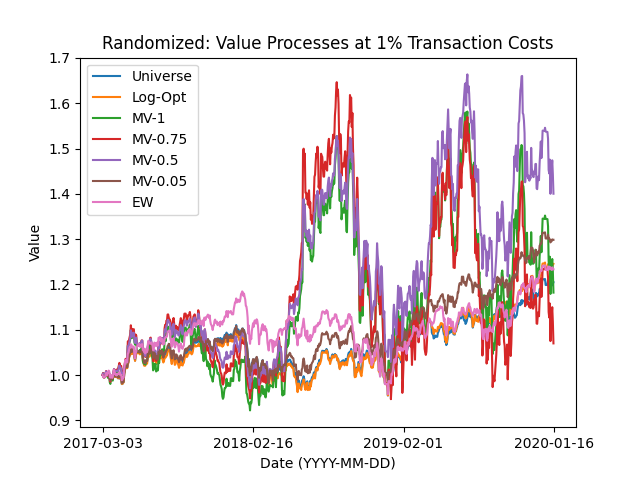}
		\end{subfigure}
		\caption{\label{fig:SMI_value_processes_TC1_rand}Wealth process with 1\% proportional transaction costs for JL-signature portfolios (left) and randomized-signature portfolios (right)}
		\end{subfigure}
		
		\begin{subfigure}{1\columnwidth}
		\begin{subfigure}{0.48\columnwidth}
			\includegraphics[width=\linewidth]{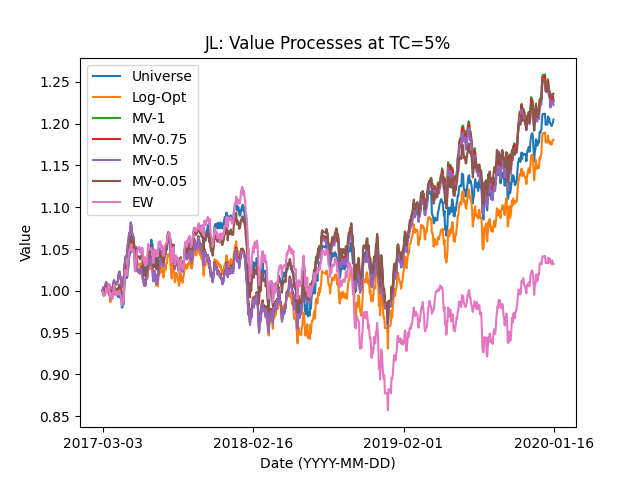}
		\end{subfigure}
		\begin{subfigure}{0.48\columnwidth}
			\includegraphics[width=\linewidth]{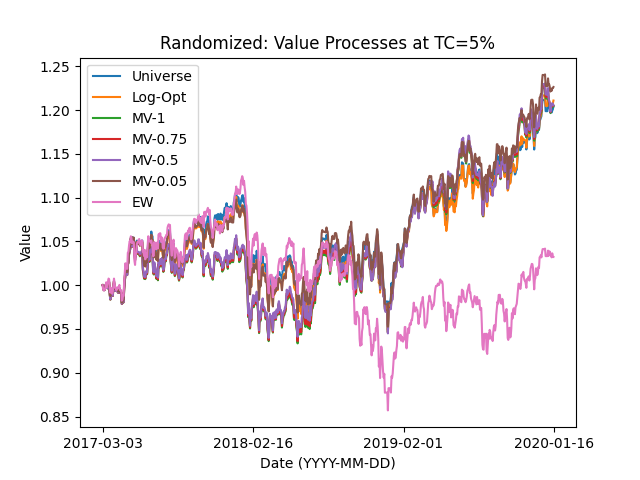}
		\end{subfigure}
		\caption{\label{fig:SMI_value_processes_TC5_rand}Wealth process with 5\% proportional transaction costs for JL-signature portfolios (left) and randomized-signature portfolios (right)}
		\end{subfigure}
		
		\caption{ \label{fig:SMI_value_processes_rand} SMI wealth processes of the JL-signature portfolios (left), randomized-signature portfolios (right), universe portfolio and equally-weighted portfolio with and without transaction costs. Note that signature portfolios in the settings with transaction cost are those, where we included the regularization for transaction costs. }
	\end{figure}

    \begin{table}
	   \centering
	        \begin{subtable}[h!]{\textwidth}
                \centering
            {\footnotesize
            \begin{tabular}{ c | c | c | c | c | c| c | c} 
	    	    Opt. Task & \multicolumn{3}{|c|}{$\log{V_T}$ without Reg. }& \multicolumn{2}{|c|}{$\log{V_T}$ with Reg. $\beta_{1\%}$ }&\multicolumn{2}{|c}{$\log{V_T}$ with Reg. $\beta_{5\%}$}\\ 
	    	    Prop. TC & 0\% & 1\% & 5\% & 0\% & 1\% & 0\% & 5\% \\\hline
		         Log-Opt. & 0.0458 & -0.0494   & -0.4333   & 0.0461  & -0.0293 &0.0540  & -0.0158  \\ 
		        MV $\lambda=1$ & 1.7896 & -1.9354  & -$\infty$ & 1.0514 & 0.1760 & 0.0257 & 0.2804\\
		        MV $\lambda=0.75$ & 1.4767 & -0.6214  & -14.6792 & 1.1547 &0.0601  & 0.2724 & 0.0247\\
		        MV $\lambda=0.5$ & 1.0748& 0.1155 &-4.3536  &1.0539& 0.1409 & 0.2620 & 0.0220   \\
		        MV $\lambda=0.05$ &0.1240& 0.0973  & -0.0096  & 0.1229 & 0.0979 & 0.1056 &0.0381  \\
    	    \end{tabular}}
            \subcaption{\label{tb:Results_SMI} Log-relative wealth of the signature portfolios with and without transaction costs trading in the SMI universe. }
            \end{subtable}
	        
	        \vspace{0.36cm}
	        
	        \begin{subtable}[h!]{\textwidth}
                \centering
            {\footnotesize
            \begin{tabular}{ c | c | c | c | c | c| c | c} 
	    	    Opt. Task & \multicolumn{3}{|c|}{$\log{V_T}$ without Reg. }& \multicolumn{2}{|c|}{$\log{V_T}$ with Reg. $\beta_{1\%}$ }&\multicolumn{2}{|c}{$\log{V_T}$ with Reg. $\beta_{5\%}$}\\ 
	    	    Prop. TC & 0\% & 1\% & 5\% & 0\% & 1\% & 0\% & 5\% \\\hline
		         Log-Opt. & 0.0620 & 0.0093   & -0.2009   & 0.0619  & 0.0094 &0.0420  & -0.0201\\ 
		        MV $\lambda=1$ & 1.6656 & -1.6336  & -$\infty$ & 0.9654 & 0.1402 & 0.2061 & 0.0187\\
		        MV $\lambda=0.75$ & 1.3676 & -0.4919  & -12.2094 & 1.1439 &-0.0583  & 0.2051 & 0.0178\\
		        MV $\lambda=0.5$ & 0.9913& 0.1398 &-3.7521  &0.9913& 0.1398 & 0.2074 & 0.0144   \\
		        MV $\lambda=0.05$ &0.1137& 0.0894  & -0.0072  & 0.1137& 0.0894 & 0.0947 &0.0254  \\
    	    \end{tabular}}
            \subcaption{\label{tb:Results_SMI_JL} Log-relative wealth of the JL-signature portfolios with and without transaction costs trading in the SMI universe. }
            \end{subtable}
	        
	        \vspace{0.36cm}

         	\begin{subtable}[h!]{\textwidth}
                \centering
            {\footnotesize
            \begin{tabular}{ c | c | c | c | c | c| c | c} 
	    	    Opt. Task & \multicolumn{3}{|c|}{$\log{V_T}$ without Reg. }& \multicolumn{2}{|c|}{$\log{V_T}$ with Reg. $\beta_{1\%}$ }&\multicolumn{2}{|c}{$\log{V_T}$ with Reg. $\beta_{5\%}$}\\ 
	    	    Prop. TC & 0\% & 1\% & 5\% & 0\% & 1\% & 0\% & 5\% \\\hline
		         Log-Opt. & 0.0397 & 0.0331   & 0.0070   & 0.0396  & 0.0331 &0.0217  & 0.0050  \\ 
		        MV $\lambda=1$ & 1.4602 & -1.3840  & -$\infty$ & 0.9378  & -0.0191 & 0.1334 & -0.0005\\
		        MV $\lambda=0.75$ & 1.1976 & -0.4161  & -11.9226 & 1.0101 &-0.1188  & 0.1381 & 0.0004\\
		        MV $\lambda=0.5$ & 0.8668& 0.1176 &-3.3104  &0.8157& 0.1504 & 0.1498 & 0.0006   \\
		        MV $\lambda=0.05$ &0.0991& 0.0752  & -0.0195  & 0.0970& 0.0749 & 0.0659 &0.0177  \\
    	    \end{tabular}}
            \subcaption{\label{tb:Results_SMI_RSIG} Log-relative wealth of the randomized-signature portfolios with and without transaction costs trading in the SMI universe. }
            \end{subtable}
	        
	        \vspace{0.36cm}

	        \begin{subtable}[h!]{\textwidth}
                \centering
                {\footnotesize
            \begin{tabular}{ c | c | c | c | c | c| c | c} 
	    	    Opt. Task & \multicolumn{3}{|c|}{$\log{V_T}$ without Reg. }& \multicolumn{2}{|c|}{$\log{V_T}$ with Reg. $\beta_{1\%}$}&\multicolumn{2}{|c}{$\log{V_T}$ with Reg.  $\beta_{5\%}$}\\ 
	    	    Prop. TC & 0\% & 1\% & 5\% & 0\% & 1\% & 0\% & 5\% \\\hline
		        Log-Opt. & 1.7442 & -2.5700  & -$\infty$  &0.9282   & 0.0402 & 0.2287  & -0.1564    \\ 
		        MV $\lambda=1$ & 0.3089 & 0.2657  & 0.0944 & 0.3089 &0.2657   & 0.2292 &0.0987 \\
		        MV $\lambda=0.75$ & 0.2342 & 0.2092 & 0.1101  & 0.2342 & 0.2092 & 0.2179 &0.1078\\
		        MV $\lambda=0.5$ & 0.1578  & 0.1460  &0.0990 & 0.1578  & 0.1460 & 0.1578 &0.0990   \\
		        MV $\lambda=0.05$ & 0.0161 & 0.0156  & 0.0137  & 0.0161 & 0.0156 & 0.0161 &0.0137\\

    	    \end{tabular}}
            \subcaption{\label{tb:Results_SPX} Log-relative wealth of the JL-signature portfolios with and without transaction costs trading in the S\&P500 universe. }
            \end{subtable}
	        
	        \vspace{0.36cm}
	        
            \begin{subtable}[h!]{\textwidth}
                \centering
                {\footnotesize
            \begin{tabular}{ c | c | c | c | c | c| c | c} 
	    	    Opt. Task & \multicolumn{3}{|c|}{$\log{V_T}$ without Reg. }& \multicolumn{2}{|c|}{$\log{V_T}$ with Reg. $\beta_{1\%}$}&\multicolumn{2}{|c}{$\log{V_T}$ with Reg.  $\beta_{5\%}$}\\ 
	    	    Prop. TC & 0\% & 1\% & 5\% & 0\% & 1\% & 0\% & 5\% \\\hline
		        Log-Opt. & 1.0296  & 0.1415  & -3.9870    & 0.7823  & 0.2411 & 0.1789  & -0.0820   \\ 
		        MV $\lambda=1$ & 0.0988 & 0.0903 & 0.0566  &  0.0988 & 0.0903 & 0.0988 & 0.0566 \\
		        MV $\lambda=0.75$ & 0.0743 & 0.0686  & 0.0455 & 0.0743 & 0.0686 &  0.0743 & 0.0455 \\
		        MV $\lambda=0.5$ & 0.0497 & 0.0462 &0.0319  & 0.0497 & 0.0462 & 0.0497 & 0.0319   \\
		        MV $\lambda=0.05$ &  0.0050 & 0.0047 & 0.0033  & 0.0050 & 0.0047  & 0.0050 & 0.0033\\

    	    \end{tabular}}
            \subcaption{\label{tb:Results_SPX_randomized} Log-relative wealth of the randomized-signature portfolios with and without transaction costs trading in the S\&P500 universe. }
            \end{subtable}
	 
	        \caption{ The first three columns show the performance with and without transaction costs of signature portfolios with no regularization for transaction costs. The next four column show the performance with and without transaction costs of portfolios trained with a regularization for transaction costs at the respective level. }
        \end{table}

    \begin{table}
    
	   \centering
	        \begin{subtable}[h!]{\textwidth}
                \centering
            {\footnotesize
            \begin{tabular}{ c | c | c | c | c | c| c | c| c | c} 
	    	    Portfolio & \multicolumn{3}{|c|}{Signature Portfolio }& \multicolumn{3}{|c|}{JL-Signature Portfolio }&\multicolumn{3}{|c}{Rand.-Signature Portfolio}\\ 
	    	    Prop. TC & 0\% & 1\% & 5\% & 0\% & 1\% & 5\%& 0\% & 1\% & 5\% \\\hline
                Universe & 0.6061 & 0.6061  & 0.6061   & 0.6061  & 0.6061 &0.6061  & 0.6061 & 0.6061& 0.6061  \\ 
                Equally-weighted &\textbf{0.7487} & \textbf{0.6287}   & 0.1482   & \textbf{0.7487} & \textbf{0.6287} &0.1482  &\textbf{ 0.7487} &\textbf{0.6287} & 0.1482 \\ 
		         Log-Opt. & 0.5786 & 0.4353   & 0.55   & \textbf{0.6735}  & 0.5459 &0.5389  & \textbf{0.7246} & \textbf{0.7052} &  \textbf{0.6239}\\ 
		        MV $\lambda=1$ & \textbf{1.282} & 0.5191  & 0.5998 & \textbf{1.263} & 0.4941 & \textbf{0.6144} & \textbf{1.2029}& 0.3374&0.5805\\
		        MV $\lambda=0.75$ & \textbf{1.3051} & 0.4052  & 0.6004 & \textbf{1.2873} &0.3081  & \textbf{0.6125} & \textbf{1.2237}& 0.2585&0.5812\\
		        MV $\lambda=0.5$ & \textbf{1.3396}& 0.4849 &0.5979 &\textbf{1.3226}& 0.4906 & 0.6016 & \textbf{1.2513}  &0.5202 & 0.5773\\
		        MV $\lambda=0.05$ &\textbf{0.946}& \textbf{0.8718}  & \textbf{0.7073} & \textbf{0.9211}& \textbf{0.8515} & \textbf{0.6726} &\textbf{0.8707} & \textbf{0.8028}& \textbf{0.6499}\\
    	    \end{tabular}}
            \subcaption{\label{tb:SR_SMI} Annualized Sharpe-ratios of the signature portfolios, universe portfolio and equally weighted portfolio with and without transaction costs trading in the SMI universe. Note that for the results with transaction costs, the corresponding regularization for transaction costs was included in the training. }
            \end{subtable}
	        
	        \vspace{0.36cm}
	        
	        \begin{subtable}[h!]{\textwidth}
                \centering
            {\footnotesize
            \begin{tabular}{ c | c | c | c | c | c| c } 
	    	    Portfolio & \multicolumn{3}{|c|}{JL-Signature Portfolio }&\multicolumn{3}{|c}{Rand.-Signature Portfolio}\\ 
	    	    Prop. TC & 0\% & 1\% & 5\%& 0\% & 1\% & 5\% \\\hline
                Universe & 0.8530 & 0.8530  & 0.8530   & 0.8530  & 0.8530 &0.8530    \\ 
                Equally-weighted &\textbf{0.9292} & 0.7466   & 0.0164   & \textbf{0.9292} & 0.7466 &0.0164  \\ 
		         Log-Opt. & \textbf{1.223} &  0.4901   & 0.3653   & \textbf{1.1904}  & 0.6777 &0.5330\\ 
		        MV $\lambda=1$ & \textbf{1.1817} & \textbf{1.1040}  & \textbf{0.8753} & \textbf{1.0077} & \textbf{0.9879} & \textbf{0.9087} \\
		        MV $\lambda=0.75$ & \textbf{1.1402} & \textbf{1.0905}  & \textbf{0.9048} & \textbf{0.9754} &\textbf{0.9614}  & \textbf{0.9054} \\
		        MV $\lambda=0.5$ & \textbf{1.0781}& \textbf{1.0521} &\textbf{0.9488} & \textbf{0.9391} & \textbf{0.9302} & \textbf{0.8945} \\
		        MV $\lambda=0.05$ &\textbf{0.8834}& \textbf{0.8821}  & \textbf{0.8772} & \textbf{0.8625}& \textbf{0.8616} & \textbf{0.8581} \\
    	    \end{tabular}}
            \subcaption{\label{tb:SR_SPX} Annualized Sharpe-ratios of the signature portfolios, universe portfolio and equally weighted portfolio with and without transaction costs trading in the  S\&P500 universe. Note that for the results with transaction costs, the corresponding regularization for transaction costs was included in the training. }
            \end{subtable}
            
            \caption{The tables show the annualized out-of-sample Sharpe-ratios of the trained signature portfolios as well as universe and equally-weighted portfolio in the SMI and SPX universe respectively. The numbers marked in bold are those that are higher than the Sharpe-ratio of the respective universe portfolio.}
        \end{table}

    \begin{table}
	        \centering
            {\footnotesize
	        \begin{tabular}{ c | c | c | c | c | c | c | c | c | c | c} 
	    	    Universe & \multicolumn{2}{|c|}{SMI} &\multicolumn{2}{|c|}{SMI (JL)} &\multicolumn{2}{|c|}{SMI (rand.)}  &  \multicolumn{2}{|c|}{S\&P500 (JL)} &  \multicolumn{2}{|c}{S\&P500 (rand.)}\\
	    	    Beta &  $\beta_{1\%}$ & $\beta_{5\%}$ & $\beta_{1\%}$ & $\beta_{5\%}$ & $\beta_{1\%}$ & $\beta_{5\%}$ & $\beta_{1\%}$& $\beta_{5\%}$ & $\beta_{1\%}$& $\beta_{5\%}$ \\\hline
		        Log-Opt. & 27.958  & 847.61 & 0.5& 619.60 & 0.4999& 960.77& 9.6437 & 97.160 & 2.9908 & 65.285 \\ 
		        MV $\lambda=1$ & 0.0879 & 1.1661 &0.0609 & 0.9611& 0.0193& 0.5510& $10^{-8} $ & 0.0059 & $10^{-8} $& $10^{-8} $ \\
		        MV $\lambda=0.75$ & 0.0240 & 0.8063 & 0.0101& 0.6598& 0.0040& 0.3692& $10^{-8} $ & 0.0012 & $10^{-8} $& $10^{-8} $\\
		        MV $\lambda=0.5$ & 0.0008  & 0.4537 &$10^{-8} $ & 0.3645& 0.0010& 0.1954 &$10^{-8} $ & $10^{-8} $ & $10^{-8} $& $10^{-8} $ \\
		        MV $\lambda=0.05$ & 0.0003 & 0.0095 &$10^{-8} $ &0.0076 & 0.0003& 0.0111& $10^{-8} $  & $10^{-8} $ & $10^{-8} $& $10^{-8} $\\

    	    \end{tabular}}
	        \caption{ \label{tb:beta} Optimal regularization parameters $\beta$ found during in-sample training. }
    \end{table}
        
    \begin{figure}[ht] 
		\centering
		
		\begin{subfigure}{1\columnwidth}
		\begin{subfigure}{0.48\columnwidth}
			\includegraphics[width=\linewidth]{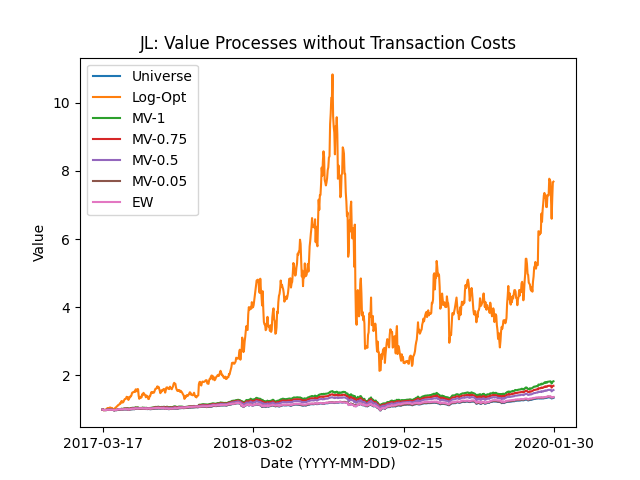}
		\end{subfigure}
		\begin{subfigure}{0.48\columnwidth}
			\includegraphics[width=\linewidth]{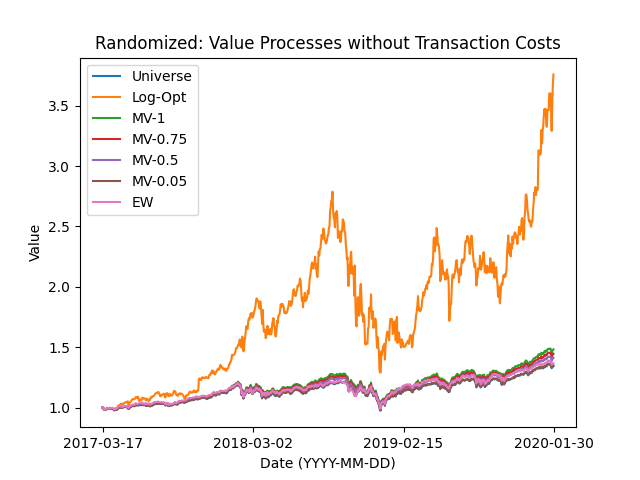}
		\end{subfigure}
		\caption{\label{fig:SPX_value_processes_noTC}Wealth process without transaction costs for JL-signature portfolios (left) and randomized-signature portfolios (right)}
		\end{subfigure}
		
		\begin{subfigure}{1\columnwidth}
		\begin{subfigure}{0.48\columnwidth}
			\includegraphics[width=\linewidth]{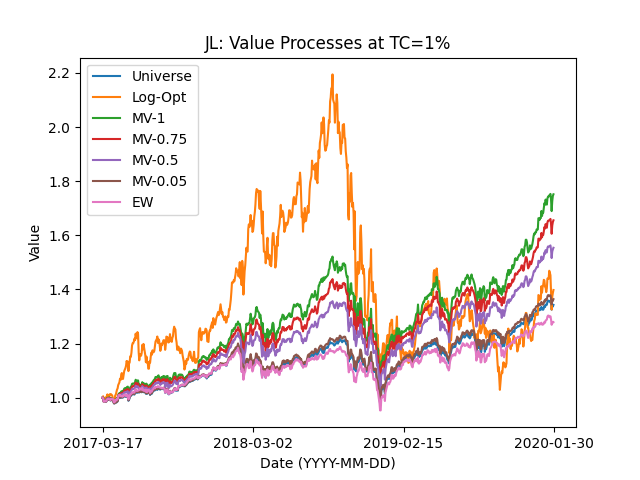}
		\end{subfigure}
		\begin{subfigure}{0.48\columnwidth}
			\includegraphics[width=\linewidth]{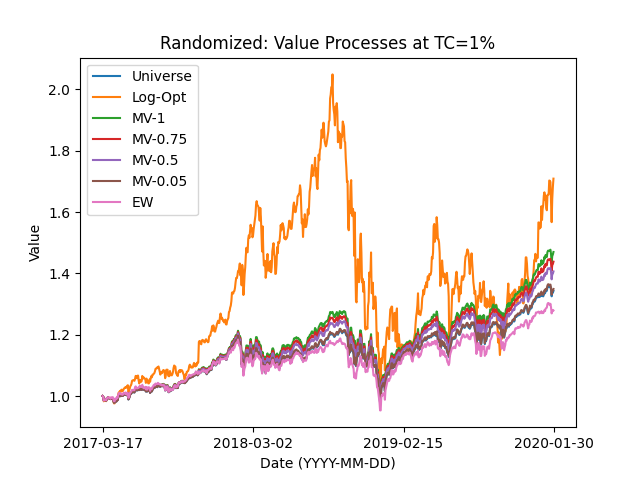}
		\end{subfigure}
		\caption{\label{fig:SPX_value_processes_TC1}Wealth process with 1\% proportional transaction costs for JL-signature portfolios (left) and randomized-signature portfolios (right)}
		\end{subfigure}
		
		\begin{subfigure}{1\columnwidth}
		\begin{subfigure}{0.48\columnwidth}
			\includegraphics[width=\linewidth]{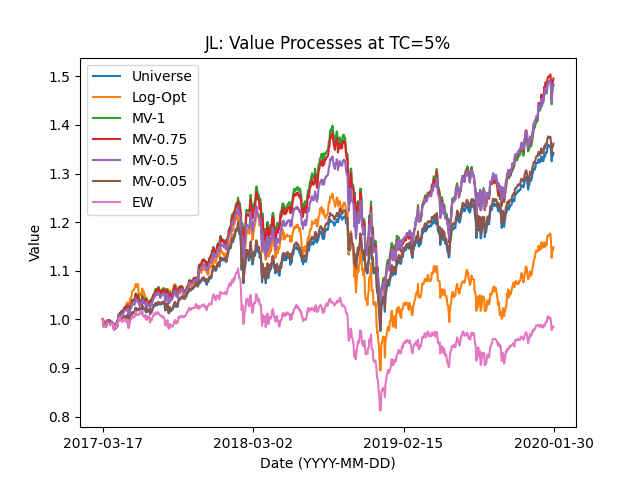}
		\end{subfigure}
		\begin{subfigure}{0.48\columnwidth}
			\includegraphics[width=\linewidth]{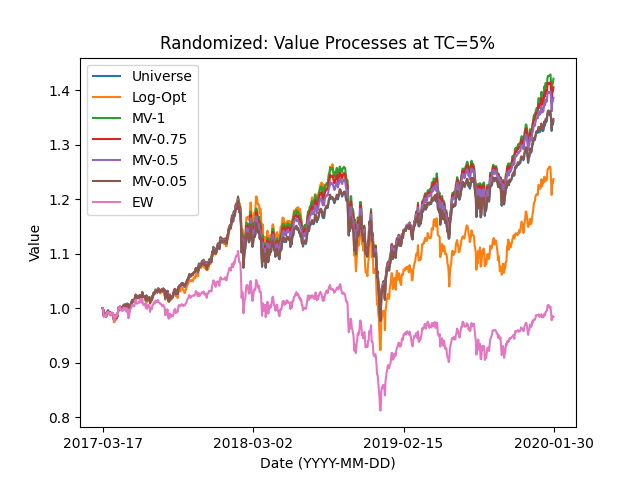}
		\end{subfigure}
		\caption{\label{fig:SPX_value_processes_TC5}Wealth process with 5\% proportional transaction costs for JL-signature portfolios (left) and randomized-signature portfolios (right)}
		\end{subfigure}
		
		\caption{ \label{fig:SPX_value_processes} S\&P500 wealth processes of the JL-signature portfolios (left), randomized-signature portfolios (right), universe portfolio and equally-weighted portfolio with and without transaction costs. Note that signature portfolios in the settings with transaction cost are those, where we included the regularization for transaction costs. }
	\end{figure}

\FloatBarrier

\newpage

\appendix \label{appendix}

\section{Proofs of Section~\ref{sec:Foundations}} \label{app:Proofs_Foundations}

Since the space of lifted stopped paths as introduced in Definition \ref{def:stoppedlifted} contains paths defined on different time-intervals, we need the following definition.
    
    \begin{definition}[Projection on the final value]
         We denote by $(\cdot)_{\odot}$ the projection to the value at the final time, i.e. for each $t\in[0,T]$ $$(x_{[0,t]})_{\odot} := x_t.$$
    \end{definition}
    
For this projection on the final value we have the following continuity result.

    \begin{lemma}\label{lem:lambda_continuity}
        The function 
        \begin{align*}
            \lambda : \,[0, T] \times \mathcal{G}_T^N(\varphi, x) &\rightarrow \Lambda_T^N(\varphi, x) \notag\\
            (t,\hat{\mathbb{X}}^N_{[0,T]}(\omega)) &\mapsto \hat{\mathbb{X}}^N_{[0,t]}(\omega)
        \end{align*}
        is continuous.
    \end{lemma}

\begin{proof}
    Take an arbitrary but fixed $(s, \hat{ \mathbb{Y}}^N_{[0,T]}(\omega_1))$. We show that $\lambda$ is continuous at $(s, \hat{ \mathbb{Y}}^N_{[0,T]}(\omega_2))$ but since this point is arbitrary, $\lambda$ is continuous everywhere. For simplicity, we write $\hat{\mathbf{x}}_{[0,t]}= \hat{\mathbb{X}}_{[0,t]}^N(\omega_2)$,  $\hat{\mathbf{y}}_{[0,t]}= \hat{ \mathbb{Y}}_{[0,t]}^N(\omega_1)$ and $\widehat{\mathbf{y}^s}_{[0,t]}= \widehat{\,  \mathbb{Y}^s}_{[0,t]}^N(\omega_1)$. For any $(t, \hat{\mathbf{x}}_{[0,T]})\in [0,T] \times \mathcal{G}_T^N$ it holds that  
    \begin{align*}
        &d_{\Lambda}\left( \lambda\left( t, \hat{\mathbf{x}}_{[0,T]} \right), \lambda\left( s, \hat{\mathbf{y}}_{[0,T]}  \right)\right) \\
        & \quad= d_{\Lambda}\left(  \hat{\mathbf{x}}_{[0,t]} ,  \hat{\mathbf{y}}_{[0,s]}\right) \\ &\quad \leq d_{\Lambda}\left(  \hat{\mathbf{x}}_{[0,t]},  \hat{\mathbf{y}}_{[0,t]} \right)+ d_{\Lambda}\left(  \hat{\mathbf{y}}_{[0,t]} ,  \hat{\mathbf{y}}_{[0,s]} \right)
        \\ &\quad =  d_{p-var;[0,t]}\left( \hat{\mathbf{x}}_{[0,t]},  \hat{\mathbf{y}}_{[0,t]} \right) + |t-s| + d_{p-var;[0,t\vee s]}\left( \widehat{\, \mathbf{y}^t}_{[0,t\vee s]},  \widehat{\mathbf{y}^s}_{[0,t\vee s]} \right)\\
        &\quad \leq d_{p-var;[0,t]}\left( \hat{\mathbf{x}}_{[0,t]},  \hat{\mathbf{y}}_{[0,t]} \right)+ |t-s| + \|\hat{\mathbf{y}} \|_{p-var; [t\wedge s, t \vee s]},
    \end{align*}
    where $\|\cdot \|_{p-var; [0,T]}$ is the norm induced by $d_{p-var;[0,T]}$.
Let now $\epsilon >0$
and choose $\delta>0$ such that 
\[
\delta +\|\hat{\mathbf{y}} \|_{p-var; [s-\delta, s+\delta]} < \epsilon.
\]    
 Note that this is possible because $\|\hat{\mathbf{y}} \|_{p-var; [s-\delta, s+\delta]} \to 0$ as $\delta \to 0$.
  
       Therefore, for all $(t, \hat{\mathbf{x}}_{[0,T]})\in [0,T] \times \mathcal{G}_T^N$ (equipped with the product topology) satisfying
    \begin{align*}
        &|t-s|+ d_{p-var;[0,T]}\left(\hat{\mathbf{x}}_{[0,T]},  \hat{\mathbf{y}}_{[0,T]} \right)< \delta,
          \end{align*}
          we have
    \begin{align*}      
        |t-s| + d_{p-var;[0,t]}\left( \hat{\mathbf{x}}_{[0,t]},  \hat{\mathbf{y}}_{[0,t]} \right) + \|\hat{\mathbf{y}} \|_{p-var; [t\wedge s, t \vee s]} 
         & < \delta + \|\hat{\mathbf{y}} \|_{p-var; [t\wedge s, t \vee s]} \\ 
       &  \leq \delta +\|\hat{\mathbf{y}} \|_{p-var; [s-\delta, s+\delta]} \\
       & < \epsilon,
    \end{align*}
    which proves continuity of $\lambda$.

\end{proof}

The following lemma can be proved by means of Corollary \ref{cor:lyons_lift}.

    \begin{lemma}\label{lem:e_I_continuous}
        For each multiindex $I$ the map
     $$\Lambda_T^2(\varphi, x) \ni \hat{\mathbb{X}}^2_{[0,t]}(\omega) \mapsto \langle e_I, S^{|I|}(\hat{\mathbb{X}}^2_{[0,t]}(\omega))_{\odot} \rangle = \langle e_I,\hat{\mathbb{X}}^N_{t}(\omega) \rangle
     \in \mathbb{R} $$         
        is continuous on bounded sets. 
    \end{lemma}

\begin{proof}
Let us show the statement for an arbitrary but fixed multiindex~$I$.
For $N=2$ we apply the same notation as in the proof of Lemma~\ref{lem:lambda_continuity} above, i.e.  $\hat{\mathbf{x}}_{[0,t]}= \hat{\mathbb{X}}_{[0,t]}^2(\omega_x)$, etc.
We show continuity at some fixed  $\hat{\mathbf{y}}_{[0,s]}$ in some bounded set of $\Lambda_T^2(\varphi, x)$.
Recall that by Corollary \ref{cor:lyons_lift} 
$$ \hat{\mathbf{x}}_{[0,t]}\mapsto \langle e_I, S^{|I|}(\hat{\mathbf{x}}_{[0,t]}) \rangle $$
        is continuous on bounded sets of $\mathcal{G}_t^2(\varphi, x)$       
for each $t\in[0,T]$. 
Thus, for every $\epsilon > 0$ there exists some $\delta_1>0$ such that for all $ \hat{\mathbf{x}}_{[0, t]}$ satisfying 
\[
d_{p-var, [0, t]} (\hat{\mathbf{x}}_{[0, t]}, \hat{\mathbf{y}}_{[0,t]}) < \delta_1
\]
we have
$$\|\langle e_I, S^{|I|}(\hat{\mathbf{x}}_{[0,t]}) \rangle -\langle e_I, S^{|I|}(\hat{\mathbf{y}}_{[0,t]} )\rangle\|_{p-var,[0,t]} < \frac{\epsilon}{2}.
$$
Here, $\| \cdot \|_{p-var, [0,t]}$ denotes the $p$-variation distance on $\mathbb{R}$.
Moreover, 
the map $$[0,T] \ni t\mapsto \langle e_I, S^{|I|}(\hat{\mathbf{y}}_{[0,t]})_{\odot} \rangle \in \mathbb{R}$$
        is continuous for each fixed $\hat{\mathbf{y}}$. Hence, there exists some $\delta_2>0$ such that for all $t$ with $|t-s|< \delta_2$
        \[
   |\langle e_I, S^{|I|}(\hat{\mathbf{y}}_{[0,t]})_{\odot} \rangle -\langle e_I, S^{|I|}(\hat{\mathbf{y}}_{[0,s]})_{\odot} \rangle| < \frac{\epsilon}{2}.    
        \]
Choose now $\delta=\min(\delta_1, \delta_2)$ and let $\hat{\mathbf{x}}_{[0,t]}$ (in the bounded set) 
be such that
        \begin{align*}
        \delta > d_{\Lambda}(\hat{\mathbf{x}}_{[0,t]},\hat{\mathbf{y}}_{[0,s]})&= |t-s| +d_{p-var, [0,t \vee s]} (\widehat{\mathbf{x}^t}_{[0,t \vee s]}, \widehat{\mathbf{y}^s}_{[0,t \vee s]})\\ &\geq |t-s|+d_{p-var, [0, t]} (\hat{\mathbf{x}}_{[0, t]}, \widehat{\mathbf{y}^{t \wedge s}}_{[0,t]}).
           \end{align*} Then, by the above
            \begin{align*}
           | \langle e_I, S^{|I|}(\hat{\mathbf{x}}_{[0,t]})_{\odot} \rangle - \langle e_I, S^{|I|}(\hat{\mathbf{y}}_{[0,s]})_{\odot}  \rangle | \leq   \, & |\langle e_I, S^{|I|}(\hat{\mathbf{x}}_{[0,t]})_{\odot} \rangle -\langle e_I, S^{|I|}(\widehat{\mathbf{y}^{t \wedge s}}_{[0,t]})_{\odot} \rangle| \\& + |\langle e_I, S^{|I|}(\widehat{\mathbf{y}^{t \wedge s}}_{[0,t]})_{\odot} \rangle -\langle e_I, S^{|I|}(\hat{\mathbf{y}}_{[0,s]})_{\odot} \rangle| \\
        &\leq    \|\langle e_I, S^{|I|}(\hat{\mathbf{x}}_{[0,t]}) \rangle -\langle e_I, S^{|I|}(\widehat{\mathbf{y}^{t \wedge s}}_{[0,t]})\rangle\|_{p-var,[0,t]}\\
      &\quad +  |\langle e_I, S^{|I|}(\hat{\mathbf{y}}_{[0,t]})_{\odot} \rangle -\langle e_I, S^{|I|}(\hat{\mathbf{y}}_{[0,s]})_{\odot} \rangle|\\
      & <\epsilon,      
        \end{align*}
which proves the assertion.

    \end{proof}

We are now prepared to prove Theorem \ref{thm:UAT_sig}.

    \begin{proof}[Proof of Theorem~\ref{thm:UAT_sig}]
    Consider the function $\lambda$ as introduced  in Lemma \ref {lem:lambda_continuity}
        and note that for any compact subset $K \subset \mathcal{G}_T^2(\varphi, x)$ the set $$\mathcal{K}_{\Lambda}=\lambda([0,T]\times K) = \left\{\hat{\mathbb{X}}^2_{[0,t]}(\omega) \mid \, t \in [0,T] \textrm{ and }  \hat{\mathbb{X}}^2_{[0,T]}(\omega)\in K\right\}$$ is a compact subset of $\Lambda_T^2(\varphi, x)$ due to the continuity of $\lambda$ as proved in Lemma \ref {lem:lambda_continuity}. For each $\hat{\mathbb{X}}^2_{[0,t]}(\omega) \in \Lambda_T^2(\varphi, x)$ and $N \geq 2$ it holds that
        $S^N(\hat{\mathbb{X}}^2_{[0,t]}(\omega))_{\odot}= \hat{\mathbb{X}}^N_t(\omega)$ a.s. Hence, we can associate with every linear function on the signature $L$ a non-anticipative path-functional $f_L : \hat{\mathbb{X}}^2_{[0,t]}(\omega) \mapsto L(\hat{\mathbb{X}}_t(\omega))$. Moreover, for every linear function on the signature $L$ it holds that $f_L \in C(\Lambda_T^2(\varphi, x); \mathbb{R})$, which follows from Lemma~\ref{lem:e_I_continuous}. Let us denote by 
        \begin{align*}
            \mathfrak{F}_{\mathfrak{L}} := \left\{ f_L \mid f_L : \hat{\mathbb{X}}^2_{[0,t]}(\omega) \mapsto L(\hat{\mathbb{X}}_t(\omega)) \textrm{ for } \hat{\mathbb{X}}^2_{[0,t]}(\omega)\in \Lambda_T^2(\varphi, x), \, L \in \mathfrak{L} \right\}
        \end{align*}
        the set of path functionals induced by linear functions on the signature and by 
        \begin{align*}
            \mathfrak{F}_{\mathfrak{L}}\rvert_{\mathcal{K}_{\Lambda}} := \left\{ f_L \mid f_L : \hat{\mathbb{X}}^2_{[0,t]}(\omega) \mapsto L(\hat{\mathbb{X}}_t(\omega)) \textrm{ for } \hat{\mathbb{X}}^2_{[0,t]}(\omega)\in \mathcal{K}_{\Lambda}, \, L \in \mathfrak{L} \right\}
        \end{align*}
        its restriction to lifted paths in $\mathcal{K}_{\Lambda}$. We will now apply the Stone-Weierstrass Theorem to show that $\mathfrak{F}_{\mathfrak{L}}\rvert_{\mathcal{K}_{\Lambda}}$ is dense in $C(\mathcal{K}_{\Lambda}; \mathbb{R})$. To this end we need to show that $\mathfrak{F}_{\mathfrak{L}}\rvert_{\mathcal{K}_{\Lambda}}$ is a sub-algebra of $C(\mathcal{K}_{\Lambda}; \mathbb{R})$ which separates points and vanishes nowhere. The algebra property follows from the fact that for each $\hat{\mathbb{X}}^2_{[0,t]}(\omega) \in \Lambda_T^2(\varphi, x)$ and for any two $f_{L_1}, f_{L_2} \in \mathfrak{F_L}$ it holds that
        \begin{align*}
            &\left(f_{L_1}\cdot f_{L_2} \right)\left(\hat{\mathbb{X}}^2_{[0,t]}(\omega)\right) = f_{L_1}\left(\hat{\mathbb{X}}^2_{[0,t]}(\omega)\right)f_{L_2}\left(\hat{\mathbb{X}}^2_{[0,t]}(\omega)\right)\\&\quad = L_1 \left(\hat{\mathbb{X}}_t(\omega)\right) L_2\left(\hat{\mathbb{X}}_t(\omega)\right) = L\left(\hat{\mathbb{X}}_t(\omega)\right)=f_{L}\left(\hat{\mathbb{X}}^2_{[0,t]}(\omega)\right),
        \end{align*} 
        where the second last equality follows from the shuffle product property. We now show that $\mathfrak{F}_{\mathfrak{L}}$ separates points. To this end take $\hat{\mathbb{X}}^2_{[0,t]}(\omega_x), \hat{\mathbb{Y}}^2_{[0,s]}(\omega_y)  \in \Lambda_T^2(\varphi, x)$
        with $\hat{\mathbb{X}}^2_{[0,t]}(\omega_x)\neq\hat{\mathbb{Y}}^2_{[0,s]}(\omega_y)  $. If $t\neq s$, the two points are separated in the time-component, i.e. $\langle e_1, \hat{\mathbb{X}}_t(\omega_x)  \rangle \neq \langle e_1,\hat{\mathbb{Y}}_s(\omega_y) \rangle$. If $t=s$, the point-separation follows from Lemma~\ref{lem:uniqueness_sig}. The algebra vanishes nowhere because $f_{\emptyset}(\hat{\mathbb{X}}^2_{[0,t]}(\omega))= \langle e_{\emptyset}, \hat{\mathbb{X}}_t(\omega) \rangle \equiv 1$ for all $ \hat{\mathbb{X}}^2_{[0,t]}(\omega)\in \Lambda^2_T$, where $\emptyset$ denotes the empty word.  The desired density statement follows by the Stone-Weierstrass Theorem, i.e. for any continuous non-anticipative path-functional $f: \Lambda_T^2(\varphi,x) \rightarrow \mathbb{R}$ and for every $\epsilon> 0$ there exists a path-functional $f_L\in \mathfrak{F_L}$ such that for almost all $\omega \in \Omega$

        $$ \sup_{(t,  \hat{\mathbb{X}}^2_{[0,T]}(\omega)) \in [0,T] \times K} |f(\hat{\mathbb{X}}^2_{[0,t]}(\omega)) - f_L(\hat{\mathbb{X}}_t)(\omega)| < \epsilon,$$ 
        which is equivalent to the claim of Theorem~\ref{thm:UAT_sig}.
    \end{proof}

\begin{proof}[Proof of Lemma \ref{lem:weighted}]
We have to prove that $K_R= \psi^{-1}((0,R]) $ is compact for any $R>0$. Let us first note that the map $\lambda$ can be easily extended to  weakly geometric $\alpha$-H\"older rough paths and is continuous again by the exact same arguments as before. Moreover, note that the set $$K_R'= \left\{ \widehat{ \, \mathbf{X}^t}_{[0,T]} \Bigg\lvert \, \psi\left(\hat{\mathbf{X}}_{[0,t]}\right) \in [0,R) \right\}$$
        equipped with the $d_{\alpha',[0,T]}$-topology is for any $R>0$ a $\|\cdot\|_{\alpha, [0,T]}$-bounded subset of $\hat{C}^{\alpha}_T(\varphi,x)$ and hence compact for $0\leq \alpha'< \alpha$ due to the compact embedding of $\alpha$-H\"older spaces into $\alpha'$-H\"older spaces, see \cite[Theorem A.3]{Cuchiero_Schmocker_Teichmann_2023}. Moreover, $K_R=\lambda([0,T]\times K_R')$ and is therefore compact for any $R>0$ due to Tychonoff's theorem and the continuity of $\lambda$. 
\end{proof}

\begin{proof}[Proof of Theorem \ref{thm:GUAT}]
        
        Using Lemma \ref{lem:weighted} the proof is analogous to the proof of \cite[Theorem 5.4]{Cuchiero_Schmocker_Teichmann_2023}. In fact it even holds that 
        \begin{equation*}
        \sup_{\hat{\mathbf{X}}_{[0,t]} \in \,\bigcup_{t\in [0,T]} \hat{C}^{\alpha}_t(\varphi, x) } \frac{\lvert  f(\hat{\mathbf{X}}_{[0,t]}) - L(\hat{\mathbb{X}}_t)\rvert}{\psi(\hat{\mathbf{X}}_{[0,t]})} < \epsilon
        \end{equation*} 
        and since $\Lambda^{2}_T(\varphi, x) \subset \bigcup_{[0,t]} \hat{C}^{\alpha}_t(\varphi, x)$ for every $\frac{1}{3}< \alpha \frac{1}{2}$ the claim follows. 
    \end{proof}

	\section{Proofs of Section~\ref{sec:approx_sig_port}}\label{app:proofs_approx_sig_port}

The following lemma is needed in the proof of Theorem~\ref{thm:approx_go}.

    \begin{lemma}\label{lem:matrix_inv}
		Let $D$ be a (subset of a) metric space. Consider a matrix-valued function $M: D\ni x \mapsto M(x) $ whose coefficients are continuous functions from $D$ to $\mathbb{R}$ and which is non-singular for all $x\in D$, then the coefficients of its matrix-inverse $M^{-1}$ are again continuous functions. 
	\end{lemma}

	\begin{proof}
	   Let us prove the statement by induction:  
	   \begin{itemize}
	       \item Let $M: D \rightarrow \mathbb{R}^{2\times2}$. Then the inverse of $$ M= \begin{pmatrix} a & b \\ c & d 
	       \end{pmatrix} $$
	       is given by $$ M^{-1}= \frac{1}{(ad-bc)}\begin{pmatrix} d & -b \\ -c & a
	       \end{pmatrix} $$
	       By assumption the determinant $(ad-bc)$ is non-zero, which implies continuity of the coefficients of $M^{-1}$. 
	       \item Assume the statement holds for an arbitrary but fixed $n\geq 2$.
	       \item Let $M: D \rightarrow \mathbb{R}^{(n+1)\times (n+1)}$. Then we can write $M$ in block matrix form $$ M = \begin{pmatrix} \tilde{M} & m^1 \\  m^2& {m^3}\end{pmatrix},$$
where $ \tilde{M}$ takes values in $ \mathbb{R}^{n\times n}$, $m^1, (m^2)^{\top}$ in  $\mathbb{R}^n$	and $m^3$ in $\mathbb{R}$.
       The inverse is then computed via	     $$ M^{-1} = \begin{pmatrix} \tilde{M}^{-1}+ \frac{1}{\Delta}\tilde{M}^{-1} m^1 \ m^2 \tilde{M}^{-1} & -\frac{1}{\Delta}\tilde{M}^{-1} m^1  \\  \frac{1}{\Delta} m^2 \tilde{M}^{-1} & \frac{1}{\Delta}\end{pmatrix}$$
	       where $\Delta=m^3-m^2 \tilde{M}^{-1} m^1 $ with 
	     values in $\mathbb{R}\setminus \{0\}$ by assumption on invertability. 
	    As (the components of) $\tilde{M}^{-1}, m^1, m^2$ and $m^3$  are all continuous functions,  continuity of the components of $M^{-1}$ follows as well.
	   \end{itemize}
	   
	\end{proof}

	\section{Proofs of Section~\ref{sec:optimizing_signature_port}}
	\label{app:optimizing_signature_port}
	\begin{proof}[Proof of Lemma~\ref{lem:rel_val_process}]
	    By Equation~\eqref{eq:rel_val_process_universe}, we know that 
	    \begin{align}
	        &\log\left( \frac{W^{\pi}_t}{W^{{\mathcal{U}}}_t}\right) = \sum_{i \in \mathcal{U}} \int_{t_0}^t \gamma_s^i f^i(s, {X}_{[0,s]}) d \mu_s^{\mathcal{U},i} \nonumber \\ &- \frac{1}{2} \sum_{i,j \in \mathcal{U}} \int_{t_0}^t \gamma_s^i \gamma_s^j f^i(s, {X}_{[0,s]}) f^j(s, {X}_{[0,s]}) d[\mu^{\mathcal{U},i},\mu^{\mathcal{U},j}]_s \nonumber \\& + \label{eq:zero_term1} \sum_{i \in \mathcal{U}} \int_{t_0}^t (1- \sum_{k=1}^d \lambda_s^kf^k(s, {X}_{[0,s]}) )d \mu_s^{\mathcal{U},i} \\ &- \frac{1}{2} \sum_{i,j \in \mathcal{U}} \label{eq:zero_term2} \int_{t_0}^t (1- \sum_{k=1}^d \lambda_s^kf^k(s, {X}_{[0,s]}) ) (1- \sum_{m=1}^d \lambda_s^mf^m(s, {X}_{[0,s]}) ) d[\mu^{\mathcal{U},i},\mu^{\mathcal{U},j}]_s,
	    \end{align}
	    for 
	    \begin{equation*}
	        \lambda_s^k = \begin{cases} \mu_s^{\mathcal{U},k} &\textrm{ if $\pi$ is of type $I$}  \\ 1 &\textrm{ if $\pi$ is of type $II$} 
	        \end{cases}.
	    \end{equation*}
	    Note that the terms in \eqref{eq:zero_term1} and \eqref{eq:zero_term2} are both zero due to  $\sum_{i\in\mathcal{U}} \mu_t^{\mathcal{U}, i} \equiv 1$. 
	        
	\end{proof}

	\section{Proofs of Section~\ref{sec:transcost}}\label{app:transcosts}
	
	\begin{proof}[Proof of Proposition~\ref{prop:transcosts}]
	    \begin{enumerate}
	
	        \item This case always holds for long-only portfolios and is discussed in~\cite{Ruf_Xie_19}. The argument in the long-only case is that $L(1)< R(1)$ and $L(0)> R(0)$, while both sides are obviously continuous and monotone. While the former statement still holds in the case with short-selling, the monotonicity is no longer true. However, one can show by a tedious case-by-case study that where $R'(\alpha)$ exists it is piece-wise constant and increasing. Therefore a unique solution for $\alpha^*\in[0,1]$ is still assured in this case. 
	        
	        \item Here the above argument does not hold anymore. This is because now $L(0)\leq R(0)$, while it still holds that $L(1)< R(1)$. Let us give three examples, one for each of the cases a) no solution b) no solution in $[0,1]$, c) no unique solution in $[0,1]$. 
	        \begin{enumerate}[label=(\alph*)]
	            \item \emph{no solution}: consider a market of two stocks and $c=0.05$, i.e. 5\% transaction costs. For the portfolio weights $$ \pi_{t^-} = (13.1, -12.1 )^{\mathsf{T}}, \hspace{1cm}\pi_t = (13, -12 )^{\mathsf{T}}$$
	            there does not exist a solution for $\alpha$.

	            \item \emph{no solution in $[0,1]$}: again for $c=0.05$ and a market of two stocks. For the portfolio weights $$ \pi_{t^-} = (11, -10 )^{\mathsf{T}}, \hspace{1cm}\pi_t = (10, -9 )^{\mathsf{T}}$$
	            the solution is $\alpha=-1$. 
	            
	            \item \emph{no unique solution in $[0,1]$}: consider $c=0.05$
	            and a market of three stocks with portfolio weights $$ \pi_{t^-} = (5,6, -10 )^{\mathsf{T}}, \hspace{1cm}\pi_t = (5.5, 6.5,-11 )^{\mathsf{T}}$$ 
	            then the solutions are $\alpha=0.3333$ and $\alpha=0.9535$ (numbers are rounded).
	            
	        \end{enumerate}
	    \end{enumerate}
	\end{proof}


\begin{thebibliography}{10}
	
	\bibitem{Gambara_2022}
	E.~Akyildirim, M.~Gambara, J.~Teichmann, and S.~Zhou.
	\newblock Applications of signature methods to market anomaly detection.
	\newblock {\em Preprint arXiv:2201.02441}, 2022.
	
	\bibitem{Gambara_2023}
	E.~Akyildirim, M.~Gambara, J.~Teichmann, and S.~Zhou.
	\newblock Randomized signature methods in optimal portfolio selection.
	\newblock {\em Preprint arxiv:2312.16448}, 2023.
	
	\bibitem{Promel_1}
	A.~Allan, C.~Cuchiero, C.~Liu, and D.~Prömel.
	\newblock Model-free portfolio theory: A rough path approach.
	\newblock {\em Math. Finance}, 33:709--765, 2023.
	
	\bibitem{andres2022signature}
	H.~Andr{\`e}s, A.~Boumezoued, and B.~Jourdain.
	\newblock Signature-based validation of real-world economic scenarios.
	\newblock {\em ASTIN Bulletin}, 24:410--440, 2024.
	
	\bibitem{Banner}
	A.~Banner and R.~Ghomrasni.
	\newblock Local times of ranked continuous semimartingales.
	\newblock {\em Stochastic Process. Appl.}, 118:1244--1253, 2008.
	
	\bibitem{Bayer20}
	C.~Bayer, P.~Hager, S.~Riedel, and J.~Schoenmakers.
	\newblock Optimal stopping with signatures.
	\newblock {\em Ann. Appl. Probab.}, 33(1), 2021.
	
	\bibitem{Benth_2022}
	F.~E. Benth, N.~Detering, and L.~Galimberti.
	\newblock Neural networks in fr\'echet spaces.
	\newblock {\em Ann. Math. Artif. Intell.}, 91:75–103, 2023.
	
	\bibitem{Horatio}
	H.~Boedihardjo, X.~Geng, T.~Lyons, and D.~Yang.
	\newblock The signature of a rough path: Uniqueness.
	\newblock {\em Adv. Math.}, 293:720--737, 2016.
	
	\bibitem{Buhler}
	H.~B{\"u}hler, B.~Horvath, T.~Lyons, I.~P. Arribas, and B.~Wood.
	\newblock A data-driven market simulator for small data environments, 2020.
	
	\bibitem{Wong}
	S.~Campbell and T.-K.~L. Wong.
	\newblock Functional portfolio optimization in stochastic portfolio theory.
	\newblock {\em SIAM J. Financial Math.}, 13(2):576--618, 2022.
	
	\bibitem{CAO2018278}
	W.~Cao, X.~Wang, Z.~Ming, and J.~Gao.
	\newblock A review on neural networks with random weights.
	\newblock {\em Neurocomputing}, 275:278--287, 2018.
	
	\bibitem{chevyrev22}
	I.~Chevyrev and H.~Oberhauser.
	\newblock Signature moments to characterize laws of stochastic processes.
	\newblock {\em J. Mach. Learn. Res.}, 23(176):1--42, 2022.
	
	\bibitem{Compagnoni_2022}
	E.~M. Compagnoni, L.~Biggio, A.~Orvieto, T.~Hofmann, and J.~Teichmann.
	\newblock Randomized signature layers for signal extraction in time series
	data.
	\newblock {\em Preprint arxiv:2201.00384}, 2022.
	
	\bibitem{Cont_2010}
	R.~Cont and D.~Fournie.
	\newblock A functional extension of the ito formula.
	\newblock {\em Comptes Rendus Mathematique}, 348(1):57--61, 2010.
	
	\bibitem{Cuchiero_Poly_SPT}
	C.~Cuchiero.
	\newblock Polynomial processes in stochastic portfolio theory.
	\newblock {\em Stochastic Process. Appl.}, 129(5):1829--1872, 2019.
	
	\bibitem{cuchiero2023joint}
	C.~Cuchiero, G.~Gazzani, J.~M{\"o}ller, and S.~Svaluto-Ferro.
	\newblock Joint calibration to spx and vix options with signature-based models.
	\newblock {\em Math. Finance}, page 1–53, 2024.
	
	\bibitem{Cuchiero_Svaluto-Ferro_Gazzani}
	C.~Cuchiero, G.~Gazzani, and S.~Svaluto-Ferro.
	\newblock Signature-based models: Theory and calibration.
	\newblock {\em SIAM J. Financial Math.}, 14(3):910--957, 2023.
	
	\bibitem{Cuchiero_Teichmann}
	C.~Cuchiero, L.~Gonon, L.~Grigoryeva, J.-P. Ortega, and J.~Teichmann.
	\newblock Discrete-time signatures and randomness in reservoir computing.
	\newblock {\em IEEE Trans. Neural Netw. Learn. Syst.}, 33(11):6321--6330, 2021.
	
	\bibitem{cuchiero2021expressive}
	C.~Cuchiero, L.~Gonon, L.~Grigoryeva, J.-P. Ortega, and J.~Teichmann.
	\newblock Expressive power of randomized signature.
	\newblock In {\em The Symbiosis of Deep Learning and Differential Equations},
	2021.
	
	\bibitem{Cuchiero_Primavera_Svaluto-Ferro2022}
	C.~Cuchiero, F.~Primavera, and S.~Svaluto-Ferro.
	\newblock Universal approximation theorems for continuous functions of
	c\`adl\`ag paths and l\'evy-type signature models.
	\newblock {\em Preprint arxiv:2208.02293}, 2022.
	
	\bibitem{christa}
	C.~Cuchiero, W.~Schachermayer, and T.-K.~L. Wong.
	\newblock Cover's universal portfolio, stochastic portfolio theory and the
	numeraire portfolio.
	\newblock {\em Math. Finance}, 29:773–803, 2019.
	
	\bibitem{Cuchiero_Schmocker_Teichmann_2023}
	C.~Cuchiero, P.~Schmocker, and J.~Teichmann.
	\newblock Global universal approximation of functional input maps on weighted
	spaces.
	\newblock {\em Preprint arxiv:2306.03303}, 2023.
	
	\bibitem{Cuchiero_Svaluto-Ferro_Teichmann2023}
	C.~Cuchiero, S.~Svaluto-Ferro, and J.~Teichmann.
	\newblock Signature sdes from an affine and polynomial perspective.
	\newblock {\em Preprint arxiv:2302.01362}, 2023.
	
	\bibitem{delbaen1994general}
	F.~Delbaen and W.~Schachermayer.
	\newblock A general version of the fundamental theorem of asset pricing.
	\newblock {\em Math. Ann.}, 300(1):463--520, 1994.
	
	\bibitem{Dupire_2019}
	B.~Dupire.
	\newblock Functional itô calculus.
	\newblock {\em Quant. Finance}, 19(5):721--729, 2019.
	
	\bibitem{E:16}
	A.~Eberle.
	\newblock Markov processes.
	\newblock {\em Lecture Notes at University of Bonn}, 2016.
	
	\bibitem{Fern02}
	R.~Fernholz.
	\newblock {\em Stochastic Portfolio Theory}.
	\newblock Springer, 2002.
	
	\bibitem{Karatzas_Vol_Stab}
	R.~Fernholz and I.~Karatzas.
	\newblock Relative arbitrage in volatility-stabilized markets.
	\newblock {\em Ann. Finance}, 1(2):149--177, 2005.
	
	\bibitem{Fernholz_Karatzas}
	R.~Fernholz and I.~Karatzas.
	\newblock Stochastic portfolio theory: an overview, 2010.
	
	\bibitem{Platen}
	D.~Filipovic and E.~Platen.
	\newblock Consistent market extensions under the benchmark approach.
	\newblock {\em Math. Finance}, 19:41--52, 2009.
	
	\bibitem{Friz_Victoir}
	P.~K. Friz and N.~Victoir.
	\newblock {\em Multidimensional Stochastic Processes as Rough Paths: Theory and
		Applications}.
	\newblock Cambridge University Press, 2010.
	
	\bibitem{Nakajima_Theo}
	K.~Fujii and K.~Nakajima.
	\newblock Harnessing disordered-ensemble quantum dynamics for machine learning.
	\newblock {\em Physical Review Applied}, 8(2), 2017.
	
	\bibitem{Blanka_23}
	O.~Futter, B.~Horvath, and M.~Wiese.
	\newblock Signature trading: A path-dependent extension of the mean-variance
	framework with exogenous signals.
	\newblock {\em Preprint arxiv: 2308.15135}, 2023.
	
	\bibitem{gallicchio2021}
	C.~Gallicchio and S.~Scardapane.
	\newblock Deep randomized neural networks.
	\newblock {\em Preprint arxiv:2002.12287}, 2021.
	
	\bibitem{Gonon}
	L.~Gonon and J.-P. Ortega.
	\newblock Fading memory echo state networks are universal.
	\newblock {\em Neural networks : the official journal of the International
		Neural Network Society}, 138:10--13, 2021.
	
	\bibitem{Ortega_Echo}
	L.~Grigoryeva and J.-P. Ortega.
	\newblock Echo state networks are universal.
	\newblock {\em Neural networks : the official journal of the International
		Neural Network Society}, 108:495--508, 2018.
	
	\bibitem{Ortega_SAS}
	L.~Grigoryeva and J.-P. Ortega.
	\newblock Universal discrete-time reservoir computers with stochastic inputs
	and linear readouts using non-homogeneous state-affine systems.
	\newblock {\em J. Mach. Learn. Res.}, 19(24):1--40, 2018.
	
	\bibitem{Herrera_Krach}
	C.~Herrera, F.~Krach, P.~Ruyssen, and J.~Teichmann.
	\newblock Optimal stopping via randomized neural networks.
	\newblock {\em Frontiers of Mathematical Finance}, 3(1):31--77, 2024.
	
	\bibitem{Huang2006}
	G.-B. Huang, L.~Chen, and C.-K. Siew.
	\newblock Universal approximation using incremental constructive feedforward
	networks with random hidden nodes.
	\newblock {\em Trans. Neur. Netw.}, 17(4):879–892, 2006.
	
	\bibitem{Schweizer_Hulley}
	H.~Hulley and M.~Schweizer.
	\newblock {\em M6---On Minimal Market Models and Minimal Martingale Measures},
	pages 35--51.
	\newblock Springer Berlin Heidelberg, 2010.
	
	\bibitem{itkin2022ergodic}
	D.~Itkin, B.~Koch, M.~Larsson, and J.~Teichmann.
	\newblock Ergodic robust maximization of asymptotic growth under stochastic
	volatility.
	\newblock {\em Preprint arXiv:2211.15628}, 2022.
	
	\bibitem{Jaeger}
	H.~Jaeger.
	\newblock The" echo state" approach to analysing and training recurrent neural
	networks-with an erratum note'.
	\newblock {\em German National Research Center for Information Technology GMD
		Technical Report}, 148, 2001.
	
	\bibitem{JL_Lemma}
	W.~B. Johnson.
	\newblock Extensions of lipschitz mappings into hilbert space.
	\newblock {\em Contemporary mathematics}, 26:189--206, 1984.
	
	\bibitem{Kalsi}
	J.~Kalsi, T.~Lyons, and I.~P. Arribas.
	\newblock Optimal execution with rough path signatures.
	\newblock {\em SIAM J. Financial Math.}, 11(2):470--493, 2020.
	
	\bibitem{Karatzas07}
	I.~Karatzas and C.~Kardaras.
	\newblock The num{\'e}raire portfolio in semimartingale financial models.
	\newblock {\em Finance Stoch.}, 11(4):447--493, 2007.
	
	\bibitem{Karatzas_Kim_20}
	I.~Karatzas and D.~Kim.
	\newblock Trading strategies generated pathwise by functions of market weights.
	\newblock {\em Finance Stoch.}, 24(2):423--463, 2019.
	
	\bibitem{Karatzas_Ruf}
	I.~Karatzas and J.~Ruf.
	\newblock Trading strategies generated by lyapunov functions.
	\newblock {\em Finance Stoch.}, 21(3):753--787, 2017.
	
	\bibitem{Kardaras12}
	C.~Kardaras.
	\newblock Market viability via absence of arbitrage of the first kind.
	\newblock {\em Finance and Stoch.}, 16(4):651--667, 2012.
	
	\bibitem{kardaras2021ergodic}
	C.~Kardaras and S.~Robertson.
	\newblock Ergodic robust maximization of asymptotic growth.
	\newblock {\em Ann. Appl. Probab.}, 31(4):1787--1819, 2021.
	
	\bibitem{Kim_22}
	D.~Kim.
	\newblock Market-to-book ratio in stochastic portfolio theory.
	\newblock {\em Finance Stoch.}, 27(2):401--434, 2023.
	
	\bibitem{kiraly19}
	F.~J. Kiraly and H.~Oberhauser.
	\newblock Kernels for sequentially ordered data.
	\newblock {\em J. Mach. Learn. Res.}, 20(31):1--45, 2019.
	
	\bibitem{levin13}
	D.~{Levin}, T.~{Lyons}, and H.~{Ni}.
	\newblock Learning from the past, predicting the statistics for the future,
	learning an evolving system.
	\newblock {\em Preprint arXiv:1309.0260}, 2013.
	
	\bibitem{Oberhauser_2014}
	C.~Litterer and H.~Oberhauser.
	\newblock On a chen–fliess approximation for diffusion functionals.
	\newblock {\em Monatshefte f{\"u}r Mathematik}, 175:577--593, 2011.
	
	\bibitem{Lyons_2007}
	T.~Lyons, M.~Caruana, and T.~Lévy.
	\newblock {\em Differential Equations Driven by Rough Paths}.
	\newblock Springer, 2007.
	
	\bibitem{Lyons2019}
	T.~Lyons, S.~Nejad, and I.~P. Arribas.
	\newblock Non-parametric pricing and hedging of exotic derivatives.
	\newblock {\em Appl. Math. Finance}, 27:457 -- 494, 2019.
	
	\bibitem{Lyons1998}
	T.~J. Lyons.
	\newblock Differential equations driven by rough signals.
	\newblock {\em Revista Matemática Iberoamericana}, 14(2):215--310, 1998.
	
	\bibitem{markovitz1959portfolio}
	H.~M. Markovitz.
	\newblock {\em Portfolio selection: Efficient diversification of investments}.
	\newblock John Wiley, 1959.
	
	\bibitem{JL_variant}
	J.~Matoušek.
	\newblock On variants of the johnson–lindenstrauss lemma.
	\newblock {\em Random Structures Algorithms}, 33(2):142--156, 2008.
	
	\bibitem{cirone2023neural}
	N.~Muca~Cirone, M.~Lemercier, and C.~Salvi.
	\newblock Neural signature kernels as infinite-width-depth-limits of controlled
	{R}es{N}ets.
	\newblock {\em Proceedings of Machine Learning Research}, 202:25358--25425,
	2023.
	
	\bibitem{Nakajima_Exp}
	M.~Negoro, K.~Mitarai, K.~Fujii, K.~Nakajima, and M.~Kitagawa.
	\newblock Machine learning with controllable quantum dynamics of a nuclear spin
	ensemble in a solid.
	\newblock {\em Preprint arxiv:1806.10910}, 2018.
	
	\bibitem{Leverage}
	C.~Ni, Y.~Li, and P.~A. Forsyth.
	\newblock Neural network approach to portfolio optimization with leverage
	constraints: a case study on high inflation investment.
	\newblock {\em Quantitative Finance}, 24(6):753--777, 2023.
	
	\bibitem{NSSBWL:21}
	H.~Ni, L.~Szpruch, M.~Sabate-Vidales, B.~Xiao, M.~Wiese, and S.~Liao.
	\newblock {Sig-Wasserstein GANs for Time Series Generation}.
	\newblock {\em Proceedings of the Second ACM International Conference on AI in
		Finance}, 2022.
	
	\bibitem{Ruf_Xie_Lyapunov}
	J.~Ruf and K.~Xie.
	\newblock Generalised lyapunov functions and functionally generated trading
	strategies.
	\newblock {\em Applied Mathematical Finance}, 26(4):293--327, 2019.
	
	\bibitem{Ruf_Xie_19}
	J.~Ruf and K.~Xie.
	\newblock The impact of proportional transaction costs on systematically
	generated portfolios.
	\newblock {\em SIAM J. Financial Math.}, 11(3):881--896, 2020.
	
	\bibitem{Schied_18}
	A.~Schied, L.~Speiser, and I.~Voloshchenko.
	\newblock Model-free portfolio theory and its functional master formula.
	\newblock {\em SIAM J. Finan. Math.}, 9(3):1074--1101, 2018.
	
	\bibitem{DeepLearning_Sharpe}
	Z.~Zhang, S.~Zohren, and S.~Roberts.
	\newblock Deep learning for portfolio optimization.
	\newblock {\em The Journal of Financial Data Science}, 2020.
	
	\bibitem{zhou2000continuous}
	X.~Y. Zhou and D.~Li.
	\newblock Continuous-time mean-variance portfolio selection: A stochastic lq
	framework.
	\newblock {\em Appl. Math. Optim.}, 42:19--33, 2000.
	
\end{thebibliography}
\end{document}